\documentclass[12pt,superscriptaddress]{article}

\usepackage{complexity}
\usepackage{graphicx}
\usepackage{amsmath}
\usepackage{amssymb}
\usepackage{amsthm}

\usepackage[pdfpagelabels,pdftex,bookmarks,breaklinks]{hyperref}
\hypersetup{colorlinks, linkcolor=black, citecolor=black}

\newcommand{\be}{\begin{equation}}
\newcommand{\ee}{\end{equation}}
\newcommand{\ba}{\begin{array}}
\newcommand{\ea}{\end{array}}
\newcommand{\bea}{\begin{eqnarray}}
\newcommand{\eea}{\end{eqnarray}}

\newclass{\StoqMA}{StoqMA}
\newclass{\StoqLH}{StoqLH}
\newclass{\TIM}{TIM}
\newclass{\HCB}{HCB}
\newclass{\HCD}{HCD}

\newcommand{\calH}{{\cal H }}
\newcommand{\calL}{{\cal L }}
\newcommand{\calN}{{\cal N }}
\newcommand{\calB}{{\cal B }}
\newcommand{\calD}{{\cal D }}

\newcommand{\calG}{{\cal G }}

\newcommand{\calE}{{\cal E }}
\newcommand{\calC}{{\cal C }}
\newcommand{\calS}{{\cal S }}
\newcommand{\calT}{{\cal T }}

\newcommand{\ZZ}{\mathbb{Z}}
\newcommand{\CC}{\mathbb{C}}

\newcommand{\eff}{\mathrm{eff}}
\newcommand{\sm}{\mathrm{sim}}
\newcommand{\tgt}{\mathrm{target}}

\newcommand{\diag}{\mathrm{diag}}
\newcommand{\extra}{\mathrm{extra}}

\newcommand{\ring}{\mathrm{chain}}
\newcommand{\main}{\mathrm{main}}

\newcommand{\all}{\mathrm{all}}
\newcommand{\la}{\langle}

\newcommand{\image}{\mathrm{Im}}

\newtheorem{dfn}{Definition}

\newtheorem{lemma}{Lemma}
\newtheorem{fact}{Fact}
\newtheorem{corol}{Corollary}

\newtheorem{theorem}{Theorem}

\title{On complexity of the quantum Ising model }

\newcommand{\IBM}{IBM  T.J. Watson  Research Center, Yorktown Heights, NY 10598, USA}
\newcommand{\MSR}{Quantum Architectures and Computation Group, Microsoft Research, Redmond, WA 98052, USA}

\author{Sergey Bravyi\footnote{\IBM} \and  Matthew Hastings\footnote{\MSR}}

\date{}

\begin{document}
\maketitle

\begin{abstract}
We study  complexity of several problems related to the  Transverse field Ising Model (TIM).
First, we consider the problem of estimating the ground state energy known as the Local Hamiltonian Problem (LHP).
It is shown that the LHP for TIM on degree-$3$ graphs is equivalent modulo polynomial reductions to
the LHP for  general $k$-local `stoquastic' Hamiltonians with any constant $k\ge 2$.
This result implies that estimating the ground state energy of TIM
on degree-$3$ graphs  is a complete problem for the complexity
class  $\StoqMA$ --- an extension  of the classical class $\MA$.
As a corollary, we complete the complexity classification of $2$-local
Hamiltonians  with a fixed set of interactions proposed recently by Cubitt and Montanaro.
Secondly, we study quantum annealing algorithms
for finding ground states
of classical spin Hamiltonians associated with hard optimization problems.
We prove that the quantum annealing  with TIM Hamiltonians
is equivalent modulo polynomial reductions to  the quantum annealing 
with a certain subclass of $k$-local stoquastic Hamiltonians.
This subclass includes all Hamiltonians representable as a sum of a $k$-local diagonal Hamiltonian
and a $2$-local stoquastic Hamiltonian. 
\end{abstract}

\newpage

\tableofcontents

\newpage

\section{Introduction and summary of results}

Numerical simulation of quantum many-body systems is a notoriously hard  problem.
 A particularly strong form of hardness known as $\QMA$-completeness~\cite{KSV} has been recently  established
for many natural  problems in this category.  Among them is
the problem of estimating the ground state energy 
 for certain physically-motivated quantum models
such as Hamiltonians  with nearest-neighbor interactions
on the two-dimensional~\cite{OT08} and one-dimensional~\cite{Aharonov09,Hallgren13}  lattices,
the Hubbard model~\cite{Schuch09,Childs13}, and the Heisenberg model~\cite{Schuch09,CM13}.
In contrast, a broad class of Hamiltonians known as sign-free or {\em stoquastic}~\cite{BDOT08}
has been identified for which certain simulation tasks become more tractable. 
By definition, stoquastic Hamiltonians must have real matrix elements
with respect to some fixed basis and all off-diagonal matrix elements must be non-positive.
Ground states of stoquastic  Hamiltonians are known to have real non-negative amplitudes 
in the chosen basis. Thus, for many purposes, the ground state can be viewed 
as a classical probability distribution which often enables efficient
simulation by quantum Monte Carlo 
algorithms~\cite{suzuki1977,prokof1996exact,sandvik1991,Trivedi1989}.
%\cite{suzuki1977,prokof1996exact,sandvik1991,Trivedi1989,Buonaura1998,Wessel2004}.
A notable example of a model in this category is the transverse field Ising model (TIM).
It has a Hamiltonian 
\begin{equation}
\label{tim}
H=\sum_{1\le u\le n} h_u X_u +g_u Z_u + \sum_{1\le u<v\le n} g_{u,v} Z_u Z_v.
\end{equation}
Here $n$ denotes the number of qubits (spins), $h_u,g_u,g_{u,v}$ are real coefficients,
and $X_u,Z_u$ are the Pauli operators acting on a qubit $u$. 
Note that $H$ is a stoquastic Hamiltonian in the standard $Z$-basis iff $h_u\le 0$ for all $u$.
This can always be achieved  by conjugating $H$ with $Z_u$. 
It is known that the ground state energy and the free energy of the TIM 
can be approximated with an additive error $\epsilon$ in time $\poly(n,\epsilon^{-1})$ using
Monte Carlo algorithms~\cite{Bravyi14}  in the special case when the Ising interactions
are ferromagnetic, that is, $g_{u,v}\le 0$ for all $u,v$.
Another important special case is the TIM defined on the one-dimensional lattice with $g_u=0$. 
In this case  the Hamiltonian Eq.~(\ref{tim}) is exactly solvable by the Jordan-Wigner transformation
and its eigenvalues can be computed analytically~\cite{Pfeuty1970}.
%SBB: what about g_u \ne 0 ? 
The ground state and the thermal equilibrium properties of the TIM have been studied in 
many different contexts including quantum phase transitions~\cite{Sachdev2007},
quantum spin glasses~\cite{Kopec1989,Laumann2008cavity}
and quantum annealing algorithms~\cite{Farhi2000,Boixo2013quantum,Ronnow2014,Shin2014quantum}.
In the present paper we address two open questions related to the TIM.
First, we consider the  problem of estimating the ground state energy of the TIM
and fully characterize its hardness in terms of the known complexity classes. 
Secondly we study quantum annealing algorithms with TIM Hamiltonians and show
that such algorithms can efficiently simulate a much broader class
of quantum annealing algorithms associated with many  important classical optimization problems.

To state our main  results  let us define  two classes of stoquastic Hamiltonians. 
Let $\TIM(n,J)$ be the set of all $n$-qubit 
 transverse field Ising Hamiltonians defined in Eq.~(\ref{tim})
such that  the coefficients $h_u,g_u,g_{u,v}$ have  magnitude at most $J$
for all $u,v$. 
 A TIM Hamiltonian is said to have interactions of degree $d$
iff each qubit is coupled to at most $d$ other qubits with $ZZ$ interactions.
Such Hamiltonian can be embedded into a degree-$d$ graph  such that 
only nearest-neighbor qubits interact. 
We note that the terms of $H$ that are linear in $Z_u$ can be absorbed into the Ising 
interaction part by introducing one ancillary qubit $a$
and replacing $g_u Z_u$ by $g_u Z_u Z_a$ for each $u$. 
This transformation does not change the spectrum of $H$ except for
doubling the multiplicity of each eigenvalue, see Section~\ref{sect:models} for details.

Let $\StoqLH(n,J)$ be the set of  stoquastic  $2$-local Hamiltonians $H$
on $n$ qubits with the maximum interaction  strength $J$.
By definition,  $H\in \StoqLH(n,J)$ iff
\[
H=\sum_{1\le u<v\le n} H_{u,v}, 
\]
where  $H_{u,v}$ is a
hermitian operator acting on the qubits $u,v$ such that 
$\|H_{u,v}\|\le J$ and all off-diagonal matrix elements of $H_{u,v}$ in the standard basis are real and non-positive. 
One can choose different operators $H_{u,v}$ for each pair of qubits. 
We shall provide a more explicit characterization of $2$-local stoquastic Hamiltonians
in terms of their Pauli expansion  in Section~\ref{sect:stoqLH}, see Lemma~\ref{lemma:stoq}.

Our first theorem asserts that  any $2$-local stoquastic Hamiltonian
can appear as an effective low-energy theory emerging from 
the TIM on a degree-$3$ graph. 

\begin{theorem}
\label{thm:main}
Consider any Hamiltonian $H\in \StoqLH(n,J)$ 
and a precision parameter $\epsilon>0$. 
There exist $n'\le \poly(n)$, $J'\le \poly(n,J,\epsilon^{-1})$, 
and a Hamiltonian $H'\in \TIM(n',J')$ such that \\
(1)  The $i$-th smallest eigenvalues of $H$ and $H'$ differ at most by $\epsilon$
for all $1\le i\le 2^n$. \\
(2) One can compute $H'$  in time $\poly{(n)}$. \\
(3)  $H'$ has interactions of degree $3$.
\end{theorem}
Here the maximum degree of all polynomial  functions is some fixed constant that does not
depend on any parameters (although we expect this constant to be quite large). 
The theorem has  important implications for classifying  complexity of 
the Local Hamiltonian Problem (LHP)~\cite{KSV,KKR06}.
Recall that the LHP is a decision problem where one has to decide whether
the ground state energy $E_0$ of a given Hamiltonian $H$ acting on $n$ qubits is
sufficiently small, $E_0\le E_{yes}$,  or sufficiently large, $E_0\ge E_{no}$.
Here  $E_{yes}<E_{no}$ are some specified thresholds such that 
$E_{no}-E_{yes}\ge \poly(n^{-1})$. The Hamiltonian must 
be representable as a sum of hermitian operators acting on at most $k$ qubits 
each, where $k=O(1)$ is a small constant. Each $k$-qubit operator 
must have norm at most  $\poly(n)$. Such Hamiltonians are known as $k$-local. 
Theorem~\ref{thm:main} implies that the LHP for $2$-local stoquastic Hamiltonians
has the same complexity as the LHP for TIM. Indeed, consider an instance of the LHP
for some Hamiltonian
$H\in \StoqLH(n,J)$ where  $J\le \poly{(n)}$. Choose a precision $\epsilon=(E_{no}-E_{yes})/3$
and  let $H'$ be the TIM Hamiltonian constructed  in Theorem~\ref{thm:main}.
Note that $H'$ acts on $\poly(n)$ qubits and has the interaction strength $\poly(n)$. 
Let $E_0'$ be the ground state energy of $H'$. Then $E_0\le E_{yes}$ implies
$E_0'\le E_{yes}+\epsilon\equiv E_{yes}'$ and $E_0\ge E_{no}$ implies
$E_0'\ge E_{no}-\epsilon \equiv E_{no}'$. Since $E_{no}'-E_{yes}'=(E_{no}-E_{yes})/3 \ge \poly{(n^{-1})}$,
the LHP for a $2$-local stoquastic Hamiltonian has been reduced to the LHP for TIM.
The converse reduction is trivial since any TIM Hamiltonian can be made stoquastic
by a local change of basis.  Thus we obtain
\begin{corol}
\label{corol:1}
The LHP for $2$-local stoquastic Hamiltonians has the same complexity
as  the LHP for  TIM with interactions of degree $3$, modulo  polynomial reductions. 
\end{corol}
It is known that the LHPs for $2$-local and $k$-local stoquastic Hamiltonians have the same
complexity for any constant $k\ge 2$, modulo polynomial reductions~\cite{BDOT08}.
Thus estimating
the ground state energy of TIM on a degree-$3$ graph is  as hard as estimating the ground state
energy of a general $k$-local stoquastic Hamiltonian for $k=O(1)$. 
Furthermore, the LHP for $6$-local stoquastic Hamiltonians 
is known to be a complete problem for the complexity class $\StoqMA$~\cite{BDOT08,BBT06}.
This is an extension of the classical class $\MA$ where the verifier can accept quantum states
as a proof. To examine the proof the verifier is allowed to apply classical reversible gates in a coherent fashion
and, finally, measure some fixed qubit in the $X$-basis. The verifier accepts the proof
if the measurement outcome is `$+$'. Let $P_{acc}(x)$ be the  acceptance probability of the verifier
for a given problem instance $x$
maximized over all possible proofs. 
A decision problem belongs to $\StoqMA$ if there exist a polynomial-size verifier as above and
 threshold probabilities $P_{yes}\ge P_{no}+\poly(n^{-1})$
such that $P_{acc}(x)\ge P_{yes}$ for any yes-instance $x$ and $P_{acc}(x)\le P_{no}$ for any no-instance $x$.
Here $n$ is the length of the problem instance $x$, see~\cite{BBT06}
for a formal definition.  Combining these known results and Corollary~\ref{corol:1} we obtain
\begin{corol}
\label{corol:2}
The Local Hamiltonian Problem for TIM with interactions of degree $3$ is complete for the complexity class $\StoqMA$.
\end{corol}
\noindent

Finally, Theorem~\ref{thm:main}  completes the complexity classification of $2$-local Hamiltonians 
with a fixed set of interactions proposed recently by Cubitt and Montanaro~\cite{CM13}.
The problem studied in~\cite{CM13} is defined as follows. 
Let $\calS$ be a fixed set of two-qubit hermitian operators. 
Consider a special case of 
the $2$-local LHP such that 
Hamiltonians are required to have a form $H=\sum_a x_a V_a$, where
$x_a$ is a  real coefficient and $V_a$ is an operator from $\calS$ applied to some
pair of qubits. For brevity, let us call the above problem  $\calS$-LHP.
The main result of Ref.~\cite{CM13} is that depending on the choice of $\calS$, the problem $\calS$-LHP is  either complete
for one of the complexity classes $\NP$, $\QMA$, or can be solved
in polynomial time on a classical computer, or can be reduced in polynomial time to the LHP for TIM.
In addition, one can efficiently determine which case is realized for a given choice of $\calS$. 
Combining this result and Corollary~\ref{corol:2} one obtains 
\begin{corol}
\label{corol:3}
Let $\calS$ be any fixed set of two-qubit hermitian operators. Then depending on $\calS$,
the problem $\calS$-LHP is  either complete
for one of the complexity classes $\NP$, $\StoqMA$, $\QMA$, or can be solved
in polynomial time on a classical computer.
\end{corol}

We also prove an analogue of  Theorem~\ref{thm:main} which
gives new insights on the  power of   quantum annealing (QA) algorithms~\cite{Farhi2000,Farhi2001quantum}
with TIM Hamiltonians which received a significant attention 
recently~\cite{Boixo2013quantum,Ronnow2014,Shin2014quantum}.
Recall that quantum annealing (QA)~\cite{Farhi2000,Farhi2001quantum}
attempts to find a global minimum of a real-valued function $f(x_1,\ldots,x_n)$ that depends on $n$
binary variables by encoding $f$ into a diagonal problem Hamiltonian $H_P=\sum_x f(x)|x\rangle\langle x|$
acting on $n$ qubits. 
To find the ground state of $H_P$ one chooses an adiabatic path 
$H(\tau)=(1-\tau)H(0)+\tau H_P$, $0\le \tau\le 1$, where $H(0)$ is some  simple Hamiltonian  usually chosen as 
the transverse magnetic field, $H(0)=-\sum_{u=1}^n X_u$. 
Initializing the system in the ground state of $H(0)$ and 
traversing the adiabatic path slowly enough 
one can approximately prepare the ground state of $H_P$.
The running time of QA algorithms scales as $\poly(n,\delta^{-1})$, where $\delta$ is the minimum
spectral gap of $H(\tau)$, see~\cite{Farhi2000,Farhi2001quantum,JRS07}.
We focus on the special case of QA such that the objective function $f(x_1,\ldots,x_n)$ is
a sum of terms that  depend on at most $k$ variables each. Here $k=O(1)$ is some small constant. 
This includes well-known  optimization problems such as $k$-SAT,  MAX-$k$-SAT
and many  variations thereof.  We show that any quantum annealing algorithm as above
can be efficiently simulated by the quantum annealing with TIM Hamiltonians.
The simulation has a  slowdown at most $\poly(n,\delta^{-1})$.

Fix some integer $k\ge 2$. We will say that $H$  is a 
$(2,k)$-local stoquastic Hamiltonian  iff $H$
is a sum of a $2$-local stoquastic Hamiltonian and a $k$-local
diagonal Hamiltonian. Let $\StoqLH^*(n,J)$ be the set of all $(2,k)$-local stoquastic Hamiltonians 
on $n$-qubits with the maximum  interaction strength $J$. 
\begin{theorem}
\label{thm:QA}
Consider any Hamiltonian $H\in \StoqLH^*(n,J)$ with a non-degenerate ground state $|g\rangle$ and a spectral
gap $\delta$. 
There exist $n'\le \poly(n)$, $J'\le \poly(n,J,\delta^{-1})$, 
a Hamiltonian $H'\in \TIM(n',J')$, and an isometry 
 $\calE\, : \, (\CC^2)^{\otimes n} \to (\CC^2)^{\otimes n'}$ such that \\
(1) $H'$ has a non-degenerate ground state $|g'\rangle$ and a  spectral gap at least $\delta/3$.\\
(2) $\| |g'\rangle - \calE |g\rangle \| \le 1/100$. \\
(3) The isometry $\calE$ maps basis vectors to basis vectors.\\
(4) One can compute $H'$ and the action of $\calE,\calE^\dag$ on any basis vector  in time $poly(n)$. \\
\end{theorem}
Here the maximum degree of all polynomial functions depends only on the locality
parameter $k$. We note that one can replace the constant $1/100$ in condition~(2) by an arbitrary
precision parameter $\eta>0$. Then the same theorem holds with a scaling $J'\le \poly(n,J,\delta^{-1},\eta^{-1})$. 
One can also impose a restriction that the Hamiltonian $H'$ has interactions of degree-$3$.
Then a similar theorem holds, but the isometry $\calE$ has slightly more complicated 
properties, see Section~\ref{sect:proof} for details. 

Let us discuss implications of the theorem.
Suppose $H(\tau)\in  \StoqLH^*(n,J)$ is an adiabatic path such that
$H(1)=H_P$ is the problem Hamiltonian and $H(0)=-\sum_{u=1}^n X_u$. 
We assume that $H(\tau)$ has a non-degenerate ground state $|g(\tau)\rangle$
and a spectral gap at least $\delta$ for all $\tau$. 
Also we assume that 
$J\le \poly(n)$. Since $X_u$ can be adiabatically rotated to  $Z_u$
without closing the gap, we can modify the path such that 
$H(0)=-\sum_{u=1}^n Z_u$. Then the initial ground state is $|g(0)\rangle=|0^{\otimes n}\rangle$.
Applying Theorem~\ref{thm:QA} to each Hamiltonian $H(\tau)$
one obtains a family of TIM Hamiltonians $H'(\tau)$ such that $H'(\tau)$ has a non-degenerate ground state $|g'(\tau)\rangle\approx \calE |g(\tau)\rangle$,
the spectral gap at least $\delta/3$, and the  interaction strength at most $\poly(n,\delta^{-1})$. 
We will show that the map $H\to H'$ is sufficiently smooth, so that the family
$H'(\tau)$, $0\le \tau\le 1$, defines an adiabatic path
and  the time it takes
to traverse the paths $H(\tau)$ and $H'(\tau)$ differ at most by a factor $\poly(n,\delta^{-1})$,
see Section~\ref{sect:proof} for details.  
Therefore one can (approximately) prepare the final state $|g'(1)\rangle$ by 
initializing the system in the basis state $\calE |0^{\otimes n}\rangle\approx |g'(0)\rangle$ and 
traversing the path $H'(\tau)$.
Measuring every qubit of the final state $|g'(1)\rangle$ in the $Z$-basis one obtains a  string
of outcomes $x\in \{0,1\}^{n'}$ such that $\calE |g(1)\rangle\approx |x\rangle$. 
Then $|g(1)\rangle\approx \calE^\dag |x\rangle$, that is, 
the ground state of $H_P$ can be efficiently computed from  $x$. 
Thus we obtain 
\begin{corol}
\label{cor:4}
Any quantum annealing  algorithm with $(2,k)$-local stoquastic Hamiltonians
can be simulated by a quantum annealing  algorithm with TIM Hamiltonians.
The simulation has overhead at most $\poly(n,\delta^{-1})$, where $n$ is the number of qubits
and $\delta$ is the minimum spectral gap of the adiabatic path. 
\end{corol}
In the rest of this section we informally sketch the proof of the main theorems,
discuss several open problems, and outline organization of the paper. 

{\bf Sketch of the proof.}
The proof of Theorems~\ref{thm:main},\ref{thm:QA} relies on
perturbative reductions~\cite{KKR06,OT08} 
and the Schrieffer-Wolff transformation~\cite{SW66,BDLT08,BDL11}.
At each step of the proof we work with two  quantum models: 
a simulator Hamiltonian $H_\sm$ acting on some Hilbert space $\calH$ 
and a  target Hamiltonian $H_\tgt$ acting on a certain subspace\footnote{More precisely, we identify $H_\tgt$ with a Hamiltonian acting on the subspace $\calH_-$ using a suitable encoding.}
$\calH_-\subseteq \calH$.
We represent $\calH_-$ as the  low-energy subspace of a suitable Hamiltonian $H_0$
which has a large energy gap $\Delta$ for all eigenvectors orthogonal to $\calH_-$. 
We choose $H_\sm=H_0+V$, where $V$ is a weak perturbation such that  $\|V\|\ll \Delta$. 
We show that $H_\tgt$ can be obtained from $H_\sm$ as an effective low-energy Hamiltonian
calculated using a few lowest orders of  the perturbation theory. 
More precisely,  $H_\tgt\approx P_- U (H_0+V) U^\dag P_-$, where
$P_-$ is the projector onto $\calH_-$ and $U$ is a unitary operator on $\calH$
known as the Schrieffer-Wolff transformation. The latter brings $H_0+V$ into a block-diagonal form
such that $U (H_0+V) U^\dag$ preserves  the subspace $\calH_-$. 
We show that the low-lying eigenvalues and eigenvectors of $H_\sm$ approximate the respective eigenvalues 
and eigenvectors of $H_\tgt$
with an error that can made arbitrarily small by choosing large enough $\Delta$.
\begin{table}[h]
\begin{center}
\begin{tabular}{|c|}
\hline
TIM, degree-$3$ graph \\
\hline
 TIM, general graph\\
\hline
 Hard-core dimers,  triangle-free graph \\
\hline
 Hard-core bosons, range-$2$  \\
\hline
Hard-core bosons, range-$1$  \\
\hline
 Hard-core bosons, range-$1$, controlled hopping  \\
\hline
$2$-local stoquastic Hamiltonians  \\
\hline
\end{tabular}
\end{center}
\caption{Perturbative reductions used in the proof of Theorem~\ref{thm:main}. Each model  is obtained as an
effective low-energy Hamiltonian for the model located one row above. 
The hard-core bosons (HCB) model describes a multi-particle quantum walk on a graph.
The Hamiltonian consists of a hopping term, on-site chemical potential, and arbitrary two-particle interactions.
Different particles must be separated from each other by a certain minimum distance
that we call a range of the model. HCB is closely related to the Bose-Hubbard model. 
 The hard-core dimers model is analogous to HCB except that admissible particle configuration
must consist of nearest-neighbor pairs of particles that we call dimers. 
Different dimers must be separated from each other by a certain minimum distance. 
A rigorous definition of the models is given in Section~\ref{sect:models}.
\label{tab:models}}
\end{table}
We apply the above step recursively several times such that 
the target Hamiltonian at the $t$-th step becomes the simulator Hamiltonian at the
$(t+1)$-th step.  The recursion starts from the TIM with interactions of degree-$3$ at the highest energy scale,
goes through several intermediate models listed in Table~\ref{tab:models}, and  arrives at 
a given   $2$-local or $(2,k)$-stoquastic Hamiltonian at the lowest energy scale.
Overall, the proof  requires nine different reductions\footnote{Some of our reductions
are `trivial' in the sense that they simply restrict a Hamiltonian to a certain subspace. 
The proof contains only six `non-trivial' reductions that actually change the Hamiltonian.}.
To simplify the analysis of recursive reductions we introduce a 
general  definition of a simulation that quantifies how close are two different models in terms
of their low-lying eigenvalues and eigenvectors. 
Our definition is shown to be stable under the composition of simulations.

For almost all of our reductions the  Hamiltonian $H_0$ is diagonal in the standard basis,
so that all eigenvalues and eigenvectors of $H_0$ can be easily computed. 
The only exception is the reduction from TIM with interactions of degree-$3$ to a general TIM.
For this reduction we encode each qubit of the target model into the approximately two-fold degenerate
ground subspace of the one-dimensional TIM on a chain of a suitable length.
Accordingly, the Hamiltonian $H_0$ describes a collection of one-dimensional TIMs. 
We simulate the logical Ising interaction $Z_uZ_v$ between some pair of logical qubits $u,v$ 
by applying the physical interaction
$Z_i Z_j$ to a properly chosen pair of qubits $i\in L_u$ and $j\in L_v$, where $L_u$ is the chain
encoding a logical qubit $u$.  The logical transverse field $X_u$ is automatically generated due to the energy splitting
between the ground states of $L_u$. 
The analysis of this reduction  exploits recent exact results on the form-factors of the one-dimensional TIM~\cite{formfactor}.

We emphasize that the word ``reduction" is used in two distinct senses. In the present paper  we speak of a perturbative reduction {\it from} a Hamiltonian $H_\sm$ {\it to} a Hamiltonian $H_\tgt$ when $H_\tgt$ is the effective low-energy Hamiltonian
derived from $H_\sm$, following terminology in physics.  However, if $H_\tgt$ belongs to some particular class 
of Hamiltonians $\calT$ and $H_\sm$ belongs to  some subclass $\calS\subseteq \calT$, this is a reduction {\it from} the class $\calT$ {\it to} the class $\calS$, according to  terminology in computer science. 

{\bf Open problems.} Our work raises several questions. 
First, we expect that Theorems~\ref{thm:main},\ref{thm:QA} can be extended in a number of ways.
For example,  one may ask whether the analogue of Theorem~\ref{thm:main}
holds for TIM Hamiltonians restricted to particular families of graphs,
such as planar graphs or regular lattices. We note that a simple modification of our degree reduction method
based on the one-dimensional TIM produces a simulator Hamiltonian which can be
embedded into   the 3D lattice of dimensions $n\times n \times 2$ with periodic boundary conditions.
 We expect that applying additional 
perturbative reductions such as those described in Ref.~\cite{OT08}
can further simplify the lattice.  Likewise, we expect that Theorem~\ref{thm:QA} can be extended
to the case when $H$ is a general $k$-local stoquastic Hamiltonian by applying perturbative reductions
of Ref.~\cite{BDOT08}. 

A challenging open question is whether  TIM  Hamiltonians defined on a 2D lattice 
can realize the topological quantum order.
 It has been recently shown that  the hard-core bosons 
model defined on the kagome lattice has a topologically ordered ground state
for a certain range of parameters~\cite{Isakov11,Isakov12}. A preliminary analysis shows that 
the chain of reductions from TIM to hard-core bosons described  in the present paper can
be modified such that all intermediate Hamiltonians have geometrically local interactions.
Assuming that the unphysical polynomial scaling of interactions in the simulator
Hamiltonian can be avoided~\cite{BDLT08,Cao2014perturbative}, this points towards existence of topologically ordered
phases described by TIM Hamiltonians.

Finally,  a big open question is whether QA algorithms with TIM  Hamiltonians
can be efficiently simulated classically. It has been recently shown that the general purpose quantum
Monte Carlo algorithms fail to simulate certain  instances  of the QA with TIM  efficiently~\cite{Hastings13},
even though these instances have a non-negligible minimum spectral gap. 
This leaves a possibility that  some more specialized algorithms taking advantage of the special structure of TIM Hamiltonians
can succeed even though the general purpose algorithm fail.  
 Our results demonstrate that this is unlikely, since simulating the QA
 with TIM is as hard as simulating the QA with much more general $(2,k)$-local stoquastic Hamiltonians.

The paper is organized as follows. 
Section~\ref{sect:models} contains a rigorous definition of the  models
listed in Table~\ref{tab:models}. Our main technical tools are introduced in Sections~\ref{sect:sim},\ref{sect:PT}
which present a general definition of a simulation, 
describe perturbative reductions based on  the Schrieffer-Wolff transformation,  and prove
several   technical lemmas used in the rest of the paper.
Section~\ref{sect:TIM3} shows how to simulate a general TIM Hamiltonian using
a special case of TIM with interactions of degree-$3$. 
Sections~\ref{sect:TIM2HCD}-\ref{sect:stoqLH} describe a chain of perturbative reductions 
between the models listed in Table~\ref{tab:models}.
These reductions are combined together in Section~\ref{sect:proof} which
contains the proof of Theorems~\ref{thm:main},\ref{thm:QA}.
Finally, Appendix~A proves certain  bounds on  eigenvalues and form-factors
of the one-dimensional TIM which are used in Section~\ref{sect:TIM3}.

\section{Hard-core bosons and dimers}
\label{sect:models}

Consider a graph $G=(U,E)$ with a set of $n$ nodes
$U$ and a set of edges $E$.  Define a Hilbert space  $\calB\cong (\CC^2)^{\otimes n}$ 
with an orthonormal basis  $\{|S\rangle\, : \, S\subseteq U\}$ such that basis
vectors are  labeled by subsets of nodes $S$.
We shall  identify subsets of nodes with configurations of particles that live at nodes of the graph.
Each node can be either empty or occupied by a single particle. 
For any node $u\in U$ define a  particle number operator  $n_u$ 
such that $n_u |S\rangle= |S\rangle$ if $u\in S$ and $n_u|S\rangle=0$ otherwise. 
We shall often consider diagonal Hamiltonians of the following form:
\begin{equation}
\label{Hdiag}
H_\diag= \sum_{u\in U} \mu_u n_u  + \sum_{\{u,v\}\subseteq  U} \omega_{u,v} n_u n_v.
\end{equation}
Here the second sum runs over all two-node subsets   (not only nearest neighbors).
The coefficients $\mu_u$ and $\omega_{u,v}$ can be viewed as a
chemical potential and a two-particle interaction potential respectively. 

Let us now define a hopping operator $W_{u,v}$. Here $u,v\in U$
are arbitrary nodes such that $u\ne v$.
By definition, $W_{u,v}$ annihilates any state $|S\rangle$  in which both nodes $u,v$ are occupied or
both nodes are empty.  If one of the nodes $u,v$ is occupied and the other node is empty,
$W_{u,v}$ transfers a particle from $u$ to $v$ or vice verse. 
Matrix elements of $W_{u,v}$ in the chosen basis are 
\begin{equation}
\label{hop}
\langle S'|W_{u,v}|S\rangle=\left\{ \ba{rcl}
1 & \mbox{if} & \mbox{$u\in S$, $v\notin S$, and $S'=(S\setminus u)\cup v$}.\\ 
1 & \mbox{if} & \mbox{$v\in S$, $u\notin S$, and $S'=(S\setminus v)\cup u$}, \\
0 && \mbox{otherwise}. \\
\ea
\right.
\end{equation}
Let $m,r\ge 1$ be fixed integers. Define a subspace  $\calB_m\subset \calB$ 
spanned by all subsets $S\subseteq U$ with exactly  $m$ nodes.
We shall refer to $\calB_m$ as an {\em $m$-particle sector}.
Obviously, the operators $W_{u,v}$ and $n_u$ preserve $\calB_m$.
A subset of nodes $S$ is said to be {\em $r$-sparse} 
iff the graph distance between any distinct pair of nodes $u,v\in S$ is at least $r$.
Define a subspace $\calB_{m,r}\subseteq \calB_m$ 
spanned by all $r$-sparse subsets $S\subseteq U$ with exactly $m$ nodes. 
By definition, any subset of nodes is $1$-sparse, so that $\calB_{m,1}= \calB_m$.
Note that the operators $W_{u,v}$ generally do not preserve $\calB_{m,r}$.
Below we consider  hopping operators $W_{u,v}$ projected onto the subspace $\calB_{m,r}$.
Matrix elements of a projected hopping operator are defined by 
Eq.~(\ref{hop}), where $S$ and $S'$ run over all $r$-sparse subsets of $m$ nodes.

Our first model is called  {\em hard-core bosons} (HCB). It is defined on the
Hilbert space $\calB_{m,r}$, where $m$ and $r$ are fixed  parameters.
We shall refer to  $r$ as the {\em range} of the model.
The Hamiltonian is 
\begin{equation}
\label{HCB}
H=-\sum_{(u,v)\in E} t_{u,v} W_{u,v}+ H_\diag.
\end{equation} 
Here $H_\diag$ is defined by Eq.~(\ref{Hdiag}) and all 
operators  are projected onto the subspace $\calB_{m,r}$. 
 Thus $W_{u,v}$ moves a particle  only if this does not
 violate the $r$-sparsity condition. Otherwise $W_{u,v}$ annihilates a state. 
 The coefficients $t_{u,v}$ are hopping amplitudes. 
We shall always assume that
\[
t_{u,v}\ge 0
\]
for all $u,v$. The coefficients  $\mu_u$ and $\omega_{u,v}$ in $H_\diag$ may have arbitrary signs. 
Note that $H$ is a stoquastic Hamiltonian. 
Let $\HCB_r(n,m,J)$ be the set of  Hamiltonians 
describing the $m$-particle sector of range-$r$ hard-core bosons on a graph 
with $n$ nodes   such that all the coefficients $\mu_u,\omega_{u,v},t_{u,v}$   have magnitude at most $J$.
Here we take the union over all graphs $G$ with $n$ nodes. 
Our proof will only use HCB models with the range $r=1,2$.
Later on we shall define certain enhanced versions
of the HCB which have multi-particle interactions, see Section~\ref{sect:multi},
and/or  controlled hopping terms, see Section~\ref{sect:HCB*}.
We note that the HCB model with non-positive hopping amplitudes $t_{u,v}\le 0$ has been recently studied 
by Childs, Gosset, and Webb~\cite{Childs13} who showed that the corresponding
LHP is $\QMA$-complete.

Our second model is called {\em hard-core dimers}.
This model also depends on a graph $G=(U,E)$. 
We shall only consider triangle-free graphs $G$. 
Let $m\ge 1$ be a fixed integer parameter. 
A subset of nodes $S\subseteq U$ is said to be a {\em dimer} iff
$S=\{u,v\}$ for some pair of nodes $u\ne v$  such that $(u,v)\in E$. 
Define an {\em $m$-dimer} as a subset of nodes  $S\subseteq U$ 
that can be represented as a disjoint union of $m$ dimers $S_1,\ldots,S_m$
such that the graph distance between $S_i$ and $S_j$ is at least 
three for all $i\ne j$. This particular choice of the distance 
guarantees that $m$-dimers can be represented as ground states
of a suitable Ising Hamiltonian, see Lemma~\ref{lemma:dimers}
in Section~\ref{sect:TIM2HCD}.
Examples of $2$-dimers are shown on Fig.~\ref{fig:2dimer}.

\begin{figure}[h]
\centerline{\includegraphics[height=3cm]{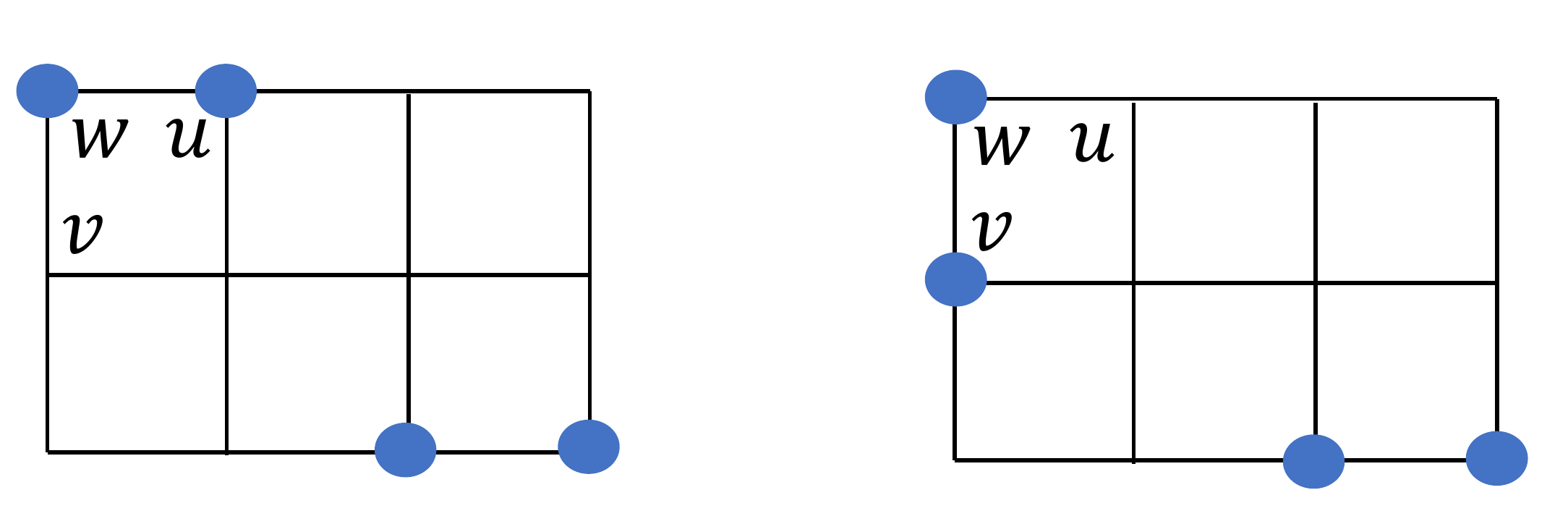}}
\caption{Examples of $2$-dimers  $S,S'$ on the square grid 
such that $|S'\rangle=W_{u,v}|S\rangle$. \label{fig:2dimer}
}
\end{figure}

Let $\calD_m\subseteq \calB_{2m}$ be the subspace spanned by all
basis vectors $|S\rangle$ such that $S\subseteq U$ is an $m$-dimer.
Note that the operators $W_{u,v}$ generally do not preserve $\calD_m$.
Below we consider  hopping operators $W_{u,v}$ projected onto the subspace $\calD_m$.
Matrix elements of a projected hopping operator are defined by 
Eq.~(\ref{hop}), where $S$ and $S'$ run over all $m$-dimers.
The hard-core dimers (HCD) model has a Hilbert space $\calD_m$ and a Hamiltonian
\begin{equation}
\label{HCD}
H=-t\sum_{\{u,v\}\subseteq U}  W_{u,v}
+H_{\diag}
\end{equation} 
where $H_{\diag}$ is defined by Eq.~(\ref{Hdiag}) and all operators are projected onto the subspace $\calD_m$.
The sum in Eq.~(\ref{HCD}) runs over all  pairs of nodes (not only nearest neighbors).
Although the Hamiltonian  does not explicitly depend on the graph structure,
the underlying Hilbert space $\calD_m$  does depends
on the graph since the latter  determines which subsets of nodes are $m$-dimers.  
A hopping process induced by $W_{u,v}$ can change a dimer $\{w,u\}$ to some other dimer $\{w,v\}$
with $u\ne v$, see  Fig.~\ref{fig:2dimer}. The coefficient $t$ is a hopping amplitude. We shall assume that  $t\ge 0$.
Then $H$ is a stoquastic Hamiltonian. Let $\HCD(n,m,J)$ be the set of Hamiltonians 
describing  the $m$-dimer sector of hard-core dimers model 
on a graph 
with $n$ nodes   such that
all coefficients in $H$ have magnitude at most $J$.
Here we take the union over all triangle-free graphs with $n$ nodes. 

Some perturbative reductions described  below will alter the underlying graph $G$. Whenever the choice of $G$
is not clear from the context, we shall use more detailed notations $\calB_m(G)$, $\calB_{m,r}(G)$, and $\calD_m(G)$
instead of $\calB_m$, $\calB_{m,r}$, and $\calD_m$.
Our notations for various classes of Hamiltonians are summarized in Table~2. 
\begin{table}[h]
\begin{center}
\begin{tabular}{c|c}
\hline
$\TIM(n,J)$ & Transverse field Ising Model \\
\hline
$\HCD(n,m,J)$ & $m$-dimer sector of Hard-Core Dimers model.\\
\hline
$\HCB_r(n,m,J)$ &  $m$-particle sector of Hard-Core Bosons with range $r$.\\
\hline
$\HCB(n,m,J)$ & same as $\HCB_1(n,m,J)$. \\
\hline
$\HCB^*(n,m,J)$ &  $\HCB(n,m,J)$ with controlled hopping terms. \\
\hline
$\StoqLH(n,J)$ &  Stoquastic  $2$-Local Hamiltonians. \\
\hline
\end{tabular}
\caption{Classes of Hamiltonians used in the proof of Theorem~\ref{thm:main}.
 Here $J$ denotes the maximum interaction strength
and $n$ denotes the number of nodes in the graph (the number of qubits).
 For those models that depend on a graph,
the corresponding class is defined by taking the union over all graphs with a fixed number
of nodes $n$ (all triangle-free graphs in the case of $\HCD$).
The class $\HCB^*$ is formally defined in Section~\ref{sect:HCB*}.} 
\label{tab:classes}
\end{center}
\end{table}

Finally, consider  a TIM Hamiltonian $H\in \TIM(n,J)$  defined in Eq.~(\ref{tim}).
Let us add an ancillary qubit labeled by `$a$' and consider a modified Hamiltonian
$H'\in \TIM(n+1,J)$ defined as 
\begin{equation}
\label{TIMext1}
H'=\sum_{u=1}^n h_u X_u +g_u Z_u Z_a + \sum_{1\le u<v\le n} g_{u,v} Z_u Z_v.
\end{equation}
Let $X_\all=X^{\otimes (n+1)}$ be the global spin flip operator.
Note that  $H$ commutes with $Z_a$ and $X_\all$, whereas $Z_a$
and $X_\all$ anti-commute. This implies that 
the restriction of $H'$ onto the sectors $Z_a=1$ and $Z_a=-1$ have exactly the same
spectrum as the original Hamiltonian $H$. Hence the full spectrum of $H'$
is obtained from the one of $H$ by doubling the multiplicity of each eigenvalue. 
In particular,  the  LHPs for Hamiltonians Eqs.~(\ref{tim},\ref{TIMext1}) have the same complexity.
Finally, 
substituting $Z_u=I-2n_u$ into Eq.~(\ref{tim})  one gets
 \begin{equation}
\label{TIMdef}
H= \sum_{u\in U} \mu_u n_u  + \sum_{\{u,v\} \subseteq  U} \omega_{u,v} n_u n_v  + h_u X_u
=H_\diag + \sum_{u\in U} h_u X_u,
\end{equation}
where $U\equiv \{1,\ldots,n\}$, $\omega_{u,v}=4g_{u,v}$ and $\mu_u=-2g_u - 2\sum_{v\ne u} g_{u,v}$.
Here we ignore the overall energy shift. 
Clearly, the coefficients $\omega_{u,v}$ and $\mu_u$ have magnitude at most $O(nJ)$.
Below we shall work with TIM Hamiltonians as defined in Eq.~(\ref{TIMdef}).

\section{Simulation of eigenvalues and eigenvectors}
\label{sect:sim}

In this section we give a formal definition of a simulation.
It  quantifies how close are two different models in terms of their
low-energy properties such as the low-lying eigenvalues and  
eigenvectors. 
We consider  a {\em target} model described
by a Hamiltonian $H$ acting on some $N$-dimensional Hilbert space $\calH$
and a {\em simulator} model described by a Hamiltonian $H_\sm$
acting on some Hilbert space $\calH_\sm$ of dimension at least $N$. 
Our definition of a simulation depends on a particular {\em encoding transformation} $\calE\, : \, \calH\to \calH_\sm$
that  embeds $\calH$ into some $N$-dimensional subspace  of $\calH_\sm$.
We assume that $\calE$ is an isometry, that is, $\calE^\dag \calE=I$. 
The encoding enables a  comparison between  eigenvectors of the two models.
We envision a situation when the spectrum of $H_\sm$ consists of two well-separated
groups of eigenvalues such that the $N$ smallest eigenvalues of $H_\sm$
are separated from the rest of its spectrum by a large gap. 
Let $\calL_N(H_\sm) \subseteq \calH_\sm$ be  the {\em low-energy subspace}
spanned by  the eigenvectors of $H_\sm$ associated with its $N$ smallest eigenvalues.
\begin{dfn}
\label{dfn:sim}
Let $H$ be a Hamiltonian acting on a Hilbert space $\calH$ of dimension $N$. 
A Hamiltonian $H_\sm$  and an isometry (encoding)  $\calE\, : \, \calH\to \calH_\sm$ are said to simulate $H$ with 
an error $(\eta,\epsilon)$  if there exists an isometry $\tilde{\calE}\, : \, \calH \to \calH_\sm$ such that 
\begin{enumerate}
\item[S1.] The image of $\tilde{\calE}$ coincides with the low-energy subspace $\calL_N(H_\sm)$.
\item[S2.]  $\| H - \tilde{\calE}^\dag H_\sm \tilde{\calE} \|\le \epsilon$.
\item[S3.]  $\| \calE-\tilde{\calE}\|\le \eta$.
\end{enumerate}
\end{dfn}
Although we do not impose any restrictions on the encoding, in practice it must be sufficiently simple. 
For all our reductions (except for the one of Section~\ref{sect:TIM3}) the encoding $\calE$ maps basis vectors to basis vectors. 
Whenever  the choice of $\calE$ is clear from the context, we shall  just say that 
$H_\sm$ simulates $H$ with an error $(\eta,\epsilon)$. 
If one is interested only in reproducing eigenvalues of the target Hamiltonian,
the encoding  and condition~(S3) can be ignored. 

In the case of a zero error, $\epsilon=\eta=0$, the target Hamiltonian $H$ coincides with the
restriction of $H_\sm$ onto the low-energy subspace of $H_\sm$, up to a change of basis described by $\calE$.
Clearly, any Hamiltonian simulates itself with a zero error since one can choose $\calE=\tilde{\calE}=I$.
We shall always assume that $\epsilon\le \|H\|$ since otherwise the definition is meaningless
(one can choose $H_\sm=0$ regardless of $H$).
Note that $\epsilon$ has the dimension of energy while $\eta$ is dimensionless.  
Loosely speaking, $\epsilon$ and $\eta$ quantify simulation error for eigenvalues and eigenvectors
respectively. 
Let us establish some basic properties of  simulations. 
\begin{lemma}[\bf Eigenvalue simulation]
\label{lemma:sim1}
Suppose $(H_\sm, \calE)$ simulates $H$ with an error  $(\eta,\epsilon)$.
Then the $i$-th smallest eigenvalues of $H_\sm$ and $H$ differ at most by $\epsilon$
for all $1\le i\le N$. 
\end{lemma}
\begin{proof}
Property~(S1) implies that the spectrum of  $\tilde{\calE}^\dag H_\sm \tilde{\calE}$ coincides with $N$ smallest
eigenvalues of $H_\sm$. The lemma now follows from (S2) and the standard Weyl's inequality.
\end{proof}
\begin{lemma}[\bf Ground state simulation]
\label{lemma:sim2}
Suppose $H$ has a non-degenerate ground state $|g\rangle$ separated from excited states
by a spectral gap $\delta$. Suppose $(H_\sm, \calE)$ simulates $H$ with an
error $(\eta,\epsilon)$ such that  $2\epsilon<\delta$. Then $H_\sm$ has a non-degenerate ground state
$|g_\sm\rangle$ and 
\begin{equation}
\label{sim2}
\| \calE |g\rangle - |g_\sm\rangle \| \le \eta+ O(\delta^{-1} \epsilon).
\end{equation}
\end{lemma}
\begin{proof}
Let $|g_\sm\rangle$ be the ground state of $H_\sm$. Note that $|g_\sm\rangle$ is non-degenerate
due to Lemma~\ref{lemma:sim1} and the assumption $2\epsilon<\delta$.
Consider an unperturbed Hamiltonian $H$ and a perturbation $V=\tilde{\calE}^\dag H_\sm \tilde{\calE} -H$.
The perturbed Hamiltonian $H+V=\tilde{\calE}^\dag H_\sm \tilde{\calE}$ has a non-degenerate
ground state $\tilde{\calE}^\dag|g_\sm\rangle$. Using the first-order perturbation theory for eigenvectors one gets
$\||g\rangle - \tilde{\calE}^\dag|g_\sm\rangle\|\le O(\delta^{-1} \epsilon)$
and thus $\|\tilde{\calE}|g\rangle - |g_\sm\rangle\| = \| \tilde{\calE}|g\rangle-\tilde{\calE}\tilde{\calE}^\dag |g_\sm\rangle\| \le O(\delta^{-1}\epsilon)$. Here we used the fact that $\tilde{\calE}$ is an isometry and
 $|g_\sm\rangle\in \calL_N(H_\sm)=\image{(\tilde{\calE})}$.
Property~(S3) then  leads to Eq.~(\ref{sim2}).
\end{proof}
Importantly,  our definition of a simulation is stable under 
compositions:   if one is given some Hamiltonians $H,H_1,H_2$ such that $H_1$ simulates $H$ with a small
error and $H_2$ simulates $H_1$ with a small error, this implies that $H_2$ simulates $H$ with a small error. 
\begin{lemma}[\bf Composition]
\label{lemma:sim3}
Suppose $(H_1, \calE_1)$ simulates $H$ with an error 
$(\eta_1,\epsilon_1)$ and   $(H_2,\calE_2)$ simulates $H_1$
with an error  $(\eta_2,\epsilon_2)$.
Let  $\Delta_1$ be the spectral gap separating
$N$ smallest eigenvalues of $H_1$ from the rest of the spectrum. 
Suppose $2\epsilon_2<\Delta_1$ and $\epsilon_1,\epsilon_2\le \|H\|$.
Then  $(H_2,\calE_2\calE_1)$ simulates $H$ with an error
$(\eta,\epsilon)$, where 
\begin{equation}
\label{compo}
\eta=\eta_1+\eta_2+ O( \epsilon_2 \Delta_1^{-1})
\quad \mbox{and} \quad
\epsilon =\epsilon_1+\epsilon_2 + O(\epsilon_2 \Delta_1^{-1} \|H\|).
\end{equation}
\end{lemma}
\noindent
 We shall always choose the simulator such that $\Delta_1\gg \|H\|$ in which case $\epsilon\approx \epsilon_1+\epsilon_2$.
\begin{proof}
Suppose $H,H_1,H_2$ act on Hilbert spaces $\calH,\calH_1,\calH_2$ respectively.
Let $N=\dim{(\calH)}$ and $N_1=\dim{(\calH_1)}$.
By Lemma~\ref{lemma:sim1}, the $N$ smallest eigenvalues of $H_2$
are separated from the rest of the spectrum by a spectral gap
at least $\Delta_1-2\epsilon_2>0$. Thus the low-energy subspace
$\calL_N(H_2)$ is well defined. 
Let $\tilde{\calE}_2\, : \, \calH_1\to \calH_2$ be an isometry satisfying 
properties~(S1-S3) for a simulator $(H_2,\calE_2)$ and a target Hamiltonian $H_1$
with an error $(\eta_2,\epsilon_2)$. 
By definition, $\tilde{\calE}_2$ maps $\calH_1$ to the low-energy subspace $\calL_{N_1}(H_2)$.
 First, let us show that $\tilde{\calE}_2$ approximately maps $\calL_N(H_1)$ 
to $\calL_N(H_2)$. More precisely, we claim that there exists
a unitary operator $U\, : \, \calL_{N_1}(\calH_2)\to \calL_{N_1}(\calH_2)$ such that
\begin{equation}
\label{compo1}
\calL_N(H_2)=U \tilde{\calE}_2 \cdot \calL_N(H_1) \quad \mbox{and} \quad \|U-I\|\le 2\sqrt{2} \Delta_1^{-1}\epsilon_2.
\end{equation}
Indeed,  let $P_N(H_i)$ be the projector onto 
the low-energy subspace 
$\calL_N(H_i)$, where $i=1,2$. 
Consider  a pertubation
$V=\tilde{\calE}_2^\dag H_2 \tilde{\calE}_2 -H_1$.
Note that $\calL_N(H_1+V)=\tilde{\calE}_2^\dag  \cdot \calL_N(H_2)$. 
Applying Lemma~3.1 of Ref.~\cite{BDL11}
with an unperturbed Hamiltonian $H_1$ and a perturbation $V$
one gets
\begin{equation}
\label{eq}
\| P_N(H_1) - \tilde{\calE}_2^\dag  P_N(H_2) \tilde{\calE}_2 \|\le \frac{2\| \tilde{\calE}_2^\dag H_2 \tilde{\calE}_2 - H_1 \|}{\Delta_1}
\le \frac{2\epsilon_2}{\Delta_1}.
\end{equation}
Taking into account that $\calL_N(H_2)\subseteq \calL_{N_1}(H_2)=\image{(\tilde{\calE}_2)}$
one gets
\[
\| \tilde{\calE}_2 P_N(H_1) \tilde{\calE}_2^\dag - P_N(H_2)\|
= \| P_N(H_1) - \tilde{\calE}_2^\dag  P_N(H_2) \tilde{\calE}_2 \|
\]
and thus
\begin{equation}
\label{eq}
\| \tilde{\calE}_2 P_N(H_1) \tilde{\calE}_2^\dag - P_N(H_2)\| \le \frac{2\epsilon_2}{\Delta_1}.
\end{equation}
For brevity denote 
\[
P\equiv \tilde{\calE}_2 P_N(H_1) \tilde{\calE}_2^\dag
\quad \mbox{and} \quad 
Q\equiv P_N(H_2).
\]
By Jordan's lemma, there exists an orthonormal basis  such that
the projectors $P$ and $Q$
are block-diagonal in this basis with all blocks being either $1\times 1$ or $2\times 2$
projectors.  Assuming that $2\Delta_1^{-1} \epsilon_2<1$ one has
$\| P-Q\|<1$ which implies that 
all $1\times 1$ blocks of $P$ and $Q$   are the same. Consider some $2\times 2$ block.
Without loss of generality, the restrictions of $P$ and $Q$ onto this block have a form 
\[
P= \left[ \ba{cc} 1 & 0 \\ 0 & 0 \\ \ea \right]
 \quad \mbox{and} \quad  Q = \left[ \ba{cc} c^2 & cs \\ cs & s^2 \\ \ea \right]
 \]
for some $0\le c,s\le 1$ such that $c^2+s^2=1$.
Then $P-Q=s^2 Z-csX$ and thus $\|P-Q\|=s$. We conclude that 
$s\le 2\Delta_1^{-1} \epsilon_2$ for any $2\times 2$ block. 
Define a unitary 
\[
U=\left[ \ba{cc} c & -s \\ s & c\\ \ea \right] =cI-isY
\]
such that $UPU^\dag =Q$. Note that $\| U-I\|= |c-1+is| \le \sqrt{2} s$.
Extending $U$ to the full space $\calL_{N_1}(\calH_2)$ we obtain 
$\image{(Q)}=U \cdot \image(P)$ and $\|U-I\|\le \sqrt{2} s$ which is equivalent to Eq.~(\ref{compo1}).

Now we are ready to prove that $(H_2,\calE_2 \calE_1)$  simulates $H$ with a small error. 
Define an isometry
\[
\tilde{\calE}=U\tilde{\calE}_2\tilde{\calE}_1.
\]
Using the first part of Eq.~(\ref{compo1}) and the fact that $\tilde{\calE}_1$ maps
$\calH$ to the low-energy subspace $\calL_N(H_1)$ 
we conclude that  $\tilde{\calE}$ maps $\calH$ to the low-energy subspace $\calL_N(H_2)$. 
Thus $\tilde{\calE}$ obeys property~(S1) for the target Hamiltonian $H$
and the simulator $(H_2,\calE_2 \calE_1)$. Furthermore, the second part of Eq.~(\ref{compo1})
implies that
\[
\eta\equiv \| \tilde{\calE}-\calE_2 \calE_1\| \le \| U-I \| + \|\tilde{\calE}_2 -\calE_2\| + \| \tilde{\calE}_1 - \calE_1\| 
\le 2\sqrt{2} \Delta_1^{-1}\epsilon_2 + \eta_2+\eta_1.
\]
Finally, let $H_i^{(N)}=H_i P_N(H_i)$ be the restriction of $H_i$
onto the low-energy subspace $\calL_N(H_i)$. 
Note that $H_1\tilde{\calE}_1 = H_1^{(N)} \tilde{\calE}_1$ and  $H_2 U\tilde{\calE}_2\tilde{\calE}_1 = H_2^{(N)} U\tilde{\calE}_2\tilde{\calE}_1$.
 Thus
 \[
\| H - \tilde{\calE}^\dag H_2 \tilde{\calE}\|  \le   \| H-\tilde{\calE}_1^\dag H_1 \tilde{\calE}_1\| + 
\| \tilde{\calE}_1^\dag (H_1^{(N)}- \tilde{\calE}_2^\dag U^\dag H_2^{(N)} U \tilde{\calE}_2)\tilde{\calE}_1 \|.
\]
The first term is  upper bounded by $\epsilon_1$. Thus 
\[
\| H - \tilde{\calE}^\dag H_2 \tilde{\calE}\|  \le \epsilon_1+ \| U \tilde{\calE}_2 H_1^{(N)} -  H_2^{(N)} U \tilde{\calE}_2\|.
\]
To bound the  second term we write $U=I+M$ and note that 
\begin{eqnarray}
U \tilde{\calE}_2 H_1^{(N)} -  H_2^{(N)} U \tilde{\calE}_2 
&=& P_N(H_2)( U\tilde{\calE}_2 H_1 - H_2 U \tilde{\calE}_2 )P_N(H_1) \nonumber \\
&=& P_N(H_2) (\tilde{\calE}_2 H_1 - H_2 \tilde{\calE}_2 ) P_N(H_1) + P_N(H_2) M \tilde{\calE}_2 H_1^{(N)} 
- H_2^{(N)} M \tilde{\calE}_2 P_N(H_1). \nonumber
\end{eqnarray}
The norm of the first term  is upper bounded by 
\[
\|\tilde{\calE}_2 H_1 - H_2 \tilde{\calE}_2\| \le \| H_1-\tilde{\calE}_2^\dag H_2 \tilde{\calE}_2\|\le \epsilon_2.
\]
Thus 
\[
\| H - \tilde{\calE}^\dag H_2 \tilde{\calE}\|  \le \epsilon_1+\epsilon_2+ \| U-I\|\cdot \| H_1^{(N)}\| 
+ \| U-I\| \cdot \| H_2^{(N)}\|.
\]
Lemma~\ref{lemma:sim1} implies that $\|H_2^{(N)}\|\le \|H_1^{(N)}\| +\epsilon_2$
and $\|H_1^{(N)}\|\le \|H\|+\epsilon_1$. Combining this and the second part of Eq.~(\ref{compo1})
one arrives at
\[
\| H - \tilde{\calE}^\dag H_2 \tilde{\calE}\|  \le \epsilon_1+\epsilon_2+ O( \Delta_1^{-1}\epsilon_2 \|H\|)+
O(\Delta_1^{-1} (\epsilon_2^2+\epsilon_1\epsilon_2)).
\]
Since we assumed that $\epsilon_1,\epsilon_2\le \|H\|$, the last term 
is at most $O(\Delta_1^{-1} \epsilon_2\|H\|)$.
\end{proof}

\section{Schrieffer-Wolff transformation and perturbative reductions}
\label{sect:PT}

Let $H$ be a target Hamiltonian chosen from some particular class of Hamiltonians $\calC$.
Suppose our goal is to simulate $H$ with a small error according to Definition~\ref{dfn:sim}
using a simulator Hamiltonian $H_\sm$ which is required  to be a member of some
smaller class $\calC'\subset \calC$. Perturbative reductions~\cite{KKR06,OT08} provide a general method
of accomplishing such simulation. 
Here we describe perturbative reductions based on the Schrieffer-Wolff transformation~\cite{SW66}, 
see for instance~\cite{BDL11}  and the references therein. 
Also we provide sufficient conditions under which a $k$-th order reduction
achieves the desired simulation error for $k=1,2,3$, see Lemmas~\ref{lemma:1st}-\ref{lemma:1st+}.

Consider a finite-dimensional Hilbert space $\calH_\sm$ decomposed into a direct sum
\begin{equation}
\label{H-+}
\calH_\sm=\calH_-\oplus \calH_+.
\end{equation}
Let $N_\pm=\dim{(\calH_{\pm})}$ and $P_\pm$ be the projector
onto $\calH_\pm$ such that $P_-+P_+=I$. 
Let $O$ be any linear operator on $\calH_\sm$. We shall write
\[
O_{--}=P_-OP_-, \quad O_{-+}=P_-OP_+, \quad O_{+-}=P_+ OP_-, \quad O_{++}=P_+ O P_+.
\]
The operator  is said to be block diagonal if $O_{-+}=0$ and $O_{+-}=0$.
The  operator  is said to be block off-diagonal if $O_{--}=0$ and $O_{++}=0$.

Let $H_0$ and $V$ be hermitian operators  on $\calH_\sm$ 
such that $H_0$ is block-diagonal, $(H_0)_{--}=0$, and such that $(H_0)_{++}$ has all eigenvalues greater or equal to one.  
Consider a perturbed Hamiltonian
\begin{equation}
\label{Hsim}
H_\sm=\Delta H_0+V, 
\end{equation}
where $\Delta$ is a large parameter. We shall always assume that 
\begin{equation}
\label{large_gap}
\|V\|< \Delta/2.
\end{equation}
The Schrieffer-Wolff transformation is a unitary operator on $\calH_\sm$
defined as $e^S$, where $S$ is an anti-hermitian operator   satisfying 
\begin{equation}
\label{SWconditions}
(e^S H_\sm e^{-S} )_{-+}=0, \quad (e^S H_\sm e^{-S} )_{+-}=0, \quad  S_{--}=0, \quad S_{++}=0,  \quad \|S\|<\pi/2.
\end{equation}
In other words, we require that  the transformed Hamiltonian $e^S H_\sm  e^{-S}$ is block diagonal
whereas $S$ itself is block off-diagonal. 
It is known that Eq.~(\ref{SWconditions}) has a unique solution  $S$,
see Lemma~2.3 and Lemma~3.1 in~\cite{BDL11}.
In particular, $S=0$ if $V=0$. 
The  effective low-energy Hamiltonian $H_{\eff}$
is a hermitian operator acting on $\calH_-$  defined as
\[
H_{\eff}=(e^S H_\sm e^{-S})_{--}.
\]
Note that $H_\eff=0$ if $V=0$. 
Since the operator $e^S$ is unitary and the transformed Hamiltonian $e^SH_\sm e^{-s}$
is block-diagonal, each eigenvalue of $H_\eff$ must be an eigenvalue of $H_\sm$. 
It is known that the $i$-th smallest eigenvalues of $H_\eff$ and $H_\sm$ coincide
for all $1\le i\le N_-$,  see~\cite{BDL11} for details.

Consider now   the Taylor series 
$S=\sum_{j=1}^\infty S_j$ and 
$H_{\eff}=\sum_{j=1}^\infty H_{\eff,j}$, where
$S_j$  and $H_{\eff,j}$ are the Taylor coefficients proportional to the $j$-th power  of $V$.
We shall only need the Taylor coefficients
\begin{equation}
\label{PT1}
H_{\eff,1}= V_{--},
\end{equation}
\begin{equation}
\label{PT2}
H_{\eff,2}= -\Delta^{-1} \,  V_{-+}H_0^{-1} V_{+-}
\end{equation}
and
\begin{equation}
\label{PT3}
H_{\eff,3}= \Delta^{-2} \, V_{-+} H_0^{-1} V_{++} H_0^{-1} V_{+-} -\frac{\Delta^{-2}}2\left( V_{-+} H_0^{-2} V_{+-} V_{--} + \mathrm{h.c.}\right),
\end{equation}
see Section~3.2 of~\cite{BDL11} for the derivation.
Note that the restriction of $H_0^{-1}$ to the subspace $\calH_+$ is well-defined
since all eigenvalues of $(H_0)_{++}$ are at least one. 
It is  known that the series for   $S$ and $H_{\eff}$ converge absolutely for $\|V\|<\Delta/16$
and the $j$-th Taylor coefficients are bounded as 
\begin{equation}
\label{Taylor}
 \|S_j\| \le (b\Delta)^{-j} \|V\|^j \quad \mbox{and} \quad \|H_{\eff,j}\|\le (c\Delta)^{1-j}\|V\|^j 
\end{equation}
for some constant coefficients $b,c>0$, see  Lemma~3.4 in~\cite{BDL11}. 
Define the $k$-th order effective Hamiltonian 
as the truncated  series
%\footnote{Although the letter $k$ has been used in the introduction 
%to denote the locality of Hamiltonians, from now on we reserve $k$ to denote  the order
%of perturbation theory. }
\begin{equation}
\label{Heff(k)}
H_{\eff}(k)=\sum_{j=1}^k H_{\eff,j}.
\end{equation}
The above implies that for any $k=O(1)$ and $\|V\|<\Delta\cdot \min{\{ b,c\}}$ one has
\begin{equation}
\label{Strunc}
\|S\| \le \sum_{j=1}^\infty (b\Delta)^{-j} \|V\|^j =O(\Delta^{-1}\|V\|)
\end{equation}
and
\begin{equation}
\label{Htrunc}
\| H_{\eff}-H_{\eff}(k)\| \le \sum_{j=k+1}^\infty (c\Delta)^{1-j}\|V\|^j = O(\Delta^{-k} \|V\|^{k+1}).
\end{equation}
Suppose now that $H_\tgt$ is a fixed  target Hamiltonian  acting on 
some Hibert space $\calH_\tgt$ and $\calE\, : \, \calH_\tgt\to \calH_\sm$
is some fixed isometry (encoding) such that $\image{(\calE)}=\calH_-$.
Define a  {\em logical target Hamiltonian}  acting on $\calH_-$  as
\begin{equation}
\label{target_logical}
\overline{H}_\tgt=\calE H_\tgt \calE^\dag.
\end{equation}
The goal of perturbative reductions is to approximate $\overline{H}_\tgt$ by the 
effective low-energy Hamiltonian $H_\eff(k)$  emerging from 
the simulator Hamiltonian  $H_\sm$ defined in Eq.~(\ref{Hsim}), where the parameter $\Delta$
controls the approximation error.  Below we outline a general strategy for constructing the simulator Hamiltonian
 proposed by Oliveira and Terhal~\cite{OT08}.
The strategy depends on the order $k$ of a reduction. 

For first-order reductions one just needs to choose $V$ such that 
$\overline{H}_\tgt=(V)_{--}$, see Eq.~(\ref{PT1}).
For second-order reductions the perturbation $V$ will be chosen as
\begin{equation}
\label{V2a}
V=\Delta^{1/2} \, V_\main+V_\extra,
\end{equation}
where $(V_\main)_{--}=0$ and $V_\extra$ is block-diagonal. Both operators $V_\main$ and $V_\extra$ are independent of $\Delta$. Substituting $V$ into Eqs.~(\ref{PT1},\ref{PT2}) gives
\begin{equation}
\label{V2b}
H_{\eff,1}=(V_{\extra})_{--} \quad \mbox{and}\quad H_{\eff,2}=  - (V_\main)_{-+} H_0^{-1} (V_\main)_{+-}.
\end{equation}
We shall choose $V_\main$ such that $H_{\eff,2}$ generates the desired logical target Hamiltonian
and, may be, some unwanted terms.
The purpose of $V_\extra$
is to cancel the unwanted terms.  
In addition, $V_\extra$ may include all the terms of the target Hamiltonian
that are members of the simulator class, such as  two-qubit diagonal  interactions. Note that the latter
 belong to all the classes listed in Table~\ref{tab:classes}.
Most of the  second-order reductions described below will achieve an exact equality $H_{\eff}(2)=\overline{H}_\tgt$.

For  third-order reductions the perturbation $V$ will be chosen as
\begin{equation}
\label{V3a}
V=\Delta^{2/3} V_\main + \Delta^{1/3} \tilde{V}_\extra +V_\extra, 
\end{equation}
where $(V_\main)_{--}=0$, and $V_\extra$, $\tilde{V}_\extra$ are block-diagonal. 
All operators $V_\main$, $V_\extra$, and $\tilde{V}_{\extra}$ are independent of $\Delta$.
Substituting $V$ into Eqs.~(\ref{PT1}-\ref{PT3}) gives
\begin{equation}
\label{V3b}
H_{\eff,1}=\Delta^{1/3} (\tilde{V}_\extra)_{--} + (V_\extra)_{--}, \quad  \quad
H_{\eff,2}=-\Delta^{1/3} (V_\main)_{-+}H_0^{-1} (V_\main)_{+-},
\end{equation}
and
\begin{equation}
\label{V3c}
H_{\eff,3}=(V_\main)_{-+} H_0^{-1} (V_\main)_{++} H_0^{-1}(V_\main)_{+-} + 
O(\Delta^{-1/3} \|\tilde{V}_{\extra}\|\cdot \|V_\main\|^2). 
\end{equation}
We shall  choose $\Delta$ large enough so that the last term in Eq.~(\ref{V3c}) can be neglected.
We shall chose $V_\main$ such that $H_{\eff,3}$ generates the desired logical target Hamiltonian and, may be,
some unwanted terms. The purpose of $V_\extra$ is to cancel the
unwanted terms. In addition, $V_\extra$ may include all the terms of the target Hamiltonian
that are members of the simulator class.
 Finally, the purpose of 
$\tilde{V}_{\extra}$ is to cancel the second-order term $H_{\eff,2}$.  Accordingly, we shall always choose $\tilde{V}_{\extra}$ such that 
\begin{equation}
\label{V3cc}
(\tilde{V}_\extra)_{--}=(V_\main)_{-+}H_0^{-1} (V_\main)_{+-}.
\end{equation}
Our proof will only use reductions of the order $k=1,2,3$.
The following lemmas provide sufficient conditions under which a $k$-th order
reduction achieves a desired simulation error.
Recall that our definition of a simulation depends on a particular encoding $\calE$,  see Definition~\ref{dfn:sim}.
The lemmas stated below apply to any fixed encoding $\calE$ and the 
 logical target Hamiltonian defined by Eq.~(\ref{target_logical}).
In Lemmas~\ref{lemma:1st}-\ref{lemma:3rd} we  assume
that $H_0$ is block-diagonal, $(H_0)_{--}=0$, and $(H_0)_{++}$ has all eigenvalues greater or equal to $1$.
\begin{lemma}[\bf First-order reduction]
\label{lemma:1st}
Suppose one can choose $H_0,V$ such that 
\begin{equation}
\label{1st}
\| \overline{H}_\tgt- (V)_{--} \| \le \epsilon/2.
\end{equation}
Then $H_\sm=\Delta H_0+V$  simulates
$H_\tgt$ with an error $(\eta,\epsilon)$, provided that $\Delta\ge O(\epsilon^{-1} \|V\|^2 + \eta^{-1}\|V\|)$.
\end{lemma}
\begin{proof}
Let $\calE\, : \, \calH_\tgt \to \calH_\sm$ be the chosen encoding. Recall that   $\image{(\calE)}=\calH_-$.
Let us check that $\tilde{\calE}=e^{-S}\calE$ satisfies conditions (S1-S3)
of Definition~\ref{dfn:sim}. By definition of the Schrieffer-Wolff transformation,
$e^{-S}$ maps $\calH_-$ to the low-energy subspace of $H_\sm$.
Thus $\tilde{\calE}$ maps $\calH_\tgt$ to the low-energy subspace of $H_\sm$
which proves (S1). From Eq.~(\ref{Htrunc}) 
one infers that $\|H_\eff-H_\eff(1)\|\le O(\Delta^{-1} \|V\|^2) \le \epsilon/2$.
Combining this and Eq.~(\ref{1st}) gives $\| \overline{H}_\tgt- H_\eff\|\le \epsilon$
and thus $\| H_\tgt-\calE^\dag H_\eff \calE\|\le \epsilon$.
Substituting $H_\eff=(e^S H_\sm e^{-S})_{--}$ one gets 
$\|H_\tgt-\tilde{\calE}^\dag H_\sm \tilde{\calE} \|\le \epsilon$ which proves condition~(S2). 
Finally, Eq.~(\ref{Strunc}) leads to 
$\|\calE-\tilde{\calE}\| =\| I-e^{-S}\| =O(\|S\|)=O(\Delta^{-1} \|V\|)\le \eta$. This proves condition~(S3). 
\end{proof} 
\begin{lemma}[\bf Second-order reduction]
\label{lemma:2nd}
Suppose one can choose $H_0,V_\main,V_\extra$ such that 
$V_\extra$ is block-diagonal, 
$(V_\main)_{--}=0$,   and 
\begin{equation}
\label{2nd}
\|  \overline{H}_\tgt- (V_\extra)_{--}  +  (V_\main)_{-+} H_0^{-1} (V_\main)_{+-}\| \le \epsilon/2.
\end{equation}
Suppose the norm of $V_\main, V_\extra$ is at most $\Lambda$.
Then $H_\sm=\Delta H_0+\Delta^{1/2} V_\main+V_\extra$  simulates
$H_\tgt$ with an error $(\eta,\epsilon)$ provided that $\Delta\ge O(\epsilon^{-2} \Lambda^6 + \eta^{-2} \Lambda^2)$.
\end{lemma}
\begin{proof}
Let $V=\Delta^{1/2} V_\main+V_\extra$. 
By assumption, $V$ has norm $O(\Delta^{1/2} \Lambda)$. 
Substituting this into Eq.~(\ref{Htrunc})  gives
$\|H_\eff-H_\eff(2)\|\le O(\Delta^{-2} \|V\|^3) = O(\Delta^{-1/2} \Lambda^3)\le \epsilon/2$. 
From Eqs.~(\ref{V2b},\ref{2nd}) one gets $\| \overline{H}_\tgt -H_\eff(2)\|\le \epsilon/2$
which gives  $\|H_\eff-  \overline{H}_\tgt\|\le \epsilon$.
Finally, Eq.~(\ref{Strunc}) leads to 
$\|\calE-\tilde{\calE}\| =\| I-e^{-S}\| =O(\|S\|)=O(\Delta^{-1} \|V\|)=O(\Delta^{-1/2} \Lambda) \le \eta$. 
The rest of the proof is identical to the one of Lemma~\ref{lemma:1st}.
\end{proof}

\begin{lemma}[\bf Third-order reduction]
\label{lemma:3rd}
Suppose one can choose $H_0,V_\main,V_\extra,\tilde{V}_\extra$ such that 
$V_\extra$, $\tilde{V}_\extra$ are block-diagonal, 
$(V_\main)_{--}=0$, 
\begin{equation}
\label{3rdA}
\| \overline{H}_\tgt- (V_\extra)_{--}  -  (V_\main)_{-+} H_0^{-1}  (V_\main)_{++}H_0^{-1}  (V_\main)_{+-}\| \le \epsilon/2,
\end{equation}
and
\begin{equation}
\label{3rdB}
(\tilde{V}_\extra)_{--}=(V_\main)_{-+}H_0^{-1} (V_\main)_{+-}.
\end{equation}
Suppose the norm of $V_\main, V_\extra,\tilde{V}_\extra$ is at most $\Lambda$.
Then $H_\sm=\Delta H_0+\Delta^{2/3}V_\main+\Delta^{1/3}\tilde{V}_\extra+V_\extra$   simulates
$H_\tgt$ with an error $(\eta,\epsilon)$ provided that $\Delta\ge O(\epsilon^{-3} \Lambda^{12}+\eta^{-3} \Lambda^3)$.
\end{lemma}
\begin{proof}
Let $V=\Delta^{2/3}V_\main+\Delta^{1/3}\tilde{V}_\extra+V_\extra$.
By assumption, $V$ has norm $O(\Delta^{2/3} \Lambda)$. 
Substituting this into Eq.~(\ref{Htrunc})  gives
$\|H_\eff-H_\eff(3)\|\le O(\Delta^{-3} \|V\|^4) = O(\Delta^{-1/3} \Lambda^4)\le \epsilon/4$. 
Combining Eqs.~(\ref{V3b},\ref{V3c}) and Eqs.~(\ref{3rdA},\ref{3rdB})  one gets 
\[
\| \overline{H}_\tgt-H_\eff(3)\| \le \epsilon/4 + O(\Delta^{-1/3} \|\tilde{V}_{\extra}\|\cdot \|V_\main\|^2)
=\epsilon/4 + O(\Delta^{-1/3}\Lambda^3)\le \epsilon/2.
\]
This gives $\|H_\eff- \overline{H}_\tgt\|\le \epsilon$.
Finally, Eq.~(\ref{Strunc}) leads to 
$\|\calE-\tilde{\calE}\| =\| I-e^{-S}\| =O(\|S\|)=O(\Delta^{-1} \|V\|)=O(\Delta^{-1/3} \Lambda) \le \eta$. 
The rest of the proof is identical to the one of Lemma~\ref{lemma:1st}.
\end{proof}

The above framework covers all our reductions except for the one
described in Section~\ref{sect:TIM3}, namely, the reduction from TIM on degree-$3$ graphs to TIM on general graphs.
The latter is a first-order reduction where the Hamiltonian $H_0$ is chosen as the
the one-dimensional TIM. In this case $H_0$ has only approximately degenerate ground subspace, $(H_0)_{--}\ne 0$,
and the rules Eq.~(\ref{PT1}-\ref{PT3}) for computing the effective Hamiltonian
no longer apply.
The following is a simple generalization of Lemma~\ref{lemma:1st}.
\begin{lemma}[\bf Generalized first-order reduction]
\label{lemma:1st+}
Suppose one can choose $H_0,V$ such that 
$H_0$ is block-diagonal,  $(H_0)_{++}$ has all eigenvalues at 
least $\Delta$, and $(H_0)_{--}$ has all eigenvalues in the interval 
$[-\Delta/2,\Delta/2]$. Suppose also that
\begin{equation}
\label{1st+}
\| \overline{H}_\tgt- (H_0)_{--} - (V)_{--} \| \le \epsilon/2.
\end{equation}
Then $H_\sm=H_0+V$  simulates
$H_\tgt$ with an error $(\eta,\epsilon)$ provided that $\Delta\ge O(\epsilon^{-1} \|V\|^2+\eta^{-1}\|V\|)$.
\end{lemma}
\begin{proof}
We can use the same arguments as in the proof of Lemma~\ref{lemma:1st}
except that now $H_{\eff}(1)=H_{\eff,0}+H_{\eff,1}$, where
$H_{\eff,0}=(H_0)_{--}$ and $H_{\eff,1}=(V)_{--}$.
By assumption, the unperturbed Hamiltonian $H_0$ has an energy gap at least $\Delta/2$ 
separating $\calH_-$ and $\calH_+$. 
Lemma~3.4 of Ref.~\cite{BDL11} implies that 
the series for $S$ and   $H_{\eff}$ converges absolutely for $\|V\|<\Delta/32$
and the Taylor coefficients $S_j$ and $H_{\eff,j}$ are bounded as in Eq.~(\ref{Strunc},\ref{Htrunc}).
Therefore $\|H_\eff-H_\eff(1)\| \le O(\Delta^{-1} \|V\|^2)$ and $\|S\|\le O(\Delta^{-1} \|V\|)$.
The rest of the proof is the same as in  Lemma~\ref{lemma:1st}.
\end{proof}

\section{Reduction from  degree-$3$ graphs to general graphs}
\label{sect:TIM3}

Consider a target Hamiltonian $H_\tgt$ describing the TIM on $n$ qubits. 
We assume that each qubit can be coupled to any other qubit with $ZZ$ interactions. 
Below we show how to simulate $H_\tgt$ using TIM with interactions of degree $3$.
Let us first informally sketch the main idea. 
We shall encode each qubit $u$ of the target model into the ground subspace of a TIM Hamiltonian 
on a  one-dimensional chain $L_u$  of some length $m$.   
The chain will be in the ferromagnetic phase such that the ground states $\psi_0$ and $\psi_1$
originating from the
two different $\ZZ_2$-symmetry sectors  are approximately degenerate 
forming one logical qubit. The basis states of the logical qubit will be defined as 
$|\overline{0}\rangle\sim |\psi_0\rangle+|\psi_1\rangle$
and $|\overline{1}\rangle\sim |\psi_0\rangle-|\psi_1\rangle$.
Important parameters of the logical qubit are the energy  splitting $\delta$ between $\psi_0$ and $\psi_1$
and the energy gap $\Delta$  separating $\psi_0,\psi_1$  from excited states. 
We shall work in the regime $\delta\ge \poly(1/m)$ and $\Delta/\delta \ge \poly(m)$ which
can be achieved if the chain is sufficiently close to the quantum phase transition point.
The logical Pauli operator $\overline{X}_u$ will be simulated by the energy splitting 
between  $\psi_0$ and $\psi_1$. The logical Pauli operator $\overline{Z}_u$ will be simulated by 
applying a magnetic field  $h Z_i$ to an arbitrarily chosen qubit  $i\in L_u$. 
The strength of the field $h$ will be  much smaller than the gap $\Delta$
to enable a perturbative analysis. We shall only need the first-order perturbation theory. 
 Since the one-dimensional TIM is exactly solvable,
all parameters of the logical qubit will be efficiently computable. 
Assuming that the target model has $n$ qubits, the simulator model
will be composed of $n$ chains $L_1,\ldots,L_n$ of length $m$ each. We can simulate a logical
interaction $\overline{Z}_u \overline{Z}_v$ by choosing an arbitrary pair of qubits 
$i\in L_u$, $j\in L_v$ and applying the Ising interaction $Z_i Z_j$.
Since each logical qubit $L_u$ is coupled to at most $n-1$ other logical qubits,
choosing $m\ge n-1$ guarantees that each qubit of $L_u$ is coupled 
to at most one qubit from a different chain.
   In addition,
each qubit of $L_u$ must be  coupled to its left and right  neighbors in $L_u$.
Thus the simulator model has interactions of degree $3$. 
A logical transverse field $\overline{X}_u$ is automatically simulated due to the ground state 
energy splitting of $L_u$. Thereby, we shall be able to simulate any logical  TIM Hamiltonian.

Let us now describe the reduction formally. 
For the sake of clarity, 
we begin by constructing a single logical qubit.  
Consider  a  chain of $m$ qubits with periodic  boundary conditions.
Qubits will be labeled by elements of the cyclic group $j\in \ZZ_m$.
Consider a TIM Hamiltonian 
\begin{equation}
\label{Hring}
H_\ring=-g\sum_{j\in \ZZ_m} Z_j Z_{j+1} - \sum_{j\in \ZZ_m} X_j,
\end{equation}
where 
\begin{equation}
\label{g}
g=1+\frac{c\log{(m)}}m
\end{equation}
for some parameter $c\gg 1$ to be chosen later. 
The Hamiltonian $H_\ring$ can be diagonalized via the Jordan-Wigner transformation
and all its eigenvalues  have been explicitly computed~\cite{Pfeuty1970}.
Let the three smallest eigenvalues of $H_\ring$
be $E_0\le E_1\le E_2$.
We shall use the following well-known fact~\cite{formfactor,Fendley13}.
\begin{fact}
\label{fact:TIM1}
Let $\omega_m=e^{2\pi i/m}$ be the $m$-th root of unity. For any $g>1$ one has
\begin{equation}
\label{E012}
E_0=-\sum_{j\in \ZZ_m} |g-\omega_m^{j+1/2}|,  \quad 
E_1=-\sum_{j\in \ZZ_m} |g-\omega_m^j|,  \quad  E_2=E_0+4|g-\omega_m^{1/2}|.
\end{equation}
Furthermore, the eigenvalues $E_0$, $E_1$ have multiplicity one and the corresponding
eigenvectors $\psi_0,\psi_1$ satisfy 
 $X^{\otimes m} \psi_0=\psi_0$ and $X^{\otimes m} \psi_1=-\psi_1$.
\end{fact}
Let $\calH=(\CC^2)^{\otimes m}$ be the full Hilbert space of $n$ qubits.
Define an encoding $\calE \, : \, \CC^2\to \calH$ as
\begin{equation}
\label{TIM3W}
\calE|0\rangle\equiv |\overline{0}\rangle =(|\psi_0\rangle + |\psi_1\rangle)/\sqrt{2}
\quad \mbox{and} \quad
\calE|1\rangle\equiv |\overline{1}\rangle =(|\psi_0\rangle - |\psi_1\rangle)/\sqrt{2}.
\end{equation}
Thus we  identify the states $\psi_0$ and $\psi_1$ with the states $|+\rangle$ and $|-\rangle$ 
of the logical qubit.  Decompose 
$\calH=\calH_-\oplus \calH_+$, 
where $\calH_-$ is the logical subspace spanned by $\psi_0,\psi_1$ and
and $\calH_+$ is the orthogonal complement of $\calH_-$.  Then
$H_\ring$ is block-diagonal.  Performing the overall energy shift
by $(E_0+E_1)/2$ we arrive at 
\begin{equation}
\label{Hring--}
(H_\ring)_{--}=-\delta \overline{X} \quad \mbox{and} \quad (H_\ring)_{++}\ge \Delta I,
\end{equation}
where $\overline{X}=\calE X\calE^\dag=|\psi_0\rangle\langle\psi_0|-|\psi_1\rangle\langle\psi_1|$ is the logical Pauli $X$ operator,
\[
\delta=(E_1 -E_0)/2 \quad \mbox{and} \quad \Delta=4|g-\omega_m^{1/2}|-\delta.
\]
Note that $\delta$ and $\Delta$ can be computed in time $\poly(m)$ using Eq.~(\ref{E012}).
In Appendix~A we prove that 
\begin{equation}
\label{delta_bound}
\Omega(m^{-c-3/2})\le \delta\le O(m^{-c-1/2}) 
\end{equation}
in the limit $m\to \infty$. Therefore
\begin{equation}
\label{Delta_bound}
\Delta\ge \frac{4c\log{(m)}}{m} - O(m^{-1})-O(m^{-c-1/2})\ge \Omega(m^{-1}).
\end{equation}
By choosing the  constant $c$ sufficiently large  we can make the ratio $\Delta/\delta$
bigger than any fixed polynomial of $m$ and, at the same time, keep $\delta$
at least polynomial in $1/m$.

Consider now a perturbation 
$V=hZ_j$, where $|h|\ll \Delta$ and $j\in \ZZ_m$ is an arbitrarily chosen qubit.
The first-order effective Hamiltonian acting
on $\calH_-$ is $h(Z_j)_{--}$.  To compute $(Z_j)_{--}$ 
we need to know matrix elements $\langle \psi_\alpha|Z_j|\psi_\beta\rangle$
for $\alpha,\beta=0,1$.
Note that $Z_j$ anti-commutes with $X^{\otimes m}$. Using Fact~\ref{fact:TIM1}
we infer that $\langle \psi_0|Z_j|\psi_0\rangle=0$ and $\langle \psi_1|Z_j|\psi_1\rangle=0$.
Therefore $(Z_j)_{--}$ must be a linear combination of the logical Pauli operators
$\overline{Z}$ and $\overline{Y}$.  Since $H_\ring$ has real matrix elements in the standard basis,
the same is true for the restrictions of $H_\ring$ onto the sectors $X^{\otimes m}=\pm 1$.
Therefore $\psi_0$ and $\psi_1$ must have real amplitudes in the standard basis. 
This shows that $\langle \psi_0|Z_j|\psi_1\rangle$ must be real and thus 
\begin{equation}
\label{Z--}
(Z_j)_{--}= \xi \overline{Z}, \quad \mbox{where} \quad \xi\equiv \langle \psi_1 |Z_j|\psi_0\rangle
\end{equation}
and $\overline{Z}=\calE Z\calE^\dag =|\psi_0\rangle\langle\psi_1| + |\psi_1\rangle\langle\psi_0|$ is the logical Pauli $Z$ operator.
It can be easily shown that $\xi$ does not depend on the choice of $j$. 
We shall need the following expression for $\xi$ computed in Ref.~\cite{formfactor},
see Eq.~(77) therein.
\begin{fact}
\label{fact:TIM2}
Suppose $g>1$. Let $\epsilon_p\equiv |g-\omega_m^p|$, where $p$ is either integer or half-integer.
Then  
\begin{equation}
\label{xi}
|\xi|=\frac{ (1-g^{-2})^{1/8} \prod_{p\in \ZZ_m} \prod_{q\in \ZZ_m+1/2} (\epsilon_p + \epsilon_q)^{1/4}}
{\prod_{p,p'\in \ZZ_m} (\epsilon_p+\epsilon_{p'})^{1/8} \prod_{q,q'\in \ZZ_m+1/2} (\epsilon_q+\epsilon_{q'})^{1/8}}
\end{equation}
\end{fact}
Furthermore, in Appendix~A we prove that $\xi$ is positive and 
\begin{equation}
\label{xi_bound}
\xi\ge (1-g^{-2})^{1/8}\ge \Omega(m^{-1/8})
\end{equation}
for all $m\ge 2$ and for all $g>1$. 
Note that $\xi$ can be computed in time $\poly(m)$ using Eq.~(\ref{xi}).
We can now simulate any  target Hamiltonian 
on a single qubit which has a form
\begin{equation}
\label{logical1a}
H_\tgt=-h^x X+ h^z Z, \quad \quad h^x\ge 0.
\end{equation}
Let $J=\max{(|h^z|,h^x)}$ be the  interaction strength of $H_\tgt$.
Choose the simulator Hamiltonian as
\begin{equation}
\label{logical1b}
H_\sm = H_0+V, \quad H_0=h^x \delta^{-1} H_\ring, \quad \quad V= h^z \xi^{-1} Z_j.
\end{equation}
Here $j\in \ZZ_m$ is an arbitrary qubit.
From Eqs.~(\ref{Hring--},\ref{Z--}) we infer that 
the first-order effective Hamiltonian  acting on $\calH_-$ is 
\[
H_{\eff}(1)=(H_\sm)_{--} = h^x \delta^{-1} (H_\ring)_{--} + h^z \xi^{-1} (Z_j)_{--} = -h^x \overline{X}+h^z \overline{Z} =
\overline{H}_\tgt.
\]
Note that $H_0$ has an energy gap $\Delta'=h^x\Delta \delta^{-1}$
separating $\calH_-$ and $\calH_+$. 
By Lemma~\ref{lemma:1st+}, the Hamiltonian $H_\sm$ and the encoding $\calE$
simulate $H_\tgt$ with an error $(\eta,\epsilon)$ 
provided that
$\Delta'\ge \poly(h^z\xi^{-1},\epsilon^{-1},\eta^{-1})$ for some constant degree polynomial.
We can assume without loss of generality  that $h^x\ge \epsilon/2$ since
we only need to approximate the target Hamiltonian with an error $\epsilon/2$, see
Lemma~\ref{lemma:1st+}. Then $\Delta'\ge \Omega(\epsilon m^{-1} \delta^{-1})$.
Here we used the bound $\Delta=\Omega(m^{-1})$, see Eq.~(\ref{Delta_bound}).
Since $|h^z|\le J$, we have to satisfy
$\delta^{-1}\ge \poly(m,J,\xi^{-1},\epsilon^{-1},\eta^{-1})$. 
Since $\xi^{-1}=O(m^{1/8})$, see Eq.~(\ref{xi_bound}), this is equivalent to
$\delta^{-1}\ge \poly(m,J,\epsilon^{-1},\eta^{-1})$.
This can always be achieved by choosing a large enough constant $c$
in Eq.~(\ref{g}) since $\delta^{-1}=\Omega( m^{c+3/2})$, see Eq.~(\ref{delta_bound}).
Finally, we express $\delta^{-1}$ in terms of the interaction strength $J'$ of the simulator Hamiltonian.
From Eq.~(\ref{logical1b}) one gets $J' =O(h^x \delta^{-1})$
and thus we can achieve a simulation error $(\eta,\epsilon)$ by choosing
$J'=\poly(m,J,\epsilon^{-1},\eta^{-1})$.

Consider now  a target Hamiltonian   
on $n$  qubits which has a form 
\begin{equation}
\label{TIMtarget}
H_\tgt=\sum_{0\le u<v\le n-1} \omega_{u,v} Z_u Z_v + \sum_{u=0}^{n-1}  h^z_u Z_u
- h^x_u X_u.
\end{equation}
Without loss of generality $h^x_u\ge 0$ (otherwise, conjugate the Hamiltonian by $Z_u$).
We shall encode each qubit $u$ into a chain $L_u$ of length $m=n$ as defined above. 
Let $H_\ring^{(u)}$ be the Hamiltonian Eq.~(\ref{Hring}) describing the chain $L_u$.
We shall arrange the  chains into a square grid of size $n\times n$ such that 
a cell $(u,i)$ of the grid represents the $i$-th qubit of the chain $L_u$. Here
$0\le u,i\le n-1$. 
All chains use the same parameter $g$. Choose the simulator Hamiltonian as 
\begin{equation}
\label{TIMsimulator}
H_\sm = H_0+V, \quad H_0=\delta^{-1} \sum_{u=0}^{n-1} h_u^x H_\ring^{(u)}, 
\end{equation}
\begin{equation}
\label{TIMsimulator1}
V=\xi^{-2} \sum_{0\le u<v\le n-1} \omega_{u,v}
Z_{(u,v)} Z_{(v,u)} + \xi^{-1} \sum_{u=0}^{n-1} h^z_u Z_{(u,0)}.
\end{equation}
Note that $H_\sm$ is a TIM Hamiltonian acting on $n^2$ qubits and such that
each qubit is coupled to at most three other qubits with $ZZ$ interactions. 
Namely, a qubit $(u,v)$ is coupled only to the qubits $(u,v\pm 1)$ and $(v,u)$. 
Let $\calH_-$ be the $n$-fold tensor product of the two-dimensional logical subspaces
describing each chain $L_u$. The above analysis for a single logical qubit shows that 
$(H_\ring^{(u)})_{--}=-\delta \overline{X}_u$ and $(Z_{(u,v)})_{--}=\xi\overline{Z}_u$
for any qubit $v$ in the chain $L_u$. Therefore 
the first-order effective Hamiltonian acting on $\calH_-$ is 
\begin{equation}
\label{logical_n}
H_\eff(1)=(H_\sm)_{--}
=\sum_{0\le u<v\le n-1} \omega_{u,v} \overline{Z}_u \overline{Z}_v + \sum_{u=0}^{n-1}  h^z_u \overline{Z}_u
- h^x_u \overline{X}_u=
\overline{H}_\tgt,
\end{equation}
where the encoding $\calE$ is the $n$-fold tensor product of 
single qubit encodings defined in Eq.~(\ref{TIM3W}).

Let  $(\eta,\epsilon)$ be the desired simulation error and $J$ be the maximum magnitude of the coefficients in $H_\tgt$. 
We can assume without loss of generality that $h_u^x\ge \epsilon/2n$ for all $u$.
Then the  energy gap of $H_0$ separating $\calH_-$ from excited states
is $\Delta'\ge \epsilon\delta^{-1} \Delta/2n$, where $\Delta$ is the energy gap of a single chain,
see Eq.~(\ref{Hring--}).
By Lemma~\ref{lemma:1st+}, the Hamiltonian $H_\sm$ and the encoding $\calE$
simulate $H_\tgt$ with an error
$(\eta,\epsilon)$ provided that $\Delta'\ge \poly(n,J,\epsilon^{-1},\eta^{-1})$.
This is equivalent to $\delta^{-1} \ge \poly(n,J,\epsilon^{-1},\eta^{-1})$ (use the same bounds as above),
which can always be satisfied by choosing a large enough constant $c$ in Eq.~(\ref{g}).
Then  the simulator Hamiltonian has interaction strength $J'=O(\delta^{-1} J)= \poly(n,J,\epsilon^{-1},\eta^{-1})$.

To conclude, we have shown that any Hamiltonian $H_\tgt\in \TIM(n,J)$
can be simulated with an error $(\eta,\epsilon)$ by a Hamiltonian
$H_\sm\in \TIM(n^2,J')$ such that $J'=\poly(n,J,\epsilon^{-1},\eta^{-1})$ and $H_\sm$ has interaction degree $3$. 
The simulation uses  the encoding $\calE$ defined in Eq.~(\ref{TIM3W}).
Furthermore, the coefficients of $H_\sm$ can be computed in time
$\poly(n)$.

\section{Reduction from TIM to  dimers}
\label{sect:TIM2HCD}

In this section we  construct a TIM simulator for  the hard-core dimers model. 
It involves a composition of a first-order and a second-order reduction.
First let us  construct a classical Ising Hamiltonian $H_0$
composed of terms proportional to $n_u$ and $n_un_v$ 
such that ground states of $H_0$ are $m$-dimers.
Consider a graph $G=(U,E)$ with $n$ nodes. Define  operators 
\begin{equation}
\label{NE}
N_U=\sum_{u\in U} n_u \quad \mbox{and} \quad N_E=\sum_{(u,v)\in E} n_u n_v.
\end{equation}
These operators act on the full Hilbert space $\calB\cong(\CC^2)^{\otimes n}$.
Define
\begin{equation}
\label{Hdimer}
H_0=N_U-2N_E+ \Gamma \sum_{D(u,v)=2} n_u n_v, \quad \quad \Gamma >2|E|.
\end{equation}
Here $D(u,v)$ denotes the graph distance between nodes $u,v$.
Note that $H_0$ is a TIM Hamiltonian (with a zero transverse field).
\begin{lemma}
\label{lemma:dimers}
The Hamiltonian $H_0$ has zero ground state energy
and its ground subspace is spanned by $m$-dimers with $0\le m\le n/2$.
Furthermore, if $S$ is an $m$-dimer and  $T=S\setminus u$ for some $u\in S$ then 
$\langle T|H_0|T\rangle=1$.
\end{lemma}
\begin{proof}
Suppose $S\subseteq U$ is an $m$-dimer. By definition, any pair of dimers in $S$
is separated by distance at least three. Thus $\langle S|n_u n_v|S\rangle=0$
whenever $D(u,v)=2$. Therefore
\[
\langle S|H_0|S\rangle = \langle S|N_U|S\rangle - 2 \langle S|N_E|S\rangle =  2m-2m=0.
\]
Next consider any subset of nodes $S$ such that $\langle S|H_0|S\rangle\le 0$.
It suffices to show that $S$ is an $m$-dimer for some integer $m$.  Indeed, the negative term in $H_0$
cannot be smaller than $-2|E|$. Since $\Gamma>2|E|$, 
the energy of $S$ can be non-positive only if 
\begin{equation}
\label{Senergy}
\langle S|n_u n_v|S\rangle=0 \quad \mbox{whenever $D(u,v)=2$}.
\end{equation}
Let $S=C_1\cup \ldots \cup C_m$ be the decomposition of $S$ into connected components.
Since the graph has no triangles, Eq.~(\ref{Senergy})
implies that each connected component of $S$ is either a single node
or a dimer. Therefore
\begin{equation}
\label{Senergy1}
0\ge \langle S|H_0|S\rangle = \sum_{\alpha=1}^m \langle C_\alpha|N_U-2N_E|C_\alpha\rangle.
\end{equation}
Clearly, $\langle C_\alpha|N_U-2N_E|C_\alpha\rangle=1$ if $C_\alpha$
is a single node and $\langle C_\alpha| N_U-2N_E|C_\alpha\rangle=0$ if $C_\alpha$
is a dimer. Thus Eq.~(\ref{Senergy1}) is possible only if all $C_\alpha$ are dimers.
From Eq.~(\ref{Senergy}) one infers that the distance between different dimers $C_\alpha$ is at least three. 
This shows that  $S$ is an $m$-dimer for some $m$. 

Finally, removing any single node $u$ from an $m$-dimer $S$ transforms one of the connected components
$C_\alpha$ into a single node. The above  shows that $T=S\setminus u$ has energy 
$\langle T|H_0|T\rangle=1$.
\end{proof}

Fix any integer $1\le m\le n/2$.
Consider a target Hamiltonian $H_\tgt\in \HCD(n,m,J)$ describing the
$m$-dimer sector of the  hard-core dimers model on some  triangle-free graph  $G=(U,E)$ with $n$ nodes,
see Eq.~(\ref{HCD}).

Our first reduction has a simulator Hamiltonian  
\begin{equation}
\label{TIM0a}
\tilde{H}_\sm=\tilde{\Delta} \tilde{H}_0+\tilde{V}, \quad \tilde{H}_0= (N_U-2m)(N_U-2m+1), \quad \quad \tilde{V} \in \TIM(n,\tilde{J}).
\end{equation}
All above operators act on the full Hilbert space $\calB$.
The perturbation $\tilde{V}$ will be chosen at the next reduction. 
Note that $\tilde{H}_\sm$ is a TIM Hamiltonian. 
Since  eigenvalues of $N_U$ are integers, 
the ground subspace of $\tilde{H}_0$  is spanned by subsets of nodes 
of cardinality $2m$ or $2m-1$.
By Lemma~\ref{lemma:1st}, the Hamiltonian $\tilde{H}_\sm$
can simulate 
the restriction of any TIM Hamiltonian $\tilde{V}$ onto the subspace
$\calH\equiv \calB_{2m}\oplus \calB_{2m-1}$.  
In the rest of this section we assume that our full Hilbert space is $\calH$.  

Our second reduction has a simulator Hamiltonian
\begin{equation}
\label{TIM1a}
H_{\sm}=\Delta H_0+V, \quad \quad V=\Delta^{1/2} V_\main+ V_{\extra},
\end{equation}
where $H_0$ is  the Hamiltonian constructed in Lemma~\ref{lemma:dimers}, 
\begin{equation}
\label{TIM1aa}
V_\main=t^{1/2}\sum_{u\in U} X_u \quad \mbox{and} \quad V_\extra=H_{\diag}+tN_U.
\end{equation}
Here $t$ and $H_\diag$ are defined by the target HCD Hamiltonian  Eq.~(\ref{HCD})
and all operators are restricted to the subspace $\calH$ with $2m$ or $2m-1$ particles.
Note that $H_\sm$ is a TIM Hamiltonian. 
Lemma~\ref{lemma:dimers} implies that 
the ground subspace of $H_0$ is spanned by $m$-dimers, that is, $\calH_-=\calD_m$. 

Let us check that the perturbation has all the properties stated in Lemma~\ref{lemma:2nd}.
First, we note that $(V_\main)_{--}=0$ since any pair of $m$-dimers either coincide or differ on at least two nodes.
Obviously, $V_\extra$ is block-diagonal. 
It remains to check Eq.~(\ref{2nd}). Let $S$ and $S'$ be arbitrary $m$-dimers.
Then
\begin{equation}
\label{Heff1}
\langle S'|(V_\main)_{-+} H_0^{-1}(V_\main)_{+-}|S\rangle
= t\sum_{u,v\in U}  \langle S'|X_v P_+H_0^{-1}P_+X_u|S\rangle.
\end{equation}
Recall that all operators in Eq.~(\ref{TIM1a}) are restricted to the subspace with $2m$ or $2m-1$ particles.
Thus $X_u|S\rangle=0$ whenever $u\notin S$ since in this case $X_u|S\rangle$ contains $2m+1$ particles. 
In the remaining case, $u\in S$,  Lemma~\ref{lemma:dimers}
implies that $X_u|S\rangle$ is an eigenvector of $H_0$ with the 
eigenvalue $1$, so that  $H_0^{-1} X_u|S\rangle=X_u|S\rangle$.
Thus 
\begin{equation}
\label{Heff3}
\langle S'|(V_\main)_{-+} H_0^{-1}(V_\main)_{+-}|S\rangle=
 t\sum_{u\in S, \, v\in S'}  \langle S'|X_v X_u|S\rangle.
\end{equation}
The sum over $u=v$ gives a contribution $t|S|\delta_{S,S'}=\langle S|tN_U|S'\rangle$.
The sum over $u\ne v$ is non-zero only if $S$ and $S'$ can be  obtained from each other
by moving one particle from some node $u$ to another node $v$,
in which case  $\langle S'|X_v X_u|S\rangle=\langle S'|W_{u,v}|S\rangle$. Thus
\begin{equation}
\label{Heff4}
(V_\main)_{-+} H_0^{-1}(V_\main)_{+-} = tN_U + t\sum_{\{u,v\}\in U} W_{u,v}
\end{equation}
and 
\begin{equation}
\label{Heff6}
(V_\extra)_{--} - (V_\main)_{-+} H_0^{-1}(V_\main)_{+-}  =  -t\sum_{\{u,v\}\in U} W_{u,v} + H_\diag =H_\tgt. 
\end{equation}
Thus all conditions of Lemma~\ref{lemma:2nd} are satisfied. 

To compose the two reductions we extend $H_\sm$ defined in Eq.~(\ref{TIM1a})
to the full Hilbert space $\calB$ and substitute $\tilde{V}=H_\sm$ into Eq.~(\ref{TIM0a}).
Combining Lemmas~\ref{lemma:sim3},\ref{lemma:1st},\ref{lemma:2nd}
we conclude that  any Hamiltonian $H_\tgt\in \HCD(n,m,J)$
can be simulated 
with an error $(\epsilon,\eta)$ by a Hamiltonian $\tilde{H}_\sm \in \TIM(n,J')$
where $J'=\poly(n,J,\epsilon^{-1},\eta^{-1})$.  
The simulation uses the trivial encoding 
$\calE\, : \, \calD_m\to \calB$, that is, $\calE|S\rangle=|S\rangle$ for any
$m$-dimer $S$.

\section{Reduction from dimers to range-$2$ bosons}
\label{sect:range2}

In this section we construct an HCD simulator for range-$2$ hard-core bosons.
It involves a  third-order reduction. 
Consider a target Hamiltonian $H_\tgt\in  \HCB_2(n,m,J)$ 
describing the $m$-particle sector of range-$2$ hard-core bosons on some graph  $G=(U,E)$
with $n$ nodes. For the sake of clarity, let us first consider 
a special case of  homogeneous hopping amplitudes,  that is, 
\begin{equation}
\label{target-range2}
H_\tgt=-t\sum_{(u,v)\in E} W_{u,v} + H_\diag, \quad \quad t\ge 0. 
\end{equation}
Recall that $H_\tgt$ acts on the Hilbert space  $\calB_{m,2}(G)$
spanned by $2$-sparse subsets of $m$ nodes.

 The HCD simulator will be  defined on 
an  extended graph
$G'=(U',E')$ obtained from $G$ by placing an extra node at the center of every edge of $G$
and attaching an extra hanging edge to every node of $G$, see Fig.~\ref{fig:d2b} 
for an example. 
The extra node located at the center of an edge $(u,v)\in E$ 
will be denoted\footnote{Here the addition is merely a symbol;  it  has no algebraic meaning.}  $u+v$.
The extra node attached  to a node $u\in U$ by a hanging edge 
will be denoted $u^*$. 
Thus the extended graph  $G'$ has a set of nodes
\begin{equation}
\label{UUU}
U'=U \cup U^* \cup U^+, \quad U^*=\{u^*\, : \, u\in U\}, \quad U^+=\{u+v \, : \, u,v\in U \; \mbox{and} \;  (u,v)\in E\}.
\end{equation}
We shall represent a boson located at a node $u\in U$ by a dimer occupying
the subset $\{u,u^*\}\subseteq U'$. 

\begin{figure}[h]
\centerline{\includegraphics[height=4cm]{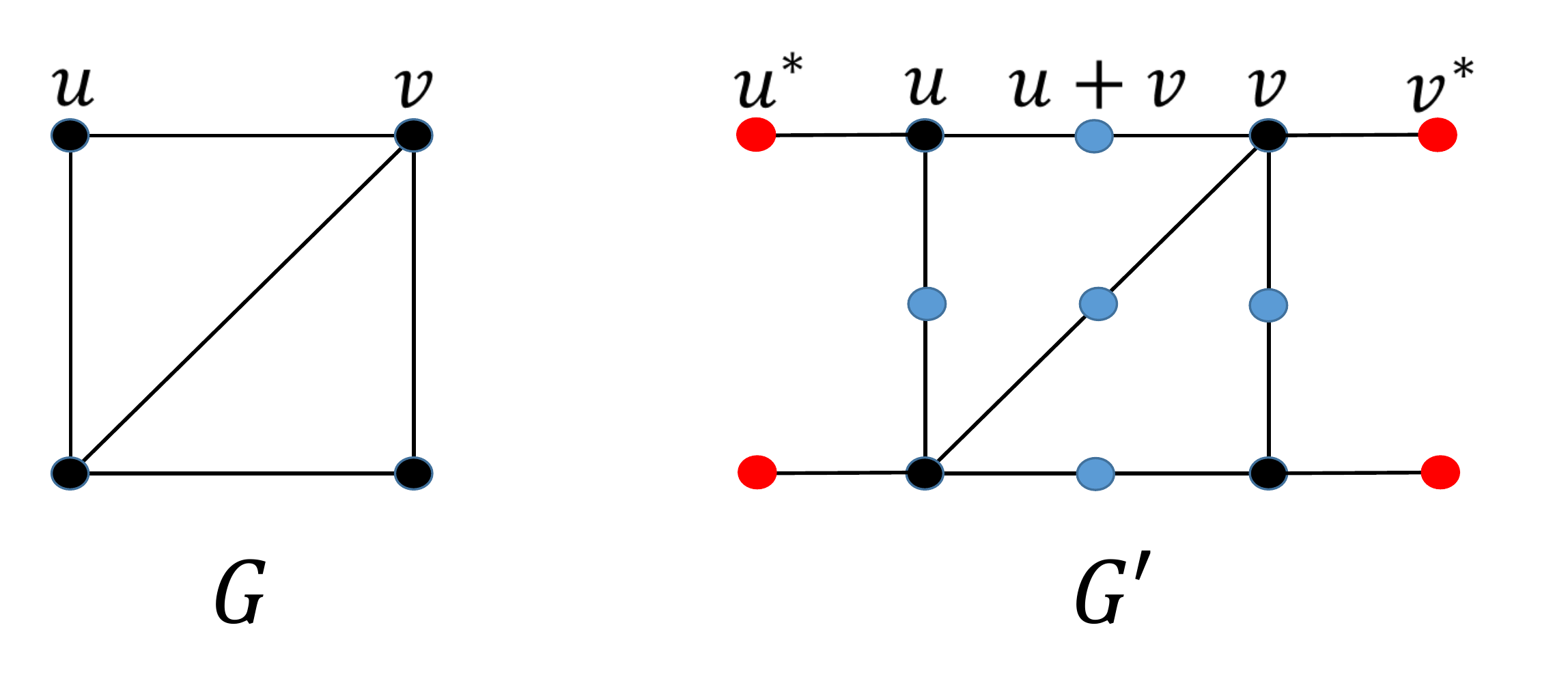}}
\caption{Construction of the extended graph $G'=(U',E')$ for hard-core dimers
starting from the graph $G=(U,E)$ of the hard-core bosons model. 
The subsets of nodes $U^*$ and $U^+$ are highlighted in red and blue respectively.
A boson located at a node $u\in U$ is represented by a dimer occupying
the subset $\{u,u^*\}\subseteq U'$. 
\label{fig:d2b}
}
\end{figure}

A simulator HCD Hamiltonian acting on the Hilbert space of $m$-dimers $\calD_m(G')$ 
is defined as $H_\sm=\Delta H_0+V$, where 
\begin{equation}
\label{HCDsim1}
H_0=\sum_{w\in U^+} \Delta_w n_w
\end{equation}
penalizes $m$-dimers occupying  the extra nodes located at the centers of edges of $G$.
For now we set $\Delta_w=1$ for all $w\in U^+$. We shall need a more general 
expression for $\Delta_w$ in the case of non-homogeneous hopping amplitudes.
Clearly, $H_0$ has zero ground state energy and its ground subspace
is spanned by $m$-dimers $S\subseteq U'$ such that $S\cap U^+=\emptyset$. 
The perturbation is defined as $V=\Delta^{2/3} V_\main + \Delta^{1/3}\tilde{V}_\extra+V_\extra$, where 
\begin{equation}
\label{HCDsim2}
V_\main=-t^{1/3} \sum_{\{u,v\}\in U'} W_{u,v}, 
\end{equation}
\begin{equation}
\label{HCDsim3}
\tilde{V}_{\extra}=t^{2/3}  \sum_{u\in U} d(u) \, n_u.
\end{equation}
\begin{equation}
\label{HCDsim4}
V_\extra=H_\diag + t\sum_{u\in U} d_2(u)  \, n_u,
\end{equation}
For now we define  $d(u)$ as the degree of a node $u$ in the original  graph $G$
and $d_2(u)\equiv d(u)(d(u)-1)$.   We shall need a more general 
expression for $d(u)$ and $d_2(u)$ in the case of non-homogeneous hopping amplitudes.

Let us check that the perturbation has all the properties stated in Lemma~\ref{lemma:3rd}.
First, we claim that $(V_\main)_{--}=0$. Indeed, suppose $|S\rangle$ is a ground state of $H_0$.
Then $S$ must be a union of  dimers $\{u,u^*\}$ with $u\in U$.
Since $(u,u^*)$ is the only edge of $G'$  attached to $u^*$, the only hopping terms that can map $S$ to some $m$-dimer 
$S'$  are those that replace some dimer $\{u,u^*\}\subseteq S$ with a dimer
$\{u,u+v\}$ for some  $(u,v)\in E$, see Fig.~\ref{fig:d2b1}. 
This requires a single hopping from $u^*$ to $u+v$.  Then
$S'$ has  a particle at some node $u+v$ and thus $|S'\rangle$  is an excited state of $H_0$. 
Thus $(V_\main)_{--}=0$. The operators $V_\extra$ and $\tilde{V}_\extra$ are block-diagonal
simply because they are diagonal. 

Let us now describe the encoding $\calE\, : \, \calB_{m,2}(G) \to \calD_m(G')$.
Recall that $\calB_{m,2}(G)$ and $\calD_m(G')$ are the Hilbert spaces of the target 
and the simulator models.  
Given a $2$-sparse subset of nodes $S\subseteq U$   in the graph $G$
let $\calE(S)\subseteq U \cup U^*$ be the subset 
of nodes in the graph $G'$ 
that includes all nodes $u\in S$ and all nodes $u^*$ such that $u\in S$. 
Define $\calE|S\rangle =|\calE(S)\rangle$. Obviously, $\calE$ is an isometry.
Let us check that $\image{(\calE)}$ coincides with the ground subspace of $H_0$. 
Indeed, suppose $S\subseteq U'$ is a ground state of $H_0$, that is, $S$ is  an $m$-dimer in $G'$ such that 
$S\subseteq U\cup U^*$. 
Then all dimers in $S$ must have a form $\{u,u^*\}$ for some $u\in U$.
Consider any distinct nodes $u,v\in S\cap U$.
By definition of an $m$-dimer, any dimers in $S$ are separated by at least three edges in
the graph $G'$. Then
the nodes $u$ and $v$ are separated by at least two edges in the graph $G$, that is,
$S\cap U$ is a $2$-sparse subset of $m$ nodes in the graph $G$.
Since $S=\calE(S\cap U)$, this shows that  $S$ belongs to the image of $\calE$. 
Conversely, if $S\subseteq U$ is any $2$-sparse
subset of  $m$ nodes in $G$ 
 then $\calE(S)$  is an $m$-dimer
in $G$ such that $\calE(S)\subseteq U\cup U^*$. 
This proves that $\image{(\calE)}$ coincides with the ground subspace of $H_0$.

Let us now check condition Eq.~(\ref{3rdA}) of Lemma~\ref{lemma:3rd}.
Consider any $m$-dimer $S\subseteq U'$ such that $S\cap U^+=\emptyset$.
We have already shown that $V_\main|S\rangle$ is a superposition of states $|S'\rangle$,
where $S'$ is obtained from $S$ by replacing a dimer $\{u,u^*\}$
with a dimer $\{u,u+v\}$ for some $u\in S\cap U$ and  $v\in U$ such that $(u,v)\in E$,
see Fig.~\ref{fig:d2b1}. 
By definition of an $m$-dimer, $\{u,u^*\}$ is separated from all other dimers of $S$ by at least 
three edges of the graph $G'$. However, since $S\cap U^+=\emptyset$, this is possible
only if $\{u,u^*\}$ is separated from all other dimers of $S$ by at least 
four edges of $G'$. Then the dimer $\{u,u+v\}$ is 
 separated from all other dimers of $S'$ by at least three
edges of $G'$, that is, $S'$ is an $m$-dimer. 
Note also that $S'$ occupies exactly one node of $U^+$, that is, $H_0|S'\rangle=|S'\rangle$
and thus $H_0^{-1} |S'\rangle=|S'\rangle$ (recall that we set $\Delta_w=1$ in the case of homogeneous hopping amplitudes).
The above arguments show that 
\begin{equation}
\label{tricky1}
(V_\main)_{+-}|S\rangle =-t^{1/3}\sum_{u\in S\cap U} \; \; \sum_{v\, : \, (u,v)\in E} \;
W_{u^*,u+v} |S\rangle.
\end{equation}
Using the above equation one can easily get
\begin{eqnarray}
\label{tricky2}
\langle S'|(V_\main)_{-+}H_0^{-1} (V_\main)_{++} H_0^{-1} (V_\main)_{+-}|S\rangle&=&
- t\sum_{(u,v)\in E} \langle S'| W_{u+v,v^*} W_{u,v} W_{u^*,u+v} |S\rangle\\
&& -t\sum_{(u,v)\ne (u,w)\in E}\; \langle S'| W_{u+w,u^*} W_{u+v,u+w} W_{u^*,u+v} |S\rangle \nonumber
\end{eqnarray}
Here $S'$ is some  $m$-dimer $S'\subseteq U'$ such that $S'\cap U^+=\emptyset$.
The terms in the first and the second line in the righthand side of Eq.~(\ref{tricky2})
 describe triple-hopping processes shown on Fig.~\ref{fig:d2b1}
and Fig.~\ref{fig:d2b2} respectively.  The former implements a logical hopping operator
$\overline{W}_{u,v}=\calE W_{u,v} \calE^\dag$, while the latter generates unwanted terms
proportional to $\overline{n}_u d(u) (d(u)-1)$, where $d(u)$ is the degree of $u$ in the graph $G$. Thus
\begin{figure}[h]
\centerline{\includegraphics[height=5cm]{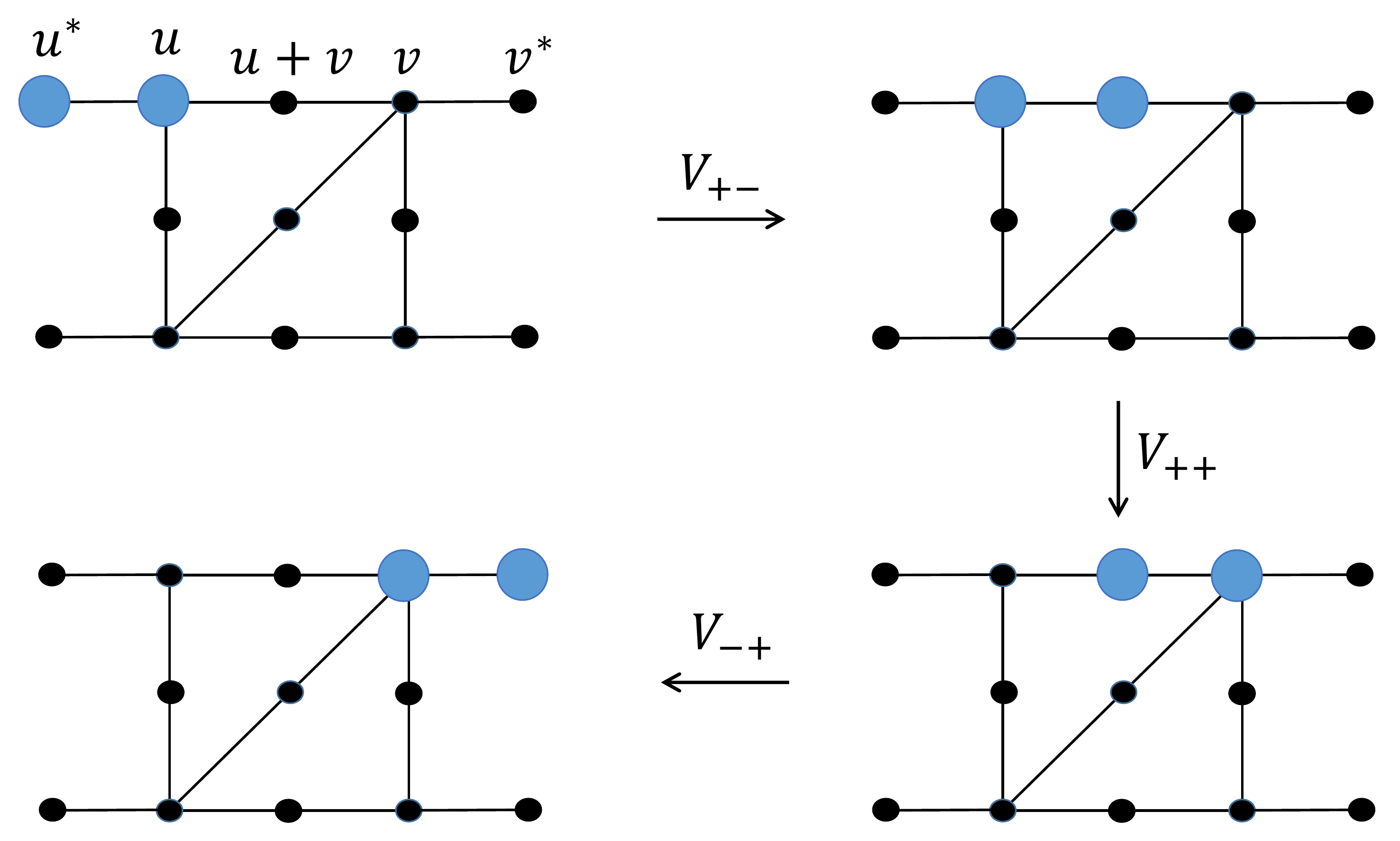}}
\caption{A third-order process transfers a dimer 
from $\{u,u^*\}$ to $\{v,v^*\}$. Here $V\equiv V_\main$.
Since each dimer in $G'$ encodes one particle in $G$,
this process simulates the logical hopping operator $\overline{W}_{u,v}$. 
\label{fig:d2b1}
}
\end{figure}

\begin{figure}[h]
\centerline{\includegraphics[height=5cm]{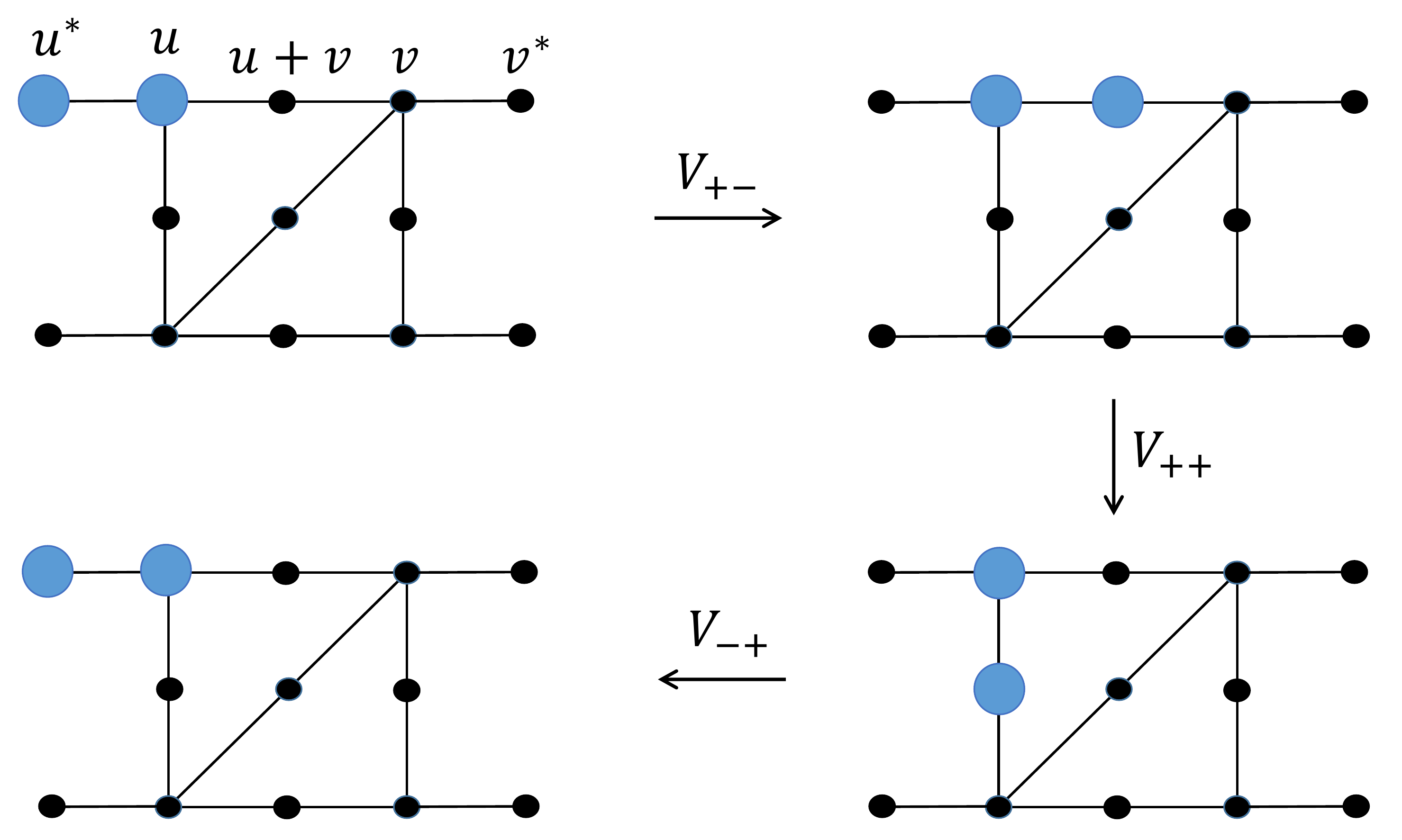}}
\caption{An unwanted  third-order process transfers a dimer 
from $\{u,u^*\}$ back to $\{u,u^*\}$. 
\label{fig:d2b2}
}
\end{figure}
\begin{equation}
\label{tricky4}
(V_\main)_{-+}H_0^{-1} (V_\main)_{++} H_0^{-1} (V_\main)_{+-}
= - t\sum_{(u,v)\in E}  \overline{W}_{u,v}
-t\sum_{u\in U} d_2(u) \overline{n}_u.
\end{equation}
Note that the last term is canceled by $(V_\extra)_{--}$, so that 
\[
(V_{\extra})_{--}+
(V_\main)_{-+}H_0^{-1} (V_\main)_{++} H_0^{-1} (V_\main)_{+-}
=\overline{H}_{\diag} - t\sum_{(u,v)\in E}  \overline{W}_{u,v} =\overline{H}_\tgt
\]
which proves  condition Eq.~(\ref{3rdA})
of Lemma~\ref{lemma:3rd}. 
It remains to check condition Eq.~(\ref{3rdB}). 
Using Eq.~(\ref{tricky1}) again one gets 
\begin{equation}
\label{tricky6}
\langle S'|(V_\main)_{-+} H_0^{-1} (V_\main)_{+-}|S\rangle=t^{2/3} \sum_{u\in S} d(u) \delta_{S,S'}=
 \langle S'|t^{2/3} \sum_{u\in U} d(u) n_u |S\rangle=\langle S'|\tilde{V}_{\extra}|S\rangle.
\end{equation}
Here $(V_\main)_{+-}$ moves a  particle from $u^*$ to $u+v$ 
and $(V_\main)_{-+}$ returns the particle back from $u+v$ to $u^*$. 
Thus all conditions of Lemma~\ref{lemma:3rd} are satisfied.

Suppose now that  $H_\tgt\in \HCB_2(n,m,J)$ has non-homogeneous hopping amplitudes, that is,
\begin{equation}
\label{target-range2gen}
H_\tgt=-\sum_{(u,v)\in E} t_{u,v} W_{u,v} + H_\diag, \quad \quad 0\le t_{u,v}\le t.
\end{equation}
By definition, $t\le J$. 
Let $(\eta,\epsilon)$ be the desired simulation error, see Definition~\ref{dfn:sim}.
Since we only need to approximate $H_\tgt$ with an error $\epsilon/2$,
see Lemma~\ref{lemma:3rd}, we can assume that 
\[
\frac{\epsilon}{2|E|} \le t_{u,v} \le t \quad \quad \mbox{for all $(u,v)\in E$}.
\]
For each node $w=u+v\in U^+$ define
\[
\Delta_w=\sqrt{\frac{t}{t_{u,v}}}.
\]
Note that $1\le \Delta_w\le \sqrt{2J \epsilon^{-1} |E| }\le \poly(n,J,\epsilon^{-1})$.
Given a node $u\in U$, let  $\calN(u)\subseteq U$ be the set of all nearest neighbors of $u$ in the graph $G$. 
Define
\[
d(u)=t^{-1/2}\sum_{v\in \calN(u)} \sqrt{t_{u,v}} \quad \quad  \mbox{and} \quad  \quad
d_2(u)=t^{-1}\sum_{v\ne v' \in \calN(u) } \;  \sqrt{t_{u,v} t_{u,v'}}.
\]
Let $H_\sm$ be the HCD simulator defined by Eqs.~(\ref{HCDsim1}-\ref{HCDsim4}). 
Exactly the same arguments as above show that $H_\sm$ satisfies conditions
of Lemma~\ref{lemma:3rd} with the target Hamiltonian Eq.~(\ref{target-range2gen}).

We conclude that any Hamiltonian $H_\tgt\in \HCB_2(n,m,J)$
can be simulated with an error $(\eta,\epsilon)$ by a Hamiltonian $H_\sm \in \HCD(n',m,J')$
where $n'=O(n^2)$ and $J'=\poly(n,J,\epsilon^{-1},\eta^{-1})$.  
The simulation uses an encoding $\calE$ that represents each particle of the target
model by a dimer in the simulator model. In particular, $\calE$ maps basis vectors to basis vectors. 
Note that the extended graph $G'$ is triangle-free
regardless of the original graph $G$, so the reduction from TIM to HCD
described in Section~\ref{sect:TIM2HCD} and the reduction
from HCD to range-$2$ HCB can be composed.

\section{Range-$2$ bosons with multi-particle interactions}
\label{sect:multi}

In this section we describe a second-order reduction 
that uses range-$2$ HCB as a simulator and 
generates the same  range-$2$ HCB Hamiltonian but with certain additional multi-particle interactions.
This  reduction  is only needed for the proof of  Theorem~\ref{thm:QA}. 

Consider a graph $G=(U,E)$. 
For any subset of nodes $S\subseteq U$ define a diagonal operator
\[
D(S)=\prod_{u\in S} (I-n_u).
\]
Let $d\le \poly(n)$ be any integer and $S_1,\ldots,S_d\subseteq U$ be arbitrary
subsets of nodes. Suppose our target Hamiltonian is
\begin{equation}
\label{enhanced1}
H_\tgt=H_{bos}-\sum_{\alpha=1}^d p_\alpha D(S_\alpha).
\end{equation}
Here $H_{bos} \in \HCB_2(n,m,J)$
describes the range-$2$ HCB on the  graph $G$ 
and $0\le p_\alpha\le J$ are arbitrary coefficients. 
The Hamiltonian $H_\tgt$ acts on the $m$-particle sector $\calB_m(G)$.
Let us show how to simulate $H_\tgt$ using the standard range-$2$ HCB model. 
The  simulator will be defined on an extended graph $G'=(U',E')$ obtained from $G$
by adding extra nodes and extra edges. For each interaction $D(S_\alpha)$
in Eq.~(\ref{enhanced1})
let us add two extra nodes denoted $a(\alpha)$ and $b(\alpha)$. 
We connect the node $b(\alpha)$ by an edge with every  node $u\in S_\alpha$. 
In addition, we connect the nodes $a(\alpha)$ and $b(\alpha)$ with each other,
see Fig.~\ref{fig:multi}.  The number of particles in the simulator model is $m'=m+d$,
where $d$ is the number of extra terms in Eq.~(\ref{enhanced1}). 
\begin{figure}[h]
\centerline{\includegraphics[height=3cm]{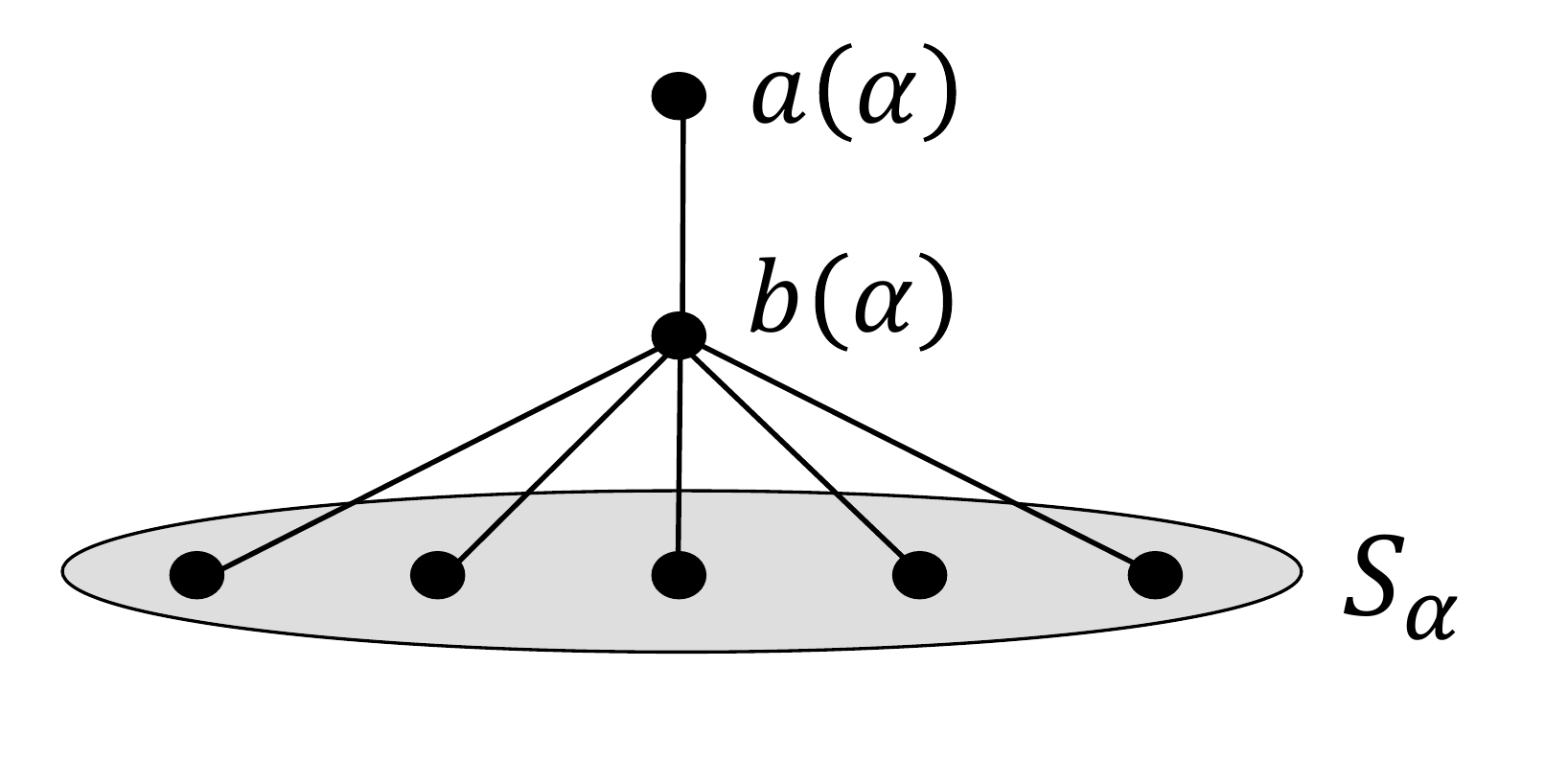}}
\caption{Simulation of  multi-particle interactions
$D(S_\alpha)$. We choose $H_0=I -n_{a(\alpha)}$ such that
the node $a(\alpha)$ is occupied for any ground state of $H_0$.
Then the node $b(\alpha)$ must be empty due to the $2$-sparsity constraint. 
The perturbation
$V$ moves the particle from $a(\alpha)$ to $b(\alpha)$ or vice verse.  \label{fig:multi}}
\end{figure}
Define a simulator
Hamiltonian as 
\begin{equation}
\label{enhanced2}
H_\sm =\Delta H_0 + V, \quad \quad H_0=\sum_{\alpha=1}^d I - n_{a(\alpha)},
\quad \quad V=\Delta^{1/2} V_{\main} + V_\extra,
\end{equation}
\begin{equation}
\label{enhanced3}
V_{\main}=-\sum_{\alpha=1}^d\,  \sqrt{p_\alpha} \, \, W_{a(\alpha),b(\alpha)}
\quad \mbox{and} \quad V_{\extra}=H_{bos}. 
\end{equation}
Clearly, $H_0$ has zero ground state energy and the ground subspace of $H_0$
is spanned by all $2$-sparse configurations of particles in $G'$ such that 
the node $a(\alpha)$ is occupied for each $\alpha$.
Note that the node $b(\alpha)$ must be empty due to the $2$-sparsity constraint. 
 Define an encoding $\calE\, : \, \calB_m(G)\to \calB_{m+d}(G')$
as follows. If $S\subseteq U$ is a $2$-sparse subset, define
$\calE(S)=S\cup \{ a(1), \cdots ,a(d)\} \subseteq U'$. 
Note that $\calE(S)$ is a $2$-sparse subset since a  node $a(\alpha)$ 
has only one neighbor $b(\alpha)$ and the latter never   belongs to $\calE(S)$. 
Define $\calE|S\rangle=|\calE(S)\rangle$. The above shows that
$\calE$ is an isometry and  the image of $\calE$ coincides with the ground subspace of $H_0$. 
Let us check that the perturbation $V$ satisfies conditions of Lemma~\ref{lemma:2nd}.
Obviously, $(V_\main)_{--}=0$ since any term in $V_\main$ moves a particle
from  $a(\alpha)$ to $b(\alpha)$ or vice verse. 
The operator $V_\extra$ is block-diagonal since it acts trivially on the extra nodes
$a(\alpha)$, $b(\alpha)$. Let us check condition Eq.~(\ref{2nd}) of Lemma~\ref{lemma:2nd}.
Note that no hopping in the original graph $G$ is
prohibited due to the presence of extra particles at $a(\alpha)$ since these
particles are separated from any node of $G$ by at least two edges. 
Thus $(V_\extra)_{--}=\overline{H}_{bos}$, where $\overline{H}_{bos}=\calE H_{bos} \calE^\dag$
is the encoded version of $H_{bos}$. 
Let us compute $(V_\main)_{-+}H_0^{-1} (V_\main)_{+-}$.
Suppose  $|S\rangle$ is a ground state of $H_0$. 
The $2$-sparsity condition  implies that a node $b(\alpha)$ cannot be occupied 
if $S_\alpha$ contains at least one particle. This shows that 
$(W_{a(\alpha),b(\alpha)})_{+-} |S\rangle=0$ if $S\cap S_\alpha\ne \emptyset$,
Otherwise, $(W_{a(\alpha),b(\alpha)})_{+-}$ moves the particle from $a(\alpha)$ to $b(\alpha)$. Thus
\[
(W_{a(\alpha),b(\alpha)})_{+-} |S\rangle=D(S_\alpha) |(S\backslash a(\alpha))\cup b(\alpha)\rangle.
\]
Note that the state in the righthand side is an eigenvector of $H_0$ with an eigenvalue $1$.
Note also that if $(V_\main)_{+-}$ moves a particle from some node $a(\alpha)$ to $b(\alpha)$
then $(V_\main)_{-+}$ must return the particle from $b(\alpha)$ to $a(\alpha)$. Thus
\[
(V_\main)_{-+}H_0^{-1} (V_\main)_{+-}=\sum_{\alpha=1}^r p_\alpha (W_{a(\alpha),b(\alpha)})_{-+}   (W_{a(\alpha),b(\alpha)})_{+-} = \sum_{\alpha=1}^r p_\alpha \overline{D(S_\alpha)}.
\]
Here we noted that $D(S_\alpha)^2=D(S_\alpha)=\overline{D(S_\alpha)}$.  
Thus
\[
(V_{\extra})_{--} -(V_\main)_{-+}H_0^{-1} (V_\main)_{+-} = \overline{H}_{bos} -  \sum_{\alpha=1}^r p_\alpha \overline{D(S_\alpha)} =\overline{H}_\tgt,
\]
that is, all conditions of Lemma~\ref{lemma:2nd} are satisfied. 
 
We have proved that any Hamiltonian $H_\tgt \in \HCB_2(n,m,J)$
with $d$ extra diagonal terms  $-p_\alpha D(S_\alpha)$ 
such that $0\le p_\alpha\le J$ 
can be simulated with an error $(\eta,\epsilon)$
 by the Hamiltonian $H_\sm \in \HCB_2(n',m',J')$,
where $n'=n+2d$, $m'=m+d$, and $J'=\poly(n,J,\epsilon^{-1},\eta^{-1})$. 
The simulation uses an encoding $\calE$ that maps basis vectors to
basis vectors. We shall absorb  the extra diagonal terms  into the Hamiltonian
 $H_\diag$  in all subsequent reductions.

\section{Reduction from range-$2$ bosons  to range-$1$ bosons}
\label{sect:range1}

In this section we construct a range-$2$ HCB simulator for 
a range-$1$ HCB model. It involves a second-order reduction. 
Consider a target Hamiltonian $H_\tgt\in  \HCB(n,m,J)$ 
describing the $m$-particle sector of range-$1$ hard-core bosons on some graph $G=(U,E)$
with $n$ nodes, 
\begin{equation}
\label{target-range1}
H_\tgt=-\sum_{(u,v)\in E} t_{u,v} W_{u,v} + H_\diag.
\end{equation}
The range-$2$ HCB  simulator will be defined on an
extended graph $G'=(U',E')$ obtained from $G$ by placing an extra node at the center of every
edge of $G$, see Fig.~\ref{fig:range2} for an example.
The extra node located at the center of an edge $(u,v)\in E$
will be denoted $u+v$. Then the extended graph $G'$ has a set of nodes
\begin{equation}
\label{range2eq1}
U'=U\cup U^+, \quad U^+=\{ u+v \, : \, (u,v)\in E\}.
\end{equation}
\begin{figure}[h]
\centerline{\includegraphics[height=4cm]{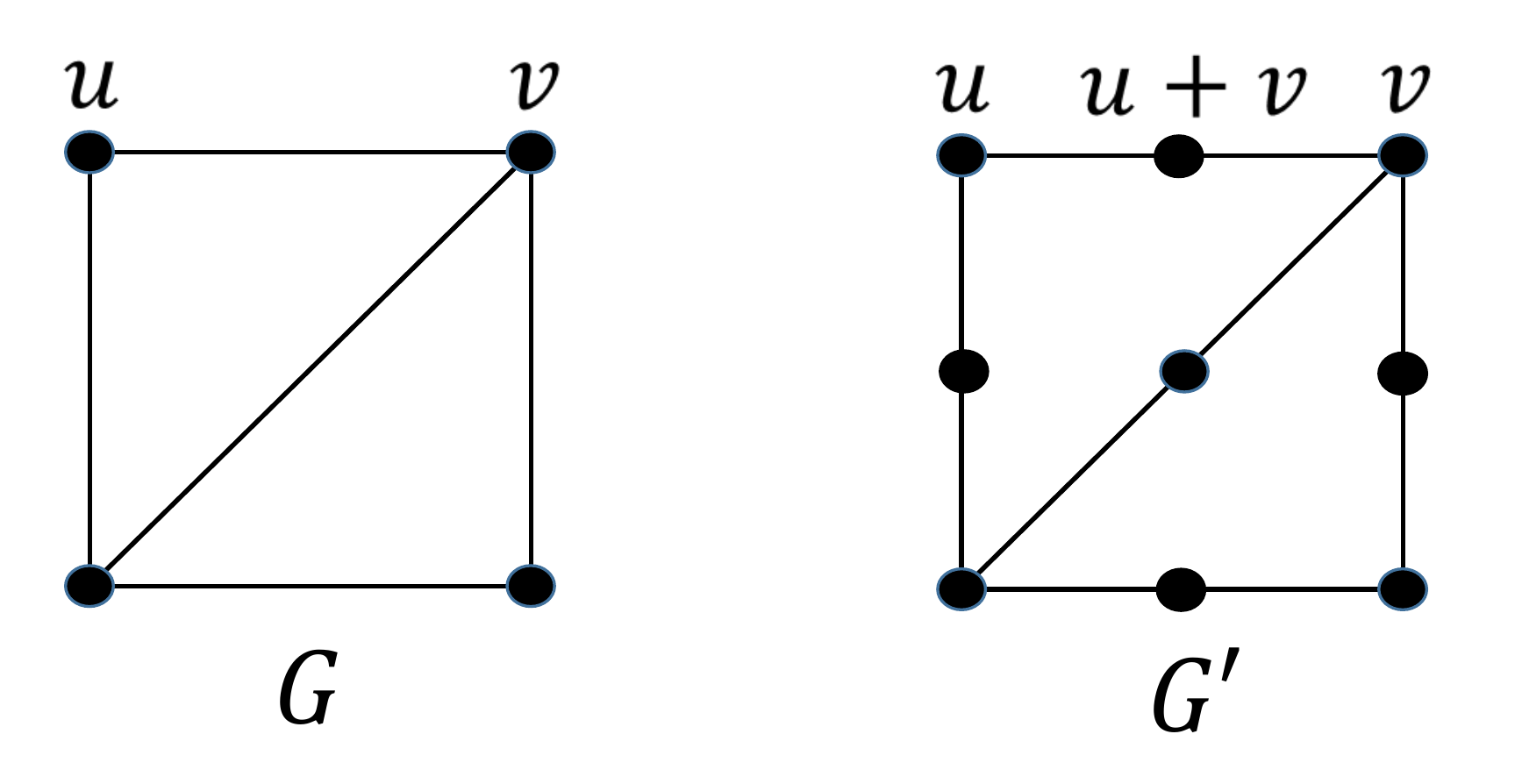}}
\caption{Construction of the extended graph $G'=(U',E')$ for the range-$2$ HCB
simulator  starting from the graph $G=(U,E)$ for the  target range-$1$ HCB.
\label{fig:range2}
}
\end{figure}

The simulator and the target models  have the same number of particles $m$. 
Thus the simulator has Hilbert space  $\calB_{m,2}(G')$ spanned
by $2$-sparse $m$-node subsets in the graph $G'$. 
 Define a simulator Hamiltonian as  $H_{\sm}=\Delta H_0+V$,
where
\begin{equation}
\label{range2eq3}
H_0=\sum_{w\in U^+}  n_w
\end{equation}
penalizes particles that occupy nodes located at the centers of edges of $G$.
We choose the perturbation as $V=\Delta^{1/2}V_{\main} + V_\extra$, where 
\begin{equation}
\label{range2eq4}
V_\main=-\sum_{(u,v)\in E'} t_{u,v}^{1/2} \, W_{u,v},
\end{equation}
\begin{equation}
\label{range2eq5}
V_{\extra}=H_\diag + \sum_{(u,v)\in E} t_{u,v} (n_u-n_v)^2.
\end{equation}
The sum in $V_{\extra}$ runs over pairs of nodes $u,v\in U$ considered as nodes
of $G'$. 
Clearly, $H_0$ has zero ground state energy and its ground subspace is
spanned by subsets of nodes  $S\subseteq U'$  such that $S\cap U^+=\emptyset$
and $|S|=m$. 
Note that any distinct nodes of $S$ are automatically separated by at least two edges of $G'$,
that is, $S$ is a $2$-sparse subset. 
Given any subset of nodes $S\subseteq U$ in the graph $G$
such that $|S|=m$, 
let $\calE(S)$ be the corresponding subset of nodes in the graph $G'$. 
We define the encoding $\calE\, : \, \calB_m(G)\to \calB_{m,2}(G')$
such that $\calE|S\rangle=|\calE(S)\rangle$.
The above shows that $\image{(\calE)}$ coincides with the ground subspace of $H_0$. 

Let us check that the perturbation $V$ satisfies the conditions of Lemma~\ref{lemma:2nd}.
First, we note that $(V_\main)_{--}=0$. Indeed, suppose $S\subseteq U'$ is a ground state of $H_0$.
Then $S\cap U^+=\emptyset$. Thus $V_\main$ can only move a particle from some node
$u\in U$ to some node $v\in U^+$ which produces an excited state of $H_0$. 
The operator $V_\extra$ is block diagonal because it is diagonal.

Let us now check condition Eq.~(\ref{2nd}) of Lemma~\ref{lemma:2nd}.
Consider any ground state of $H_0$, that is, $m$-node subset $S\subseteq U'$ such that $S\cap U^+=\emptyset$.
We claim that 
\begin{equation}
\label{range2eq6}
(V_\main)_{+-}|S\rangle = - \sum_{u\in S} \sum_{(u,v)\in E} t_{u,v}^{1/2} \, (1-n_v) W_{u,u+v} |S\rangle.
\end{equation}
Indeed,  the hopping terms in $V_\main$ can only move a particle from some node $u\in S$
to some node $u+v$ such that $(u,v)\in E$
and such that the resulting configuration of particles is $2$-sparse. 
The latter condition is satisfied iff $n_v=0$. 
Taking into account that $W_{u,u+v}|S\rangle$  is an eigenvector of $H_0^{-1}$ with an eigenvalue 
one, we get 
\begin{equation}
\label{range2eq7}
(V_\main)_{-+}H_0^{-1} (V_\main)_{+-}=
\sum_{(u,v)\in E} t_{u,v} W_{u,v}
+\sum_{u\in U} \sum_{(u,v)\in E} t_{u,v} n_u (1-n_v).
\end{equation}
Here the last term accounts for double-hopping processes where $(V_\main)_{+-}$ moves a particle
from $u$ to $u+v$ and $(V_\main)_{-+}$ returns the particle back to $u$.
 Using the identity $n_u(1-n_v)+n_v(1-n_u)=(n_u-n_v)^2$ one gets
 \begin{equation}
\label{range2eq7a}
(V_\main)_{-+}H_0^{-1} (V_\main)_{+-}=
\sum_{(u,v)\in E} t_{u,v} W_{u,v} + \sum_{(u,v)\in E} t_{u,v} (n_u-n_v)^2.
\end{equation}
The last term is exactly cancelled by $V_\extra$ which proves condition Eq.~(\ref{2nd}) of
Lemma~\ref{lemma:2nd}. Note that in this case the logical operators
$\overline{W}_{u,v}$ and $\overline{n}_u$ coincide with 
$W_{u,v}$ and $n_u$ since we encode each particle of the target model
by a single particle in the simulator model. 

To conclude, we have proved that 
any Hamiltonian
 $H_\tgt\in  \HCB(n,m,J)$ can be simulated with an error $(\eta,\epsilon)$
 by the Hamiltonian $H_\sm \in \HCB_2(n',m,J')$, where
$n'=O(n^2)$ and $J'=\poly(n,J,\epsilon^{-1},\eta^{-1})$. 
The simulation uses an encoding $\calE$ that maps basis vectors to
basis vectors.

\section{Range-$1$ bosons with a controlled hopping}
\label{sect:HCB*}

Consider a graph $G=(U,E)$ with $n$ nodes and a Hamiltonian 
$H_{bos}\in \HCB(n,m,J)$ describing the $m$-particle sector
of the range-$1$ HCB model on the graph $G$.
Suppose our target Hamiltonian is 
\begin{equation}
\label{cont1}
H_{\tgt}=H_{bos} - \sum_{(c;u,v)} \; t_{c;u,v}\,  n_cW_{u,v}, 
\end{equation}
where the sum runs over all triples of nodes $(c;u,v)$ such that $c\in U$,
$(u,v)\in E$,  and $c\notin \{u,v\}$. The term $n_cW_{u,v}$ describes a controlled
hopping process where the presence of particle at the node $c$ controls
whether the hopping between nodes $u,v$ is turned on or off.
The coefficients 
$t_{c;u,v}$ are the controlled hopping amplitudes.
We shall always assume that $t_{c;u,v}\ge 0$.
The Hamiltonian $H_\tgt$ acts on the $m$-particle sector   $\calB_m(G)$.
Let $\HCB^*(n,m,J)$ be the set of Hamiltonians $H_\tgt$ defined above
where $0\le t_{c;u,v}\le J$.
In this section we show how to simulate $H_\tgt$ by the standard range-$1$ HCB.
The simulation involves
a composition of a first-order and a second-order reduction.

The simulator model  will be defined on a graph $G'=(U',E')$
obtained from $G$ by adding certain extra nodes and extra edges.
Namely, for each triple $(c;u,v)$ that appears in Eq.~(\ref{cont1})
we add  an extra node $a(c;u,v)$ 
and a pair of extra edges connecting $a(c;u,v)$ to $u$ and $v$.
Let $n'=|U'|$ be the number of nodes in the extended graph and 
$U^+\subseteq U'$ be the set of all extra nodes $a(c;u,v)$. 
The simulator and the target models have the same number of particles $m$.

Our first reduction has a simulator Hamiltonian
\begin{equation}
\label{cont1a}
\tilde{H}_\sm=\tilde{\Delta} \tilde{H}_0+\tilde{V}, \quad \tilde{H}_0= \sum_{(c;u,v)} (I-n_c) n_{a(c;u,v)},
\end{equation}
where $\tilde{V} \in \HCB(n',m,\tilde{J})$ will  be chosen at the next reduction. 
The Hamiltonian $\tilde{H}_\sm$ acts on the Hilbert space $\calB_m(\calG')$. 
Let $\tilde{\calH}_-$ be the ground subspace of $\tilde{H}_0$. 
Obviously, $\tilde{\calH}_-$ is spanned by configurations of particles such that a node $a(c;u,v)$ 
can be occupied only if $c$ is occupied. 
This property must hold for each extra node  $a(c;u,v)$. 
Lemma~\ref{lemma:1st} shows that $\tilde{H}_\sm$ can simulate
the restriction of  any Hamiltonian from $\HCB(n',m,\tilde{J})$ onto  the subspace $\tilde{\calH}_-$.
In the rest of this section we assume that our full Hilbert space is $\calH=\tilde{\calH}_-$.
The above simulation uses the trivial encoding, that is, $\calE|S\rangle=|S\rangle$
if $S\subseteq U'$ is a ground state of $\tilde{H}_0$ and $\calE|S\rangle=0$ otherwise.

Our second reduction has  a simulator Hamiltonian 
\begin{equation}
\label{cont2}
H_\sm=\Delta H_0+V, \quad \quad H_0=\Delta \sum_{a\in U^+}  n_a, \quad \quad V=\Delta^{1/2} V_{\main} +V_\extra,
\end{equation}
where 
\begin{equation}
\label{cont3}
V_\main=-\sum_{(c;u,v)} ( t_{c;u,v})^{1/2} (W_{u,a(c;u,v)} +W_{v,a(c;u,v)}),
\end{equation}
\begin{equation}
\label{cont3a}
V_{\extra}=H_{bos}+ \sum_{(c;u,v)} t_{c;u,v} n_c(n_u+n_v).
\end{equation}
Here all operators are restricted to the subspace $\calH$ defined above. 
We choose an encoding $\calE\, : \, \calB_m(G)\to \calH$ 
that maps subsets of nodes in the graph $G$ to the corresponding subsets
of nodes in the extended graph $G'$. 
Obviously, $S\subseteq U'$ is a ground state of $H_0$ iff 
$|S|=m$ and 
$S\cap U^+=\emptyset$, that is, all the extra nodes $a(c;u,v)$ are empty. 
Thus $\image{(\calE)}$ coincides with the ground subspace of $H_0$.

Let us check that the perturbation $V$ satisfies all conditions of Lemma~\ref{lemma:2nd}.
We note that $(V_\main)_{--}=0$ since $V_\main$ can only move a particle
from (to) some ancillary node $a(c;u,v)$ which must be empty in any ground state of $H_0$.
The Hamiltonian $V_\extra$ is block-diagonal since $H_{bos}$ acts trivially on
all ancillary nodes whereas the second term in $V_\extra$ is diagonal. 
It remains to check condition Eq.~(\ref{2nd}) of Lemma~\ref{lemma:2nd}.
Let  $S\subseteq U'$ be any ground state of $H_0$. Then $S\cap U^+=\emptyset$.
We claim that  $W_{u,a(c;u,v)}|S\rangle=0$ unless $n_c=1$.
Indeed, if $n_c=0$ and $n_u=0$ then both nodes $u$ and $a(c;u,v)$ are empty.
If $n_c=0$ and $n_u=1$ then $W_{u,a(c;u,v)}$ moves a particle from $u$ to $a(c;u,v)$.
However, a state in which the node $a(c;u,v)$ is occupied and the node $c$ is empty
is orthogonal to the subspace $\calH$. Since the simulator model is restricted to $\calH$,
we have  $W_{u,a(c;u,v)}|S\rangle=0$ in both cases.
In the remaining case,  $n_c=1$,  one has
$W_{u,a(c;u,v)}|S\rangle=|S'\rangle$, where $S'=(S\setminus u) \cup a(c;u,v)$.
The above shows that 
\begin{equation}
\label{cont4}
(V_\main)_{-+}H_0^{-1} (V_\main)_{+-}=\sum_{(c;u,v)} t_{c;u,v} n_c(W_{u,v} + n_u+n_v).
\end{equation}
Here the first term describes processes where $(V_\main)_{+-}$  
moves a particle from $u$ to $a(c;u,v)$ 
and $(V_\main)_{-+}$  moves the particle from  $a(c;u,v)$  to $v$.
The last two terms describe processes where 
$(V_\main)_{+-}$ moves a particle from $u$  to $a(c;u,v)$ 
and $(V_\main)_{-+}$  returns the particle back to $u$.
The last two terms in Eq.~(\ref{cont4}) are canceled by $V_\extra$.
This proves condition Eq.~(\ref{2nd}) of Lemma~\ref{lemma:2nd}.
Note that in this case the logical operators
$\overline{W}_{u,v}$ and $\overline{n}_u$ coincide with 
$W_{u,v}$ and $n_u$ since we encode each particle of the target model
by a single particle in the simulator model. 

Combining Lemmas~\ref{lemma:sim3},\ref{lemma:1st},\ref{lemma:2nd} we 
 conclude that any  Hamiltonian $H_\tgt\in \HCB^*(n,m,J)$ can be simulated
with an error $(\eta,\epsilon)$ by the Hamiltonian $\tilde{H}_\sm \in \HCB(n',m,J')$ 
where $n'=O(n^3)$ and $J'=\poly(n,J,\epsilon^{-1},\eta^{-1})$.
The simulation uses an encoding $\calE$ that maps basis vectors to basis vectors.

\section{From range-$1$ bosons to $2$-local stoquastic Hamiltonians}
\label{sect:stoqLH}

Let us start from a simple classification of two-qubit stoquastic interactions. 
In this section we use the standard $|0\rangle$, $|1\rangle$ basis for a single qubit
and the corresponding product basis for $n$ qubits. 
\begin{lemma}
Let $H$ be a two-qubit hermitian operator such that
$H$ has real matrix elements in the standard basis and all
off-diagonal matrix elements of $H$ are non-positive. 
Then $H$ can be written as a sum of some diagonal two-qubit Hamiltonian $H_\diag$
 and a  convex linear combination
of operators
\begin{enumerate}
\item $-X\otimes |0\rangle\langle 0|$ and $-X\otimes |1\rangle\langle 1|$
\item $-|0\rangle\langle 0|\otimes X$ and $-|1\rangle\langle 1|\otimes X$
\item $-X\otimes X -Y\otimes Y$
\item $-X\otimes X + Y\otimes Y$
\end{enumerate}
\label{lemma:stoq} 
\end{lemma}
\begin{proof}
Let $G\equiv -H$. 
Since $G$ has real matrix elements, 
 the expansion of $G$ in the basis of Pauli operators contains only 
 the terms 
with even number of $Y$'s. Thus 
\begin{equation}
\label{lem0}
G=-H_\diag+ h_{XI} X\otimes I + h_{IX} I\otimes X + h_{XX} X\otimes X + h_{XZ} X\otimes Z
+ h_{ZX} Z\otimes X + h_{YY} Y\otimes Y,
\end{equation}
where $H_\diag$ is some diagonal Hamiltonian.  From 
\[
\langle 0,0|G|1,1\rangle=h_{XX}-h_{YY}\ge 0 \quad \mbox{and} \quad
\langle 0,1|G|1,0\rangle=h_{XX}+h_{YY}\ge 0
\]
one gets
\begin{equation}
\label{lem1}
h_{XX} X\otimes X+h_{YY} Y\otimes Y=p(X\otimes X+Y\otimes Y) + q(X\otimes X-Y\otimes Y),
\end{equation}
where $p=(h_{XX}+h_{YY})/2$ and $q=(h_{XX}-h_{YY})/2$ are non-negative coefficients. 
From 
\[
\langle 0,0|G|1,0\rangle=h_{XI}+h_{XZ}\ge 0 \quad \mbox{and} \quad 
\langle 0,1|G|1,1\rangle=h_{XI}-h_{XZ}\ge 0
\]
one gets
\begin{equation}
\label{lem2}
h_{XI}X\otimes I + h_{XZ} X\otimes Z = p X\otimes |0\rangle\langle0| + q X\otimes |1\rangle\langle 1|,
\end{equation}
where $p=h_{XI}+h_{XZ}$ and $q=h_{XI}-h_{XZ}$ are non-negative coefficients.
Similar calculation shows that 
\begin{equation}
\label{lem3}
h_{IX}I\otimes X + h_{ZX} Z\otimes X = p |0\rangle\langle0|\otimes X + q |1\rangle\langle 1|\otimes X,
\end{equation}
where $p=h_{IX}+h_{ZX}$ and $q=h_{IX}-h_{ZX}$ are non-negative coefficients.
The lemma now follows from Eqs.~(\ref{lem0}-\ref{lem3}).
\end{proof}

Let $H_\tgt\in \StoqLH(n,J)$ be some fixed $2$-local stoquastic Hamiltonian on $n$ qubits. 
Our goal is to simulate $H_\tgt$ by some Hamiltonian $H_\sm\in \HCB^*(n',m,J')$.
Recall that the latter 
describes the $m$-particle sector of range-$1$ hard-core bosons
with a controlled hopping 
on some graph $G=(U,E)$ with $n'$ nodes, see Section~\ref{sect:HCB*}.
We shall represent the $j$-th qubit of the target model   
by a pair of nodes $\{2j-1,2j\}\subseteq U$. 
The two basis states $|0\rangle$ and $|1\rangle$ of the $j$-th qubit
are represented by a particle located at the node $2j-1$ and $2j$ respectively
(the dual rail representation). 
Thus the number of particles in the simulator model is $m=n$.

For the sake of clarity we shall first explain how to construct an $\HCB^*$ 
 simulator  individually for each  two-qubit stoquastic interaction
listed in Lemma~\ref{lemma:stoq}. Thus we shall first consider the case $n=m=2$. 
We shall simulate interactions (1) and (2) using
a first-order reduction. Interactions (3) and (4) will require 
a composition of  a first-order and a third-order reductions. 
Then we shall explain how to combine the simulators together.
 To avoid  interference between simulators  we shall  introduce some  ancillary nodes such that 
a simulator is activated only if the corresponding ancillary node is occupied by a particle. 
 
Consider first the case $H_{\tgt}=-pX\otimes |0\rangle\langle 0|$ with $p>0$.
The $\HCB^*$ simulator is defined on a graph $G=(U,E)$,
where $U=\{1,2,3,4\}$ and $E=\{1,2\}$. 
The total number of particles is $m=2$, so that the simulator Hilbert space is $\calB_2(G)$.  
The simulator Hamiltonian is chosen as
\begin{equation}
\label{stoq(1)}
H_\sm = \Delta H_0+V, \quad \quad  H_0=n_1n_2 + n_3 n_4,
\end{equation}
\begin{equation}
\label{stoq(1)a}
V=-pn_3 W_{1,2}.
\end{equation}
Ground states of $H_0$ are subsets of nodes $\{i,j\}\subseteq U$, where $i\in \{1,2\}$
and $j\in \{3,4\}$. 
The ground subspace of $H_0$ encodes two logical qubits  as follows:
\begin{equation}
\label{stoq01}
|\overline{0,0}\rangle=|1,3\rangle, \quad |\overline{0,1}\rangle=|1,4\rangle, \quad |\overline{1,0}\rangle=|2,3\rangle, \quad
|\overline{1,1}\rangle=|2,4\rangle.
\end{equation}
This is analogous  to applying the dual-rail representation to each qubit. 
Obviously, $V$ commutes with $H_0$. 
Note that  $W_{1,2}$ implements the logical $\overline{X}$ on the first qubit
and $n_3$ implements the logical  operator $|\overline{0}\rangle\langle\overline{0}|$
on the second qubit. Lemma~\ref{lemma:1st} implies that 
$H_\sm$ can simulate the restriction of $V$ onto the logical subspace, that is,
$V_{--}= -p\overline{X}\otimes |\overline{0}\rangle\langle\overline{0}|$.
This is the desired target Hamiltonian. 
Using the same method one can simulate all elementary interactions (1) and (2)
in Lemma~\ref{lemma:stoq}.

Next consider the case 
\begin{equation}
\label{XXYY}
H_{\tgt}=-(p/2)(X\otimes X +Y\otimes Y)=-p(|1,0\rangle\la 0,1|+|0,1\rangle\langle 1,0|), \quad p> 0.
\end{equation}
The $\HCB^*$ simulator will be defined on a graph $G=(U,E)$ 
where $U=\{1,2,3,4,a\}$, see  Fig.~\ref{fig:graph1}.
We choose the total number of particles $m=2$.
 The simulation involves a composition of  a first-order and a second-order reductions. 

\begin{figure}[h]
\centerline{\includegraphics[height=4cm]{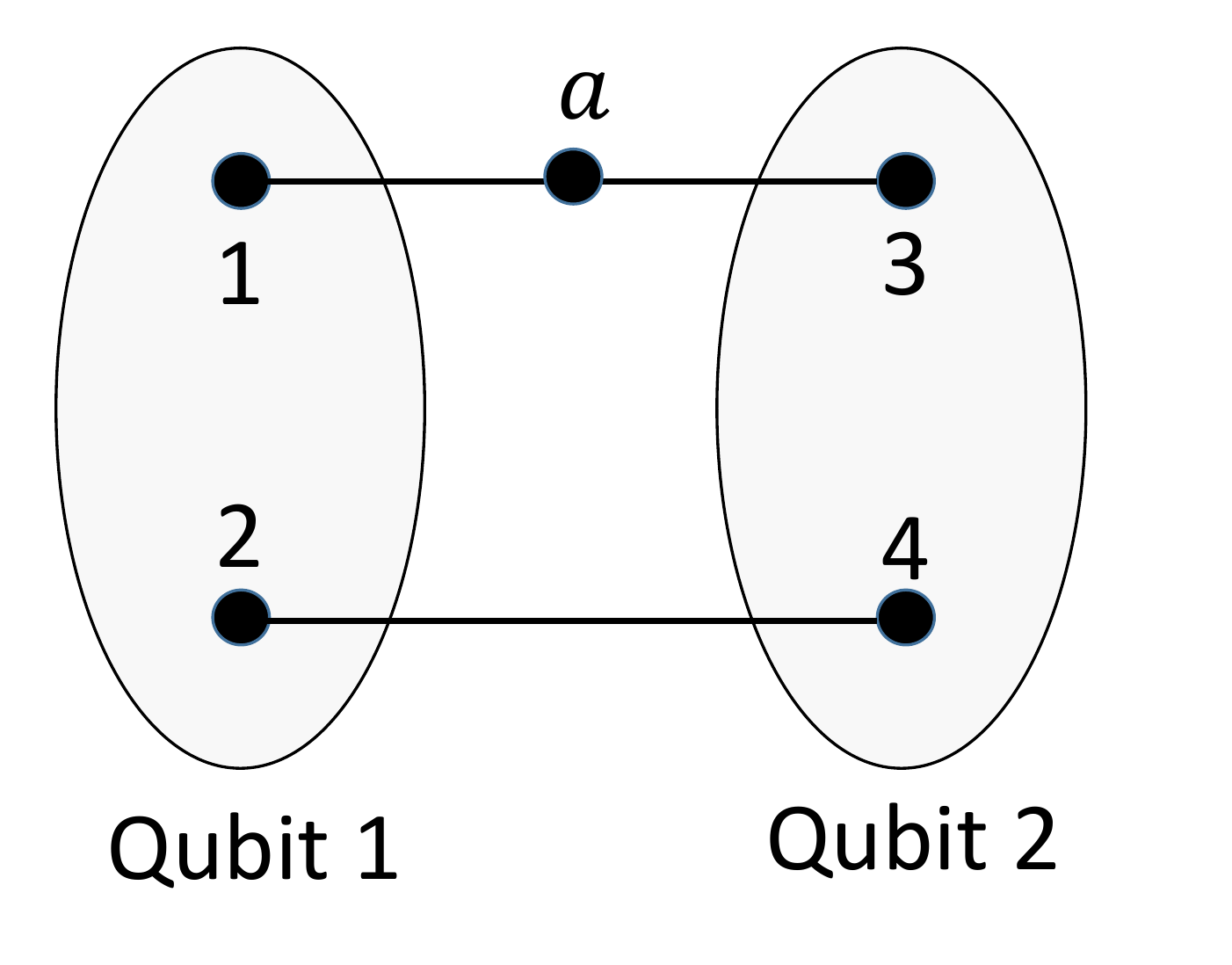}}
\caption{Graph $G=(U,E)$  of the $\HCB^*$ simulator  for  the target Hamiltonian
$-X\otimes X - Y\otimes Y$. The total number of particles is $m=2$.
   \label{fig:graph1}
}
\end{figure}

Our first reduction has a simulator Hamiltonian 
\begin{equation}
\label{stoq0}
\tilde{H}_\sm =\tilde{\Delta} \tilde{H}_0+\tilde{V}, \quad \quad  \tilde{H}_0=n_1n_2 + n_3 n_4,
\quad \quad  \tilde{V}\in \HCB^*(5,2,\tilde{J}).
\end{equation}
All above operators act on the Hilbert space $\calB_2(G)$.
The perturbation $\tilde{V}$ will be chosen at the next reduction. 
Let $\calH$ be the ground subspace of $\tilde{H}_0$. It is spanned by 
configurations of particles such that each qubit $\{1,2\}$
and $\{3,4\}$ contains at most one particle. 
Lemma~\ref{lemma:1st} implies that $\tilde{H}_\sm$ can simulate
the restriction
of any $\HCB^*$ Hamiltonian $\tilde{V}$ onto the subspace $\calH$.
Below we assume that our full Hilbert space is $\calH$.  
 
Our second reduction has a simulator Hamiltonian  
\begin{equation}
\label{stoq(3)}
H_\sm=\Delta H_0+V, \quad \quad H_0=n_a, \quad  \quad V=\Delta^{2/3} V_\main +\Delta^{1/3}\tilde{V}_\extra,
\end{equation}
where
\begin{equation}
\label{stoq(3)a}
V_{\main}= -p^{1/3} (W_{1,a}+W_{3,a} +n_a W_{2,4}),
\end{equation}
\begin{equation}
\label{stoq(3)b}
\tilde{V}_{\extra}= p^{2/3} (2n_1n_3+n_1n_4 + n_2n_3).
\end{equation}
Here all operators are restricted to the subspace $\calH$ with at most one particle per qubit.
Note that $H_\sm$ is an $\HCB^*$ Hamiltonian. 
The ground subspace of $H_0$ encodes two logical qubits according to Eq.~(\ref{stoq01}). 
Let us check that the perturbation satisfies all conditions of Lemma~\ref{lemma:3rd}.
Note that $(V_\main)_{--}=0$ since $V_\main$ can only move 
a particle from (to) the ancillary node $a$. The last term in $V_\main$ 
does not contribute to $(V_\main)_{--}$ since $n_a=0$ for any ground state of $H_0$.
Obviously, $\tilde{V}_\extra$ is block-diagonal. 

Let us check condition Eq.~(\ref{3rdA}) of Lemma~\ref{lemma:3rd}.
Informally, it says that the third-order hopping process generated by $V_\main$
must implement the logical hopping operator between the two logical qubits. 
 For example, suppose the initial
state is $|\overline{0,1}\rangle=|1,4\rangle$. Then $(V_\main)_{+-}$ moves
a particle from $1$ to $a$ by applying $W_{1,a}$, then $(V_\main)_{++}$ moves a particle from $4$ to $2$
by applying $n_aW_{2,4}$, 
and then $(V_\main)_{-+}$ moves a particle from $a$ to $3$ by applying $W_{3,a}$. This produces
the correct final state $|\overline{1,0}\rangle=|2,3\rangle$.
More formally, one can easily check that 
\begin{eqnarray}
\label{stoq5}
(V_\main)_{+-} |\overline{1,1}\rangle=(V_\main)_{+-} |2,4\rangle&=& 0,\nonumber \\
(V_\main)_{+-}  |\overline{0,0}\rangle=(V_\main)_{+-} |1,3\rangle&=& -p^{1/3} \left( |3,a\rangle + |1,a\rangle \right), \nonumber \\
(V_\main)_{+-}  |\overline{0,1}\rangle=(V_\main)_{+-}|1,4\rangle&=& -p^{1/3}  \, |4,a\rangle, \nonumber \\
(V_\main)_{+-}  |\overline{1,0}\rangle=(V_\main)_{+-}|2,3\rangle&=& -p^{1/3} \, |2,a\rangle. \nonumber 
\end{eqnarray}
From this one easily gets
\begin{equation}
\label{stoq7}
(V_\main)_{-+}H_0^{-1} (V_\main)_{++} H_0^{-1} (V_\main)_{+-}=-p|\overline{0,1}\rangle\langle \overline{1,0}| 
-p|\overline{1,0}\rangle\langle \overline{0,1}|=\overline{H}_\tgt.
\end{equation}
This proves condition Eq.~(\ref{3rdA}).
A similar calculation shows that 
\begin{equation}
\label{stoq8}
(V_\main)_{-+}H_0^{-1} (V_\main)_{+-} =p^{2/3} \left( 2|\overline{0,0}\rangle\langle \overline{0,0}|
 +|\overline{0,1}\rangle\langle \overline{0,1}| + |\overline{1,0}\rangle\langle \overline{1,0}|\right)
 =p^{2/3}(2n_1n_3+n_1n_4 + n_2n_3)
\end{equation}
which  proves condition Eq.~(\ref{3rdB}).

To compose the two reductions we extend $H_\sm$ defined in Eq.~(\ref{stoq(3)})
to the full Hilbert space $\calB_2(G)$ and substitute $\tilde{V}=H_\sm$ into Eq.~(\ref{stoq0}).
Combining Lemmas~\ref{lemma:sim3},\ref{lemma:1st},\ref{lemma:3rd} 
we conclude that $H_\tgt$ can be simulated with an error $(\eta,\epsilon)$ by 
$\tilde{H}_{\sm}\in \HCB^*(5,2,J')$ where $J'=\poly(p,\eta^{-1},\epsilon^{-1})$. 
The simulation uses the dual rail encoding $\calE$ defined in Eq.~(\ref{stoq01}).

Finally, consider the case $H_{\tgt}=-pX\otimes X+pY\otimes Y$, where $p>0$.
This Hamiltonian can be obtained from $H_{\tgt}=-p(X\otimes X+Y\otimes Y)$
by conjugating the second qubit with $X$. This is equivalent to exchanging
nodes $3$ and $4$ in the reduction described above. 
Hence we have constructed an $\HCB^*$ simulator for all elementary stoquastic interactions.

Suppose now that $H_{\tgt}\in \StoqLH(n,J)$ is a general $2$-local stoquastic Hamiltonian 
on $n$ qubits.
By Lemma~\ref{lemma:stoq},  
\begin{equation}
\label{stoq15}
H_{\tgt}=H_{\diag}+\sum_{\alpha=1}^m p_\alpha H_\alpha,
\end{equation}
where $H_{\diag}$ is a diagonal $2$-local Hamltonian
with terms proportional to $n_u$ and $n_u n_v$, where  $p_\alpha> 0$ are some coefficients,
and each term $H_\alpha$ is one of the four elementary stoquastic interactions 
applied to some pair of qubits. Obviously, the number of terms is $m\le \poly{(n)}$.
The corresponding $\HCB^*$ simulator $H_{\sm}$ will be defined on a graph $G=(U,E)$
with $2n+m'$ nodes, where $m'\le m$ is the number of interactions of type (3) or (4)
in $H_{\tgt}$, see Lemma~\ref{lemma:stoq}.
Without loss of generality, the first $m'$ interactions $H_\alpha$ are of type (3) or (4). 
The total number of particles in the $\HCB^*$ simulator will be $m=n$, so that
the Hilbert space of the simulator is $\calB_n(G)$. 
We shall need a composition of a first-order and a third-order reductions. 

Let us first define the set of nodes of $G$. 
Each qubit $i$ gives rise to a pair of nodes   $t(i)$ 
and $b(i)$  which
form the dual-rail representation of the qubit. 
We encode the basis states $|0\rangle$ and $|1\rangle$ by
a single particle located at the node $t(i)$ and $b(i)$ respectively. 
Each interaction $H_\alpha$ of type (3) or (4) gives rise to an extra node $a(\alpha)$
in $G$. 
This is the ancillary node used in the  graph shown on Fig.~\ref{fig:graph1}.

The set of edges of $G$ is defined as follows. For each interaction $H_\alpha$ of type (3)
coupling qubits $i,j$ 
we add an edge connecting nodes $t(i)$ and $a(\alpha)$, an edge
connecting nodes $a(\alpha)$ and $t(j)$, and an edge connecting nodes $b(i)$ and $b(j)$.
The last edge represents a controlled hopping with a control node $a(\alpha)$. 
For each interaction $H_\beta$ of type (4) coupling qubits $i,j$ we add 
an edge connecting nodes $t(i)$ and $a(\beta)$, an edge connecting
nodes $a(\beta)$ and $b(j)$, and an edge connecting nodes $b(i)$ and $t(j)$.
The last edge represents a controlled hopping with a control node $a(\beta)$.
The resulting subgraph of $G$ is 
shown on Fig.~\ref{fig:graph2}. This completes definition of the graph $G$ for the HCB simulator.

\begin{figure}[h]
\centerline{\includegraphics[height=4cm]{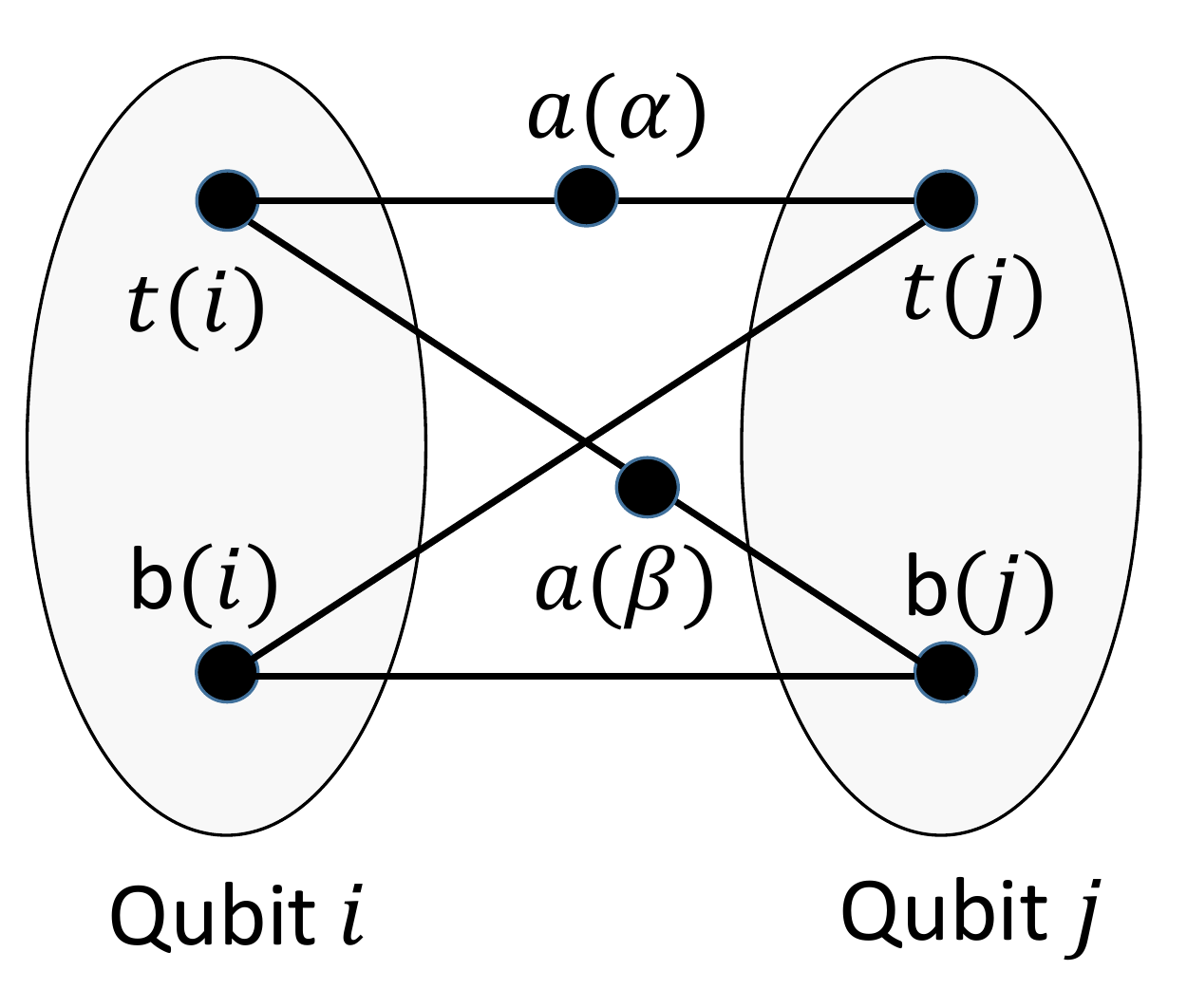}}
\caption{Construction of the graph $G$ for the HCB simulator. Each qubit $i$
is represented by a pair of nodes $t(i),b(i)$ using the dual-rail representation. 
The ancillary node $a(\alpha)$  controls hopping between $b(i)$ and $b(j)$.
The ancillary node $a(\beta)$ controls hopping between $b(i)$ and $t(j)$. 
\label{fig:graph2}
}
\end{figure}

Our first reduction has a simulator Hamiltonian 
\begin{equation}
\label{stoq16a}
\tilde{H}_\sm = \tilde{\Delta} \tilde{H}_0+\tilde{V}, \quad \tilde{H}_0=\sum_{i=1}^n n_{t(i)} n_{b(i)},
\quad \tilde{V}\in \HCB^*(2n+m',n,\tilde{J}).
\end{equation}
All above operators act on the Hilbert space $\calB_n(G)$.
The perturbation $\tilde{V}$ will be chosen at the next reduction. 
Let $\calH$ be the ground subspace of $\tilde{H}_0$. It is is spanned by
configurations of particles such that each qubit $\{t(i),b(i)\}$ contains at most one particle. 
Lemma~\ref{lemma:1st} implies that $\tilde{H}_\sm$ can  simulate the restriction
of any $\HCB^*$ Hamiltonian $\tilde{V}$ onto the subspace $\calH$.
Below we assume that our full Hilbert space is $\calH$. 

Our second reduction has a simulator Hamiltonian 
\begin{equation}
\label{stoq16}
H_\sm=\Delta H_0+V,\quad  \quad
H_0= \sum_{\alpha=1}^{m'}  n_{a(\alpha)},
\quad  \quad 
V=\Delta^{2/3} V_\main +\Delta^{1/3} \tilde{V}_\extra + V_\extra.
\end{equation}
Since the total number of particles is $n$ and each qubit
may contain at most one particle, the ground subspace of $H_0$ is spanned
by states with exactly one particle per qubit. It  encodes $n$ logical qubits
under the dual rail representation.  The perturbation operators are defined as the sums
of respective perturbation operators over all individual simulators. More formally,
\begin{equation}
\label{stoq17}
V_\main=\sum_{\alpha=1}^{m'} V_{\main}^\alpha, \quad \quad 
\tilde{V}_\extra=\sum_{\alpha=1}^{m'} \tilde{V}_{\extra}^\alpha,
\quad \quad V_\extra=H_\diag +\sum_{\alpha=m'+1}^m V^\alpha,
\end{equation}
where the perturbation operators carrying an index $\alpha$ 
are defined by Eqs.~(\ref{stoq(1)a},\ref{stoq(3)a},\ref{stoq(3)b}), depending
on the type of the interaction $H_\alpha$, 
with $p$ replaced by $p_\alpha$, with the nodes $\{1,2\}$ replaced by 
the nodes of the first logical qubit acted upon by $H_\alpha$, and with
the nodes $\{3,4\}$ replaced by the nodes of the second logical qubit acted upon by $H_\alpha$.

Let us check that the perturbation $V$ satisfies all conditions of Lemma~\ref{lemma:3rd}.
By definition,  $V_\extra$ and $\tilde{V}_\extra$ 
act trivially on the ancillary nodes $a(\alpha)$ and thus they are block diagonal.
 Any term in $V_\main$ moves a particle to (from) some
ancillary node $a(\alpha)$. Since all these nodes must be empty in the ground subspace of $H_0$,
one has $(V_\main)_{--}=0$. 

Let us now check condition Eq.~(\ref{3rdB}). We claim that 
\begin{equation}
\label{stoq18}
(V_\main)_{-+}H_0^{-1} (V_\main)_{+-} =\sum_{\alpha=1}^{m'} (V_\main^{\alpha})_{-+} H_0^{-1} (V_\main^\alpha)_{+-}.
\end{equation}
Indeed, any term in $(V_\main)_{+-}$ must move a particle
from some qubit node $u\in \{t(i),b(i)\}$ to some ancillary node $a(\alpha)$. 
In order to return the system to the ground subspace of $H_0$, the factor
$(V_\main)_{-+}$ must move the particle from $a(\alpha)$ to $u$. 
Thus $(V_\main^\alpha)_{-+} H_0^{-1} (V_\main^\beta)_{+-}=0$ for $\alpha\ne \beta$. This implies Eq.~(\ref{stoq18}).
Since $V_\main^\alpha$ and $\tilde{V}_\extra^\alpha$ satisfy condition Eq.~(\ref{3rdB})
for each individual simulator, Eq.~(\ref{stoq18}) implies that
$V_\main$ and $\tilde{V}_\extra$ also satisfy condition Eq.~(\ref{3rdB}).

Next let us check condition Eq.~(\ref{3rdA}) of Lemma~\ref{lemma:3rd}.
We claim that 
\begin{equation}
\label{stoq20}
(V_\main)_{-+}H_0^{-1} (V_\main)_{++} H_0^{-1} (V_\main)_{+-}
=\sum_{\alpha=1}^m (V_\main^\alpha)_{-+} H_0^{-1} (V_\main^\alpha)_{++} H_0^{-1} (V_\main^\alpha)_{+-}
\end{equation}
Indeed, suppose $(V_\main)_{+-}$ moves a particle from a qubit node $u\in \{t(i),b(i)\}$ to some 
ancillary node $a(\alpha)$. For concreteness, assume that $H_\alpha$ is 
an interaction of type~(3) coupling qubits $i,j$ and $u=t(i)$.
The factor $(V_\main)_{++}$ can either apply the controlled hopping term
proportional to $n_{a(\alpha)} W_{b(i),b(j)}$ 
or  move a particle from some other qubit node
$v\in \{t(i'),b(i')\}$ to some ancillary node $a(\beta)$ with $\alpha\ne \beta$.
In the latter case, however,
we create two excited ancillary nodes so that  $(V_\main)_{-+}$ will not be able to return the system  back to
the ground subspace of $H_0$. In the former case $(V_\main)_{-+}$ can return the system
to the ground subspace of $H_0$ only by moving the particle from $a(\alpha)$
to some node of qubit $i$ or $j$. This proves Eq.~(\ref{stoq20}).
Since $V_\main^\alpha$ satisfies condition Eq.~(\ref{3rdA})
for each individual simulator with $V_\extra^\alpha=0$, Eq.~(\ref{stoq20}) implies that
\begin{equation}
\label{stoq21}
(V_\main)_{-+}H_0^{-1} (V_\main)_{++} H_0^{-1} (V_\main)_{+-}=\sum_{\alpha=1}^{m'} p_\alpha \overline{H}_\alpha.
\end{equation}
Furthermore, one has $(V^\alpha)_{--}=p_\alpha \overline{H}_\alpha$ for interactions of type
(1) and (2), that is, for $m'<\alpha\le m$. Therefore
\begin{equation}
\label{stoq22}
(V_\extra)_{--}=\overline{H}_\diag + \sum_{\alpha=m'+1}^m p_\alpha \overline{H}_\alpha. 
\end{equation}
Combining this and Eq.~(\ref{stoq21}) one arrives at
\begin{equation}
\label{stoq23}
(V_\extra)_{--} + (V_\main)_{-+}H_0^{-1} (V_\main)_{++} H_0^{-1} (V_\main)_{+-} = \overline{H}_\tgt.
\end{equation}
Thus all conditions of Lemma~\ref{lemma:3rd} are satisfied. 

To compose the two reductions we extend $H_\sm$ defined in Eq.~(\ref{stoq16})
to the full Hilbert space $\calB_n(G)$ and substitute $\tilde{V}=H_\sm$ into Eq.~(\ref{stoq16a}).
Combining Lemmas~\ref{lemma:sim3},\ref{lemma:1st},\ref{lemma:3rd}
we conclude that any Hamiltonian
$H_\tgt\in \StoqLH(n,J)$ can be simulated with an error $(\eta,\epsilon)$
by the Hamiltonian $\tilde{H}_\sm\in \HCB^*(n',n,J')$ where
$n'=2n+m'=O(n^2)$ and $J'=\poly(n,J,\epsilon^{-1},\eta^{-1})$. 
The simulation uses the dual rail encoding $\calE$
that maps basis vectors to basis vectors.

Finally,  let us extend the above reduction  to $(2,k)$-local stoquastic Hamiltonians.
We have to modify the simulator model by adding multi-particle interactions
as described in Section~\ref{sect:multi}. More precisely,  
consider the target Hamiltonian defined in Eq.~(\ref{stoq15}) and 
suppose the term $H_\diag$ contains  $k$-qubit diagonal interactions
with strength at most $J$. 
Such Hamiltonian can be written as
\begin{equation}
\label{diag2k}
H_\diag=-\sum_{x\in \{0,1\}^k} \sum_{M\subseteq [n]}  p(x,M) |x\rangle\langle x|_M,
\end{equation}
where $p(x,M)$ are some real coefficients such that $|p(x,M)|\le \poly(J,n)$, 
and $|x\rangle\langle x|_M$ is the $k$-qubit projector  $|x\rangle\langle x|$ 
acting on a subset of qubits $M$. Performing an overall energy shift one can achieve
$p(x,M)\ge 0$ for all $x$ and $M$.  Consider any fixed  projector $|x\rangle\langle x|_M$. 
Recall that  the basis states  $|0\rangle$ and $|1\rangle$ of the $i$-th qubit
of the target model
are encoded by a particle located at the node
$t(i)$ and $b(i)$ respectively. It follows that 
the encoded version of a projector $|0\rangle\langle 0|_i$
can be written as $I-n_{b(i)}$. Likewise, the encoded version of a projector $|1\rangle\langle 1|_i$
can be written as $I-n_{t(i)}$. Thus the encoded version of
the projector  $|x\rangle\langle x|_M$ is
\[
\overline{|x\rangle\langle x|}_M=\prod_{i\in M\, : \, x(i)=0} (I-n_{b(i)}) \prod_{i\in M\, : \, x(i)=1} (I-n_{t(i)}).
\] 
Here $x(i)$ denotes the bit of the string $x$ associated with the $i$-th qubit. 
Let $S(x,M)\subseteq U$ be the union of all nodes $b(i)$ with $i\in M$ and $x(i)=0$
and all nodes $t(i)$ with $i\in M$ and $x(i)=1$. 
We conclude that the encoded version of $H_\diag$ is
\begin{equation}
\label{diag2k'}
\overline{H}_\diag = -\sum_{x\in \{0,1\}^k} \sum_{M\subseteq [n]}  p(x,M) \prod_{u\in S(x,M)} (I-n_u).
\end{equation}
As we have shown in Section~\ref{sect:multi} any such Hamiltonian can be
included into the range-$2$ HCB model by adding one extra second-order reduction.
Thus all our results obtained for $2$-local stoquastic Hamiltonians hold for
$(2,k)$-local stoquastic Hamiltonians.

\section{Proof of the main theorems}
\label{sect:proof}

Now we have all  ingredients needed for the proof of Theorems~\ref{thm:main},\ref{thm:QA}.
Consider a target Hamiltonian $H\in \StoqLH(n,J)$ or $H\in \StoqLH^*(n,J)$.
Let us prove that $H$ can be simulated with an arbitrarily small error $(\eta,\epsilon)$,
by a TIM Hamiltonian $H'\in \TIM(n',J')$ such that  $n'\le \poly(n)$ and 
$J'\le \poly(n,J,\epsilon^{-1},\eta^{-1})$. 
Here we use the definition of simulation given in Section~\ref{sect:sim}. 
The parameters $\eta,\epsilon$ will be specified later. 
Indeed, consider  the sequence of perturbative reductions constructed in 
Sections~\ref{sect:TIM3}-\ref{sect:stoqLH}.
It can be described by a sequence of Hamiltonians
$H_1,H_2,\ldots,H_R$ and encodings $\calE_1,\calE_2,\ldots,\calE_{R-1}$ 
such that  the Hamiltonian $H_t$ and the encoding $\calE_t$ simulate
$H_{t+1}$ with a small error $(\eta_t,\epsilon_t)$ for each $t=1,\ldots,R-1$.
Here $H_R=H$ is the desired target Hamiltonian and  $H_1=H'$ is a TIM
Hamiltonian. 
Choose  a simulation error $\eta_t=\eta/2R$ and
and $\epsilon_t=\epsilon/2R$  for each  each individual reduction.
Since $R=O(1)$, this implies $\eta_t=\Omega(\eta)$ and $\epsilon_t=\Omega(\epsilon)$. 
By construction, each Hamiltonian $H_t$  belongs to one of the classes defined in Table~\ref{tab:classes}
with some number of nodes (qubits)  $n_t$ and some interaction strength $J_t$.
In addition, each Hamiltonian $H_t$ (except for $H_R$) 
is a sum of a strong unperturbed part  $(H_t)_0$ with a spectral gap $\Delta_t\gg \|H_{t+1}\|$
and a weak perturbation.
We have shown that each reduction satisfies
\begin{equation}
\label{reduction3}
n_t\le \poly(n_{t+1}) \quad \mbox{and} \quad   J_t\le \poly(n_{t+1},J_{t+1},\epsilon^{-1},\eta^{-1})
\end{equation}
where $t=1,\ldots,R-1$. Using the initial conditions $n_R=n$, $J_R=J$, and
taking into account that $R=O(1)$, 
we conclude that   $n'\equiv n_1\le \poly(n)$ and $J'\equiv J_1\le \poly(n,J,\epsilon^{-1},\eta^{-1})$.

Let $\calE=\calE_1 \calE_2 \cdots \calE_{R-1}$ be the 
composition of all individual  encodings. 
Applying Lemma~\ref{lemma:sim3} one infers that the Hamiltonian
$H_1$ and the encoding $\calE$  simulate $H_R$ with
an error $(\tilde{\eta},\tilde{\epsilon})$, where 
\begin{equation}
\label{eq}
\tilde{\eta} \le \frac{\eta}2 + O(\epsilon) \max_t \Delta_t^{-1}
\quad \mbox{and} \quad
\tilde{\epsilon} \le \frac{\epsilon}2 + O(\epsilon) \max_t \frac{\|H_{t+1}\|}{\Delta_t}
\end{equation}
Increasing, if necessary, the spectral gaps $\Delta_t$ by a factor $\poly(\eta^{-1})$
one can achieve $\tilde{\eta}\le \eta$ and $\tilde{\epsilon}\le \epsilon$.

To prove Theorem~\ref{thm:main} we choose $\epsilon$
as the precision specified in the statement of the theorem. The parameter $\eta$
does not play any role here. We have to use all reductions described in Sections~\ref{sect:TIM3}-\ref{sect:stoqLH}
except for the one of 
Sections~\ref{sect:multi} (the latter generates multi-particle interactions that are only needed
for the proof of Theorem~\ref{thm:QA}).  Then $H_1$ is a TIM Hamiltonian with interactions
of degree-$3$. 
Lemma~\ref{lemma:sim1} implies that the $i$-th smallest eigenvalues of $H_1$ and $H_R$
differ at most by $\epsilon$ for all $i=1,\ldots,2^n$. This  proves Theorem~\ref{thm:main}
with $H'=H_1$.

To prove Theorem~\ref{thm:QA} we shall choose $\epsilon\le \delta/3$, where $\delta$ is the spectral gap 
of $H$.  We have to use all reductions described in Sections~\ref{sect:TIM2HCD}-\ref{sect:stoqLH}.
By construction, the encodings $\calE_t$  used in all these reductions map basis vectors to basis vectors.
In addition, one can efficiently compute the action of $\calE_t$ and $\calE_t^\dag$ on any basis vector. 
Thus the same properties hold for the full encoding $\calE=\calE_1 \calE_2 \cdots \calE_{R-1}$.
Lemma~\ref{lemma:sim1} implies that the Hamiltonian $H_1$
has a non-degenerate ground state  and a spectral gap at least $\delta-2\epsilon\ge \delta/3$.
 Let $|g\rangle$ and $|g'\rangle$ be the ground states
of $H$ and $H_1$.  By Lemma~\ref{lemma:sim2}, 
\[
\| \, |g'\rangle -\calE|g\rangle\,  \| \le \eta + C\delta^{-1} \epsilon
\]
for some constant coefficient $C$. Choosing $\epsilon=\min{ (\delta/3, \eta \delta C^{-1})}$ 
one can achieve   $\| |g'\rangle -\calE|g\rangle \|\le 2\eta$. 
Since $\eta$ can be arbitrarily small, this proves Theorem~\ref{thm:QA}
with $H'=H_1$. 

Let us remark that  the map $H\to H'$ in Theorem~\ref{thm:QA} can be made sufficiently smooth. 
More precisely, suppose $H$ smoothly depends on some parameter $\tau$ such that 
the $j$-th derivative of $H$ with respect to $\tau$ has norm at most $\poly(n)$ for any constant $j$. 
Then we claim that the $j$-th derivative of $H'$ with respect to $\tau$ has norm at most $\poly(n,\delta^{-1})$.
Indeed, suppose $H_\tgt=H_{t+1}$ and $H_\sm=H_t$ are the target and the simulator Hamiltonians 
used in some individual reduction and $(\eta_t,\epsilon_t)$ is the desired simulation error. 
For concreteness, consider the reduction of
Section~\ref{sect:stoqLH}. Note that the derivative of $H_\sm$ 
becomes  infinite if some of the coefficients $p_\alpha$ in Eq.~(\ref{stoq15})
becomes zero since $H_\sm$ contains terms  proportional to $p_\alpha^{1/3}$
and $p_\alpha^{2/3}$. 
 To avoid such singularities, let us choose a sufficiently small
cutoff value $p_{min}$ and replace  $p_\alpha$ by $\tilde{p}_\alpha=\sqrt{p_\alpha^2 +p_{min}^2}$
in Eq.~(\ref{stoq15})
(recall that the coefficients $p_\alpha$ must be non-negative).
This gives a new target Hamiltonian $\tilde{H}_\tgt$ such that 
$\|H_\tgt-\tilde{H}_\tgt\|\le p_{min} \poly(n)$.
We choose $p_{min}$ small enough so that $\|H_\tgt-\tilde{H}_\tgt\|\ll \epsilon_t$.
Let $\tilde{H}_\sm$ be the simulator Hamiltonian constructed for $\tilde{H}_\tgt$.
Then $\tilde{H}_\sm$ simulates $H_\tgt$ with the error approximately $(\eta_t,\epsilon_t)$ and 
the $j$-th derivative of $\tilde{H}_\sm$ has norm at most 
$\poly(n,t_{min}^{-1})=\poly(n,\epsilon^{-1}_t)=\poly(n,\delta^{-1})$. 
By introducing a similar cutoff in all remaining reductions one can easily check that 
the $j$-th derivative of $H'$ with respect to $\tau$ has norm at most 
$\poly(n,\delta^{-1})$.
It is known that an adiabatic path with the minimum spectral gap
$\delta$ such that the $j$-th derivative has norm at most $C_j$ can be
traversed in time $T=O(C_1\delta^{-2}+C_2\delta^{-2} + C_1^2\delta^{-3})$,
see~\cite{JRS07}. This implies Corollary~\ref{cor:4}.

Finally, let us remark that Theorem~\ref{thm:QA} can be extended to TIM
Hamiltonians with interactions of degree-$3$, although the corresponding encoding 
$\calE$ would no longer map basis vectors to basis vectors. 
Indeed, let us modify the above proof of Theorem~\ref{thm:QA} by 
including the reduction of Section~\ref{sect:TIM3}.
Then the final TIM Hamiltonian $H'$ has  interactions of degree-$3$. 
Let $\calE_1$ be the encoding used in the  reduction of  Section~\ref{sect:TIM3}.
Recall that $\calE_1$ encodes each qubit $u$ of the target model into 
a one-dimensional chain $L_u$  with a Hamiltonian
$H_\ring=-\sum_{j\in \ZZ_m } g Z_j Z_{j+1} + X_j$,
where $m\le \poly(n)$ and $g\approx 1$, see Section~\ref{sect:TIM3}.
Basis states of the logical qubit 
are $|\overline{0}\rangle\sim |\psi_0\rangle + |\psi_1\rangle$ and $|\overline{1}\rangle\sim |\psi_0\rangle - |\psi_1\rangle$,
where $\psi_0$ and $\psi_1$ are the ground states of $H_\ring$
satisfying $X^{\otimes m} \psi_{0,1}=\pm \psi_{0,1}$. 
Accordingly, the full encoding $\calE=\calE_1 \calE_2 \cdots \calE_{R-1}$
maps any basis vector   to a tensor product of the states $|\overline{0}\rangle$ and $|\overline{1}\rangle$.
Let us argue that the logical qubits can be efficiently initialized and measured.
Choose any physical qubit $i\in L_u$. 
Using Eqs.~(\ref{Z--},\ref{xi},\ref{xi_bound}) one gets $\langle \overline{0}|Z_i|\overline{0}\rangle =\xi$ and
$\langle \overline{1}|Z_i|\overline{1}\rangle =-\xi$, where $\xi\ge \poly(n^{-1})$. 
Thus one can measure the logical qubit $u$ in the $Z$-basis by measuring 
any physical qubit of $L_u$ in the $Z$-basis.
However, the measurement has to be repeated $\poly(n)$ times to get a reliable statistics. 
One can measure the logical qubit in the $X$-basis in a single shot by measuring
every qubit of $L_u$ in the $X$-basis. Computing the product of the measured outcomes
gives the eigenvalue of $X^{\otimes m}$  which differentiates between $\psi_0$ and $\psi_1$. 
Finally, the state $|\overline{+}\rangle=|\psi_0\rangle$ can be prepared by the adiabatic evolution 
starting from the product state $|+^{\otimes m}\rangle$ and adiabatically turning on the parameter
$g$ in the Hamiltonian $H_\ring$. It is well-known that the minimum spectral gap of $H_\ring$ is
$\Omega(m^{-1})$, so that the initialization can be done in time $\poly(n)$. 
The logical state $|\overline{0}\rangle$ can be obtained from $|\overline{+}\rangle$
by adiabatically changing the logical Hamiltonian from $-\overline{X}$ to $-\overline{Z}$.

\appendix
\section{Bounds on the energy splitting and matrix elements for the Ising chain} 
\label{sect:app}

In this section we prove Eqs.~(\ref{delta_bound},\ref{xi_bound}).

Let us first prove Eq.~(\ref{delta_bound}). Choose any $g^{-1}<R<1$ and consider contours
\[
C=\{z\in \CC^2 : \, |z|=R\} \quad \mbox{and} \quad C^{-1}=\{z\in \CC^2 : \, |z|=R^{-1}\}
\] 
We orient $C$ and $C^{-1}$ clockwise and counter-clockwise respectively. 
Denote $E_+\equiv E_0$ and $E_-\equiv E_1$, see Fact~\ref{fact:TIM1}.
Using Eq.~(\ref{E012}) one can easily check that
\begin{equation}
\label{int1}
E_\pm=\frac{m}{2\pi i }\oint_{C\cup C^{-1}} \frac{dz \sqrt{g^2 +1 -g(z+z^{-1})}}{z(z^m\mp 1)}.
\end{equation}
Since the contours $C$ and $C^{-1}$ can be mapped to each other via a change
of variable $z\to z^{-1}$, one gets 
\begin{equation}
\label{int2}
E_\pm =-\frac{m}{2\pi i} \oint_C \frac{dz (1\pm z^m)\sqrt{g^2 +1 -g(z+z^{-1})} }{z(z^m\mp 1)}.
\end{equation}
Therefore
\begin{equation}
\label{int3}
\delta=E_--E_+=-\frac{2m}{\pi i} \oint_C \frac{dz z^{m-1} \sqrt{g^2 +1 -g(z+z^{-1})} }{1-z^{2m}}.
\end{equation}
The function $\sqrt{g^2 +1 -g(z+z^{-1})}$ is analytic in the complex plane
with cuts along the intervals $[0,g^{-1}]$ and $[g,\infty]$. Deforming the contour $C$ such that
it goes from $0$ to $g^{-1}$ in the upper half-plane and then returns
to $0$ in the lower half-plane one gets 
\begin{equation}
\label{Eint}
\delta=\frac{4m\sqrt{g}}{\pi}\int_0^{g^{-1}} \frac{x^{m-1} \sqrt{x+x^{-1}-g-g^{-1}}dx}{1-x^{2m}}.
\end{equation}
Using a bound
\begin{equation}
\label{eq1}
x+x^{-1}-g-g^{-1} =x^{-1}(g-x)(g^{-1}-x)\ge (g-g^{-1})x^{-1} (g^{-1}-x)
\end{equation}
we obtain
\begin{equation}
\label{eq2}
\delta\ge \frac{4m\sqrt{g^2-1}}{\pi} \int_0^{g^{-1}} x^{m-3/2}\sqrt{g^{-1}-x}dx.
\end{equation}
Making a change of variables $x=g^{-1}y$ one gets
\begin{equation}
\label{eq3}
\delta\ge \frac{4mg^{-m} \sqrt{g^2-1}}{\pi}\int_0^1 dy y^{m-3/2}\sqrt{1-y}
\ge  \Omega(1) g^{-m} m^{-1} \sqrt{g^2-1}.
\end{equation}
Here we noted that the integral over $y$ is equal to the beta function  $B(m-1/2,3/2)=\Omega(m^{-2})$.
Finally, since $\sqrt{g^2-1}=\Omega(m^{-1/2})$, one gets $\delta\ge \Omega(m^{-c-3/2})$.
To get an upper bound in Eq.~(\ref{delta_bound}) we note that 
\[
x^{-1}(g-x)(g^{-1}-x)\le g x^{-1} (g^{-1}-x)\quad  \mbox{and} \quad (1-x^{2m})^{-1} \le (1-g^{-2m})^{-1} \le O(1).
\]
Performing the same change of variable as above one gets
\begin{equation}
\label{eq4}
\delta \le O(mg^{-m+1})\int_0^1 dy y^{m-3/2} \sqrt{1-y} = O(mg^{-m+1})B(m-1/2,3/2) =O(m^{-c-1/2}).
\end{equation}

Let us now prove Eq.~(\ref{xi_bound}).  We shall use the notations of Fact~\ref{fact:TIM2}.
In the limit $g\to \infty$ the Hamiltonian Eq.~(\ref{Hring}) has ground states
$|0^{\otimes m}\rangle$ and $|1^{\otimes m}\rangle$, that is,
\[
|\psi_0\rangle=\frac1{\sqrt{2}}(|0^{\otimes m}\rangle +  |1^{\otimes m}\rangle)
\quad \mbox{and} \quad
|\psi_1\rangle=\frac1{\sqrt{2}}(|0^{\otimes m}\rangle -  |1^{\otimes m}\rangle).
\]
Thus $\xi=\langle \psi_1|Z_j|\psi_0\rangle = 1$.
Since $\xi$ is a real continuous function of $g$, it suffices to show that $|\xi|\ge (1-g^{-2})^{1/8}$
for all $g>1$. 
Let us write $|\xi|=(1-g^{-2})^{1/8} \eta$. Then we have to prove that $\eta\ge 1$.
Since we already know that $\eta=1$ in the limit $g\to \infty$,
it suffices to show that $\eta$ is a monotone decreasing function of $g$ for all $g>1$.
Below we prove that
\begin{equation}
\label{bound2}
\eta^{-1} \frac{\partial \eta}{\partial g} \le 0 \quad \mbox{for all $g>1$}.
\end{equation}
Computing the derivative over $g$ one gets
\[
\dot{\epsilon_p}=\frac12\left( \frac{g}{\epsilon_p} + \frac{\epsilon_p}{g} - \frac1{g\epsilon_p}\right)
\]
and thus
\begin{equation}
\label{der1}
\frac{\partial }{\partial g} \sum_p \sum_q \log{(\epsilon_p+\epsilon_q)}=
\frac{m^2}{2g} + \frac12 (g-g^{-1}) \left( \sum_p \epsilon_p^{-1}\right)
 \left( \sum_q \epsilon_q^{-1}\right).
\end{equation}
Here the sums over $p$ and $q$ can range over either $\ZZ_m$ or $\ZZ_m+1/2$.
Using Eq.~(\ref{der1}) one gets
\begin{equation}
\label{der2}
\eta^{-1} \frac{\partial \eta}{\partial g}=\frac1{16} (g^{-1}-g)(Z_+-Z_-)^2,
\end{equation}
where 
\begin{equation}
\label{Z0Z1}
Z_+\equiv \sum_{p\in \ZZ_m} \epsilon_p^{-1} \quad \mbox{and} \quad
Z_-\equiv \sum_{q\in \ZZ_m+1/2} \epsilon_q^{-1}
\end{equation}
This implies Eq.~(\ref{bound2}) and proves that $\eta\ge 1$ for all $g>1$.

\vspace{1cm}

{\bf Acknowledgments --}
The authors would like to thank Barbara Terhal for helpful discussions 
on perturbative reductions. S.B.  acknowledges NSF Grant CCF-1110941. 

%\bibliographystyle{unsrt}
%\bibliography{mybib}

\end{document}